\documentclass{fundam}
\pdfoutput=1
\usepackage{amssymb}
\usepackage{extarrows}
\usepackage{tikz}
\usetikzlibrary{decorations,decorations.markings}
\usetikzlibrary{arrows}
\usetikzlibrary{shapes}
\usetikzlibrary{plotmarks}
\usepackage{centernot}

\tikzset{
    pos plot label/.style 2 args={
        postaction={
            decorate,
            decoration={
                markings,
                mark=at position #1 with \node#2;
            }
        }
    }
}
\makeatletter
\tikzset{
    scale plot marks/.is choice,
    scale plot marks/false/.code={
        \def\pgfuseplotmark##1{\pgftransformresetnontranslations\csname pgf@plot@mark@##1\endcsname}
    },
    scale plot marks/true/.style={},
    scale plot marks/.default=true
}
\makeatother

\newcommand{\es}{\emptyset}

\newcommand{\dt}{\bullet}
\newcommand{\leer}{\varepsilon}
\newcommand{\nsymbol}{\mathbb{N}}
\newcommand{\zsymbol}{\mathbb{Z}}

\newcommand{\goesto}{\rightarrow}
\renewcommand{\goesto}[1]{\stackrel{#1}{\longrightarrow}}

\newcommand{\impl}{\Rightarrow}
\newcommand{\Parikh}{\Psi}

\newcommand{\EndSy}{\hfill\protect\makebox[1.0em][c]{
\protect\setlength{\unitlength}{0.2em}
\protect\begin{picture}(3,3)(0,0)
        \begin{thinlines}
\protect\put(0,0){\line(1,0){3}}
\protect\put(0,0){\line(0,1){3}}
\protect\put(0,3){\line(1,0){3}}
\protect\put(3,0){\line(0,1){3}}
        \end{thinlines}
\protect\end{picture}
\protect\setlength{\unitlength}{1mm}
}}

\newif\ifcomment
\commenttrue
\definecolor{gris}{gray}{0.3}

\newcommand{\BX}[1]{{\unskip\nobreak\hfil\penalty50
                    \hskip2em\hbox{}\hfil
\EndSy\/ {{\rm #1}}
                    \parfillskip=0pt \finalhyphendemerits=0 \par
                   }}

\newcommand{\DEF}[2]{\goodbreak\begin{definition}
                     \label{#1}\begin{rm}{\sc #2}

                    }
\newcommand{\ENDDEF}[1]{\BX{\ref{#1}}
                        \end{rm}\end{definition}
                       }
\newcommand{\ENXDEF}{\end{rm}\end{definition}}
\newcommand{\KRYPT}[2]{\goodbreak\begin{kryptosystem}
                     \label{#1}\begin{rm}{\sc #2}

                    }
\newcommand{\ENDKRYPT}[1]{\BX{\ref{#1}}
                        \end{rm}\end{kryptosystem}
                       }
\newcommand{\KOR}[2]{\goodbreak\begin{corollary}
                     \label{#1}{\sc #2}

                    }
\newcommand{\ENDKOR}[1]{\BX{\ref{#1}}
\end{corollary}
                       }
\newcommand{\ENXKOR}{
\end{corollary}}
\newcommand{\PROP}[2]{\goodbreak\begin{proposition}
                      \label{#1}{\sc #2}

                     }
\newcommand{\ENDPROP}{
\end{proposition}} 
\newcommand{\ENXPROP}[1]{\BX{\ref{#1}}
\end{proposition}} 
\newcommand{\THEO}[2]{\goodbreak\begin{theorem}
                     \label{#1}{\sc #2}

                    }
\newcommand{\ENDTHEO}{
\end{theorem}}
\newcommand{\SATZ}[2]{\goodbreak\begin{theorem}
                     \label{#1}{\sc #2}

                    }
\newcommand{\ENDSATZ}{
\end{theorem}}
\newcommand{\ENXSATZ}[1]{\BX{\ref{#1}}
\end{theorem}}
\newcommand{\LEM}[2]{\goodbreak\begin{lemma}
                     \label{#1}{\sc #2}

                    }
\newcommand{\ENDLEM}{
\end{lemma}}
\newcommand{\ENXLEM}[1]{\BX{\ref{#1}}
\end{lemma}}

\newcommand{\NOT}[2]{\goodbreak\begin{notation}
                     \label{#1}\begin{rm}{\sc #2}

                    }
\newcommand{\ENDNOT}[1]{\BX{\ref{#1}}
                        \end{rm}\end{notation}
                       }
\newcommand{\ENXNOT}{\end{rm}\end{notation}}
\newcommand{\REM}[2]{\goodbreak\begin{remark}
                     \label{#1}\begin{rm}{\sc #2}
                    }
\newcommand{\ENDREM}[1]{\BX{\ref{#1}}
                        \end{rm}\end{remark}
                       }
\newcommand{\ENXREM}{\end{rm}\end{remark}}
\newcommand{\BSP}[2]{\goodbreak\begin{beispiel}
                     \label{#1}\begin{rm}{\sc #2}

                    }
\newcommand{\ENDBSP}[1]{\BX{\ref{#1}}
                        \end{rm}\end{bespiel}
                       }

\newcommand{\is}{\iota}

\newcommand{\drop}[1]{}
\newcounter{exampleno}

\newcommand{\lts}{{\em LTS}}
\newcommand{\TS}{\mathit{TS}}
\newcommand{\PN}{\mathit{PN}}

\newcommand{\adj}{\mathit{adj}}
\newcommand{\art}{\mathit{articul}}
\newcounter{exampleTScounter}
\newcommand{\TSref}[1]{\TS_{\ref{#1}}}

\newcommand{\PNref}[1]{\PN_{\ref{#1}}}

\newcommand{\bk}{\mathbb{B}}\newcommand{\fw}{\mathbb{F}}

\usepackage{hyperref}

\begin{document}

\setcounter{page}{1}
\publyear{2021}
\papernumber{2080}
\volume{183}
\issue{1-2}

 \finalVersionForARXIV

\title{Articulations and Products of Transition Systems\\ and their Applications to Petri Net Synthesis}

\author{Raymond Devillers\thanks{Address  for correspondence: D\'epartement d'Informatique,
                   Universit\'e Libre de Bruxelles, Boulevard du Triomphe, C.P. 212, B-1050 Bruxelles, Belgium}
 \\
D\'epartement d'Informatique \\
Universit\'e Libre de Bruxelles \\
Boulevard du Triomphe, C.P. 212, B-1050 Bruxelles, Belgium\\
rdevil@ulb.ac.be
}

\maketitle

\runninghead{R. Devillers}{Articulations and Products of Transition Systems and their Applications to Petri Net Synthesis}

\begin{abstract}
In order to speed up the synthesis of Petri nets from labelled transition systems, a divide and conquer strategy consists
in defining  decompositions of labelled transition systems, such that each component is synthesisable iff so is the original system.
Then corresponding Petri Net composition operators are searched to combine the solutions of the various components
 into a solution of the original system.
The paper presents two such techniques, which may be combined: products and articulations.
They may also be used to structure transition systems,
and to analyse the performance of synthesis techniques when applied to such structures.
\end{abstract}


\section{Introduction}\label{intro.sec}

Instead of analysing a given system to check if it satisfies a set of desired properties,
the synthesis approach tries to build a system ``correct by construction'' directly from those properties.
In particular, more or less efficient algorithms have been developed to build a bounded Petri net
(possibly of some subclass, called a PN-solution)
with a reachability graph isomorphic to (or close to) a given finite labelled transition system
\cite{bbd,besdev-lata,EB-US-15,BDS-acta17,US-Lata18}.

The synthesis problem is usually polynomial in terms of the size of the LTS, with a degree between 2 and 7
depending on the subclass of Petri nets one searches for \cite{bbd,BBD95,BDS-acta17,besdev-lata},
but can also be NP-complete~\cite{BBD97}.
Hence the interest to apply a ``divide and conquer'' synthesis strategy when possible.
The general idea is to decompose the given LTS into components, to synthesise each component separately
and then to recombine the results in such a way to obtain a solution to the global problem.
We thus have to find a pair of operators, one acting on transition systems and a corresponding one acting on Petri nets,
such that a composed LTS has a PN-solution if and only if so are its components,
and a possible solution is given by the application of the Petri net operator applied to solutions of the components.
It is also necessary to be able to  rapidly decompose a given LTS, or to state it is not possible.
This is summarised in Figure~\ref{divide.fig}.

\begin{figure}[hbt]
\vspace*{1mm}
\begin{center}
{$\TS=\TS_1 \;\mathbf{op_{TS}}\; \TS_2$ \\[1em]}
{$\TS$ PN-solvable $\iff$ $\TS_1$ and $\TS_2$ PN-solvable \\[1em]}
{$sol(\TS)= sol(\TS_1) \;\mathbf{op_{PN}}\; sol(\TS_2)$ \\[1em]}
{$\TS$ $\impl$ discovering of $\TS_1$ and $\TS_2$ \\[1em]}
\end{center}\vspace*{-4mm}
\caption{Divide and conquer strategy for synthesis}
\label{divide.fig}
\end{figure}

We shall here present two such strategies, which have been introduced recently and may be combined efficiently.
The first one is given by the  disjoint products of LTS,
which correspond to disjoint sums of Petri nets~\cite{devacta17,fact18},
and the second one is given by articulations on a state of a transition system
and on a non-dominated reachable marking of a Petri net~\cite{RD-articul-PN},
which may occur in various forms (choice, sequence, loop).

\medskip
An example of their mixed usage is illustrated in Figure~\ref{seq.fig}, where articulations are instantiated in their sequence form, and one of the components has a product form which was not apparent initially.
Not only this allows to simplify the Petri net synthesis, if needed, but also this allows to exhibit an interesting internal structure
for complex systems which could otherwise be considered as ``spaghetti-like''.

\begin{figure}[hbt]
\vspace*{-2mm}
\begin{center}
\begin{tikzpicture}[scale=0.7]
\node[]at(1,-2){$\TS(x)$};
\node[circle,fill=black!100,inner sep=0.07cm](s0)at(0,0)[label=left:$\is$]{};
\node[circle,fill=black!100,inner sep=0.05cm](s1)at(2,0)[label=right:$f$]{};
\draw[-latex](s0)--node[auto,swap,inner sep=1.5pt,pos=0.5]{$x$}(s1);
\end{tikzpicture}\hspace{0.4cm}
\begin{tikzpicture}[scale=0.7]
\node[]at(4,-2){$\TS=\TS(start);(\TS(a)\otimes\TS(b));\TS(end)$};
\node[circle,fill=black!100,inner sep=0.07cm](s0)at(0,0)[label=left:$\is$]{};
\node[circle,fill=black!100,inner sep=0.05cm](s1)at(2,0)[label=below:$s_1$]{};
\draw[-latex](s0)--node[auto,swap,inner sep=1.5pt,pos=0.5]{$start$}(s1);
\node[circle,fill=black!100,inner sep=0.05cm](s2)at(4,-1)[label=below:$s_2$]{};
\node[circle,fill=black!100,inner sep=0.05cm](s3)at(4,1)[label=above:$s_3$]{};
\node[circle,fill=black!100,inner sep=0.05cm](s4)at(6,0)[label=below:$s_4$]{};
\draw[-latex](s1)--node[auto,swap,inner sep=1.5pt,pos=0.5]{$a$}(s2);
\draw[-latex](s1)--node[auto,swap,inner sep=1.5pt,pos=0.5]{$b$}(s3);
\draw[-latex](s2)--node[auto,swap,inner sep=1.5pt,pos=0.5]{$b$}(s4);
\draw[-latex](s3)--node[auto,swap,inner sep=1.5pt,pos=0.5]{$a$}(s4);
\node[circle,fill=black!100,inner sep=0.05cm](s5)at(8,0)[label=right:$f$]{};
\draw[-latex](s4)--node[auto,swap,inner sep=1.5pt,pos=0.5]{$end$}(s5);
\end{tikzpicture}
\end{center}\vspace*{-7mm}
\caption{Combination of sequence operators with a product.}
\label{seq.fig}
\end{figure}
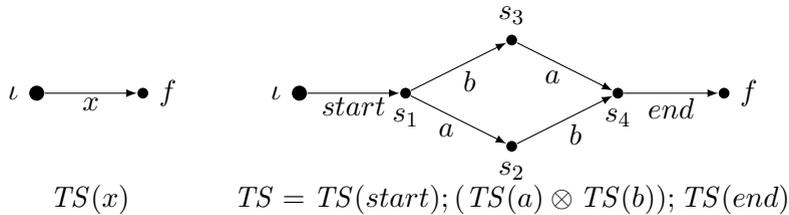

The structure of the paper is as follows.
First, we shall recall the bases of labelled transition systems and Petri nets.
Then, products of transition systems and sums of Petri nets are examined,
followed by articulations around states and markings, as well as the combination of both techniques.
Performance issues are detailed in section~6 and, as usual, the last section concludes.
With respect to previous papers on the subject, sections 5 and 6 are new, as well as some results and proofs (for instance Propositions~\ref{seqdiam.prop}, \ref{adeq.prop}, \ref{constr.prop} and Theorem~\ref{PNsynt.thm});
some small improvements are also scattered all over the rest of the presentation.

\section{Labelled transition systems and Petri nets}
\label{lts.sct}

A classic way for representing the possible (sequential) evolutions of a dynamic system is
through a (labelled) transition system~\cite{arnold94}.

\begin{definition}\label{LTS.def}{\sc Labelled Transition Systems}\\
A {\it labelled transition system} (LTS for short)
with initial state is a tuple $\TS=(S,\to,T,\is)$
with node (or state) set $S$, edge label set $T$, edges $\mathord{\to}\subseteq(S\times T\times S)$,
and an initial state $\is\in S$. We shall denote $s[t\rangle$ for $t\in T$ if there is an arc labelled $t$ from $s$,
$[t\rangle s$ if there is an arc labelled $t$ to $s$,
and $s[\alpha\rangle s'$ if there is a path labelled $\alpha\in T^*$ from $s$ to $s'$.
Such a path will also be called an evolution of the LTS (from $s$ to $s'$); $s'$ is then said reachable from $s$
and the set of states reachable from $s$ is denoted $[s\rangle$.

\medskip
In some proofs, we shall need the following extension of the reachability notion:
for each label $t\in T$ we shall denote by $-{t}$
the corresponding {\it reverse label},
i.e., $s[-{t}\rangle s'$ if $s'[t\rangle s$  (this may also be denoted $s\langle t]s'$).
We shall assume that $-{T}=\{-{t}\mid t\in T\}$ is disjoint from $T$,
and that $-{-{t}}=t$; then $\pm T=T\cup -T$ is the set of all forward and reverse labels.
The general paths $s[\alpha\rangle s'$ for $\alpha\in(\pm T)^*$ are then defined like the forward ones,
and we shall denote by $\langle s\rangle$ the set of states reachable from $s$ through a general path.
If $T'\subseteq T$, we shall denote by $\langle \is\rangle^{T'}$ the set of states reachable from $\is$
with general paths only using labels from $T'$: $\langle\is\rangle^{T'}=\{s'\in S\;|\is[\sigma\rangle s'\mbox{ for some }
\sigma\in(\pm T')^*\}$, as well as the restriction of $\TS$ to this set (the context will indicate which one is used);
similarly, $[\is\rangle^{T'}$ will be the same but with directed paths only.

If $\sigma\in(\pm T)^*$, we shall denote by $\Parikh(\sigma)$ the $T$-indexed vector in $\zsymbol^T$ such that $\forall t\in T:
\Parikh(\sigma)(t)=$ the number of occurrences of $t$ in $\sigma$ minus the  number of occurrences of $-t$ in $\sigma$, i.e.,
the generalised Parikh vector of $\sigma$.
We shall also denote by $s[-\sigma\rangle s'$ the path $s'[\sigma\rangle s$ ran backwards.

In the following we shall only consider finite transition systems, i.e., such that $S$ and $T$ (hence also $\to$) are finite.

Two transition systems $\TS_1=(S_1,\to_1,T,\is_{1})$
and $\TS_2=(S_2,\to_2,T,\is_{2})$  with the same label set $T$
are \text{(state-)isomorphic}, denoted $\TS_1\equiv_T\TS_2$
(or simply $\TS_1\equiv\TS_2$ if $T$ is clear from the context),
if there is a bijection  $\zeta\colon S_1\to S_2$ with $\zeta(\is_{1})=\is_{2}$ and
$(s,t,s')\in\mathord{\to_1}\Leftrightarrow(\zeta(s),t,\zeta(s'))\in\mathord{\to_2}$, for all $s,s'\in S_1$ and $t\in T$.
We shall usually consider LTSs up to isomorphism.

A transition system $\TS$ is said totally reachable if each state is reachable from the initial one: $[\is\rangle=S$.

It is reversible if $\forall s\in [\is\rangle: \is\in[s\rangle$, i.e.,
it is always possible to return to the initial state.\\
It is deterministic if $\forall s,s',s''\in S\;\forall t\in T: s[t\rangle s' \land s[t\rangle s''\impl s'=s''$
and $s'[t\rangle s\land s''[t\rangle s \impl s'=s''$.\\
It is weakly periodic if, for every $\alpha\in T^*$ and
 infinite path $s_1[\alpha\rangle s_2[\alpha\rangle s_3 \cdots$,
either for every $i,j\in\mathbb{N}\colon s_i=s_j$
or for every $i,j\in\mathbb{N}\colon i\neq j \impl s_i\neq s_j$ (the second case is of course excluded for finite transition systems).\QED
\end{definition}

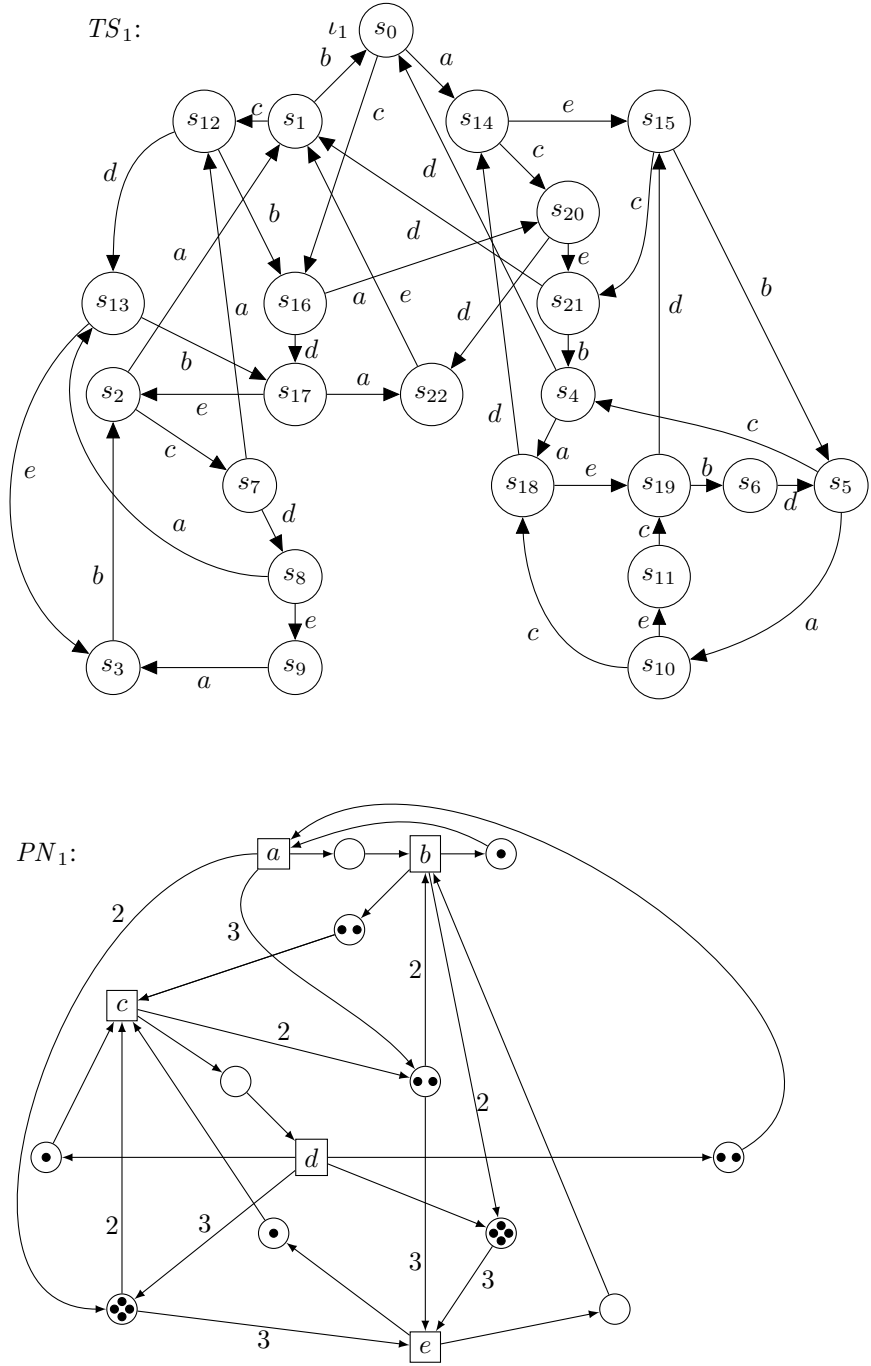
\begin{figure}[htbp]

 A  non-trivial LTS is illustrated on of Figure~\ref{aug21-synet-aut.fig}.

\begin{center}
\vspace*{3mm}
{\small
\begin{tikzpicture}[scale=1.2]
\refstepcounter{exampleTScounter}\label{ts21}
\draw(-3,10)node{$\TSref{ts21}$:};
\node[circle,draw,minimum size=0.5cm](s0)at(0,10)[label=left:$\is_{\ref{ts21}}$]{$s_0$};
\node[circle,draw,minimum size=0.5cm](s1)at(-1,9){$s_{1}$};
\node[circle,draw,minimum size=0.5cm](s2)at(-3,6){$s_{2}$};
\node[circle,draw,minimum size=0.5cm](s3)at(-3,3){$s_{3}$};
\node[circle,draw,minimum size=0.5cm](s4)at(2,6){$s_{4}$};
\node[circle,draw,minimum size=0.5cm](s5)at(5,5){$s_{5}$};
\node[circle,draw,minimum size=0.5cm](s6)at(4,5){$s_{6}$};
\node[circle,draw,minimum size=0.5cm](s7)at(-1.5,5){$s_{7}$};
\node[circle,draw,minimum size=0.5cm](s8)at(-1,4){$s_{8}$};
\node[circle,draw,minimum size=0.5cm](s9)at(-1,3){$s_{9}$};
\node[circle,draw,minimum size=0.5cm](s10)at(3,3){$s_{10}$};
\node[circle,draw,minimum size=0.5cm](s11)at(3,4){$s_{11}$};
\node[circle,draw,minimum size=0.5cm](s12)at(-2,9){$s_{12}$};
\node[circle,draw,minimum size=0.5cm](s13)at(-3,7){$s_{13}$};
\node[circle,draw,minimum size=0.5cm](s14)at(1,9){$s_{14}$};
\node[circle,draw,minimum size=0.5cm](s15)at(3,9){$s_{15}$};
\node[circle,draw,minimum size=0.5cm](s16)at(-1,7){$s_{16}$};
\node[circle,draw,minimum size=0.5cm](s17)at(-1,6){$s_{17}$};
\node[circle,draw,minimum size=0.5cm](s18)at(1.5,5){$s_{18}$};
\node[circle,draw,minimum size=0.5cm](s19)at(3,5){$s_{19}$};
\node[circle,draw,minimum size=0.5cm](s20)at(2,8){$s_{20}$};
\node[circle,draw,minimum size=0.5cm](s21)at(2,7){$s_{21}$};
\node[circle,draw,minimum size=0.5cm](s22)at(0.5,6){$s_{22}$};
\draw[-triangle 45](s0)to node[auto]{$a$}(s14);
\draw[-triangle 45](s0)to node[auto,pos=0.2]{$c$}(s16);
\draw[-triangle 45](s1)to node[auto]{$b$}(s0);
\draw[-triangle 45](s1)to node[auto,swap,inner sep=2pt,pos=0.35]{$c$}(s12);
\draw[-triangle 45](s2)to node[auto,pos=0.45]{$a$}(s1);
\draw[-triangle 45](s2)to node[auto,swap,inner sep=2pt,pos=0.5]{$c$}(s7);
\draw[-triangle 45](s3)to node[auto,pos=0.3]{$b$}(s2);
\draw[-triangle 45](s4)to node[auto]{$a$}(s18);
\draw[-triangle 45](s4)to node[auto,pos=0.7]{$d$}(s0);
\draw[-triangle 45](s5)to[out=-90,in=15] node[auto]{$a$}(s10);
\draw[-triangle 45](s5)to [out=150,in=-15]node[auto,swap,inner sep=2pt,pos=0.35]{$c$}(s4);
\draw[-triangle 45](s6)to node[auto,swap,inner sep=2pt,pos=0.35]{$d$}(s5);
\draw[-triangle 45](s7)to node[auto,swap,inner sep=2pt,pos=0.45]{$a$}(s12);
\draw[-triangle 45](s7)to node[auto]{$d$}(s8);
\draw[-triangle 45](s8)to[out=180,in=-130] node[auto,swap,pos=0.3]{$a$}(s13);
\draw[-triangle 45](s8)to node[auto]{$e$}(s9);
\draw[-triangle 45](s9)to node[auto]{$a$}(s3);
\draw[-triangle 45](s10)to[out=180,in=-90] node[auto]{$c$}(s18);
\draw[-triangle 45](s10)to node[auto]{$e$}(s11);
\draw[-triangle 45](s11)to node[auto]{$c$}(s19);
\draw[-triangle 45](s12)to node[auto,pos=0.65]{$b$}(s16);
\draw[-triangle 45](s12)to[out=-160,in=90] node[auto,swap,inner sep=2pt,pos=0.5]{$d$}(s13);
\draw[-triangle 45](s13)to node[auto,swap,inner sep=2pt,pos=0.45]{$b$}(s17);
\draw[-triangle 45](s13)to [out=-140,in=150] node[auto,pos=0.4]{$e$}(s3);
\draw[-triangle 45](s14)to node[auto]{$c$}(s20);
\draw[-triangle 45](s14)to node[auto]{$e$}(s15);
\draw[-triangle 45](s15)to node[auto]{$b$}(s5);
\draw[-triangle 45](s15)to[out=-100,in=15] node[auto,swap,inner sep=2pt,pos=0.35]{$c$}(s21);
\draw[-triangle 45](s16)to node[auto,swap,inner sep=2pt,pos=0.1]{$a$}(s20);
\draw[-triangle 45](s16)to node[auto]{$d$}(s17);
\draw[-triangle 45](s17)to node[auto]{$a$}(s22);
\draw[-triangle 45](s17)to node[auto]{$e$}(s2);
\draw[-triangle 45](s18)to node[auto,pos=0.2]{$d$}(s14);
\draw[-triangle 45](s18)to node[auto]{$e$}(s19);
\draw[-triangle 45](s19)to node[auto]{$b$}(s6);
\draw[-triangle 45](s19)to node[auto,swap]{$d$}(s15);
\draw[-triangle 45](s20)to node[auto,swap,pos=0.7]{$d$}(s22);
\draw[-triangle 45](s20)to node[auto]{$e$}(s21);
\draw[-triangle 45](s21)to node[auto]{$b$}(s4);
\draw[-triangle 45](s21)to node[auto]{$d$}(s1);
\draw[-triangle 45](s22)to node[auto,swap,pos=0.25]{$e$}(s1);
\end{tikzpicture}\\
\begin{tikzpicture}[scale=1.0]
\draw(-3,3)node{$\PNref{ts21}$:};
\node[draw,minimum size=0.4cm](a)at(0,3){$a$};
\node[draw,minimum size=0.4cm](b)at(2,3){$b$};
\node[draw,minimum size=0.4cm](c)at(-2,1){$c$};
\node[draw,minimum size=0.4cm](d)at(0.5,-1){$d$};
\node[draw,minimum size=0.4cm](z)at(2,-3.5){$e$};
\node[circle,draw,minimum size=0.4cm](p0)at(1,3){};
\node[circle,draw,minimum size=0.4cm](p1)at(-0.5,0){};
\node[circle,draw,minimum size=0.4cm](p2)at(4.5,-3){};
\node[circle,draw,minimum size=0.4cm](p3)at(1,2){};
\filldraw[black](0.9,2)circle(1.5pt);\filldraw[black](1.1,2)circle(1.5pt);
\node[circle,draw,minimum size=0.4cm](p4)at(2,0){};
\filldraw[black](1.9,0)circle(1.5pt);\filldraw[black](2.1,0)circle(1.5pt);
\node[circle,draw,minimum size=0.4cm](p5)at(3,3){};
\filldraw[black](3,3)circle(1.5pt);
\node[circle,draw,minimum size=0.4cm](p6)at(-3,-1){};
\filldraw[black](-3,-1)circle(1.5pt);
\node[circle,draw,minimum size=0.4cm](p7)at(-2,-3){};
\filldraw[black](-2.1,-3)circle(1.5pt);\filldraw[black](-1.9,-3)circle(1.5pt);
\filldraw[black](-2,-3.1)circle(1.5pt);\filldraw[black](-2,-2.9)circle(1.5pt);
\node[circle,draw,minimum size=0.4cm](p8)at(3,-2){};
\filldraw[black](2.9,-2)circle(1.5pt);\filldraw[black](3.1,-2)circle(1.5pt);
\filldraw[black](3,-2.1)circle(1.5pt);\filldraw[black](3,-1.9)circle(1.5pt);
\node[circle,draw,minimum size=0.4cm](p9)at(6,-1){};
\filldraw[black](5.9,-1)circle(1.5pt);\filldraw[black](6.1,-1)circle(1.5pt);
\node[circle,draw,minimum size=0.4cm](p10)at(0,-2){};
\filldraw[black](0,-2)circle(1.5pt);
\draw[-latex](p5)to[out=150,in=20](a);
\draw[-latex](p9)to[out=30,in=40](a);
\draw[-latex](a)to(p0);
\draw[-latex](a)to[out=-135,in=130] node[auto,swap,pos=0.25,inner sep=1pt]{$3$}(p4);
\draw[-latex](a)to[out=180,in=180] node[auto,swap,pos=0.25,inner sep=1pt]{$2$}(p7);
\draw[-latex](p0)to(b);
\draw[-latex](p2)to(b);
\draw[-latex](p4)to node[auto,pos=0.5,inner sep=1pt]{$2$}(b);
\draw[-latex](b)to(p3);
\draw[-latex](b)to(p5);
\draw[-latex](b)to node[auto,pos=0.7,inner sep=1pt]{$\!2$}(p8);
\draw[-latex](p3)to(c);
\draw[-latex](p3)to(c);
\draw[-latex](p6)to(c);
\draw[-latex](p10)to(c);
\draw[-latex](p7)to node[auto,pos=0.25,inner sep=1pt]{$2$}(c);
\draw[-latex](c)to node[auto,pos=0.5,inner sep=1pt]{$2$}(p4);
\draw[-latex](c)to(p1);
\draw[-latex](p1)to(d);
\draw[-latex](d)to(p8);
\draw[-latex](d)to(p6);
\draw[-latex](d)to(p9);
\draw[-latex](d)to node[auto,swap,pos=0.5,inner sep=1pt]{$3$}(p7);
\draw[-latex](p4)to node[auto,swap,pos=0.7,inner sep=1pt]{$3$}(z);
\draw[-latex](p7)to node[auto,swap,pos=0.5,inner sep=1pt]{$3$}(z);
\draw[-latex](p8)to node[auto,pos=0.25,inner sep=1pt]{$3$}(z);
\draw[-latex](z)to(p2);
\draw[-latex](z)to(p10);
\end{tikzpicture}
}
\end{center}\vspace*{-2mm}
\caption{A reversible \lts{}, with a possible Petri net solution.}
\label{aug21-synet-aut.fig}
\end{figure}

\begin{definition}\label{PN.def} {\sc Weighted P/T nets}\\
A (finite, place-transition, arc-weighted) Petri net is a triple $\PN\!=\!(P,T,F)$
such that $P$ is a finite set of places,
$T$ is a finite set of transitions, with $P\cap T=\es$,
$F$ is a flow function $F\colon((P\times T)\cup(T\times P))\to\nsymbol$.
The incidence matrix $C$ of $\PN$ is the member of $\zsymbol^{P\times T}$ such that $\forall p\in P,t\in T:C(p,t)=F(t,p)-F(p,t)$.
The predecessors of a node $x$ form the set ${}^\dt x=\{y|F(y,x)>0\}$.
Symmetrically its successor set is $x^\dt=\{y|F(x,y)>0\}$.

\medskip
A marking is a mapping $M\colon P\to\nsymbol$,
indicating the number of (black) tokens in each place.
Let $M_1$ and $M_2$ be two markings, we shall say that
$M_1$ is dominated by $M_2$ if $M_1\lneqq M_2$, i.e., $M_1$ is distinct from $M_2$ and componentwise not greater.
A Petri net system is a net provided with an initial marking $(P,T,F,M_0)$.
A transition $t\in T$ is enabled by a marking $M$,
denoted by $M\goesto{t}$ or $M[t\rangle$, if for all places $p\in P$, $M(p)\geq F(p,t)$.
If $t$ is enabled at $M$, then $t$ can occur (or fire) in $M$,
leading to the marking $M'$ defined by $M'(p)=M(p)-F(p,t)+F(t,p)$
and denoted by $M\goesto{t}M'$ or $M[t\rangle M'$;
as usual, $[M\rangle$ denotes the set of markings reachable from $M$.
A Petri net system is bounded if, for some integer $k$, $\forall M\in[M_0\rangle\forall p\in P: M[p]\leq k$.
It is safe if $\forall M\in[M_0\rangle\forall p\in P: M[p]\leq 1$ (i.e., no place will ever receive more than 1 token).
It is $k$-safe, for some  known bound $k$, if $\forall M\in[M_0\rangle\forall p\in P: M[p]\leq k$.

Two Petri net systems $N_1=(P_1,T,F_1,M^1_0)$ and $N_2=(P_2,T,F_2,M^2_0)$
with the same transition set $T$ are isomorphic, denoted $N_1\equiv_T N_2$
(or simply $N_1\equiv N_2$ if $T$ is clear from the context),
if there is a bijection $\zeta\colon P_1\to P_2$ such that, $\forall p_1\in P_1,t\in T$:
$M^1_0(p_1)=M^2_0(\zeta(p_1))$, $F_1(p_1,t)=F_2(\zeta(p_1),t)$
and $F_1(t,p_1)=F_2(t,\zeta(p_1))$.

The reachability graph of a Petri net system is the labelled transition system whose initial state is $M_0$,
whose vertices are the reachable markings,
and whose edges are $\{(M,t,M')\mid M\goesto{t}M'\}$. It is finite iff the system is bounded.
Two isomorphic Petri net systems have isomorphic reachability graphs, hence we shall usually consider Petri nets up to isomorphism.
In examples, when it is clear that we have an initial marking, we shall often use the shorter terminology
 ``Petri net''  instead of the longer ``Petri net system''.
A labelled transition system is PN-solvable if it is isomorphic to the reachablility graph of  a Petri net system
(called a possible solution).\QED
\end{definition}

A classical property of reachability graphs of Petri net systems is the following:

\begin{proposition}{\bf (State equation)}\label{steq.prop}\\
In the reachability graph of a Petri net system $(P,T,F,M_0)$, if $\sigma\in(\pm T)^*$ and $M[\sigma\rangle M'$, then\\
$M'=M+C\cdot\Parikh(\sigma)$.
\end{proposition}

An immediate observation is that the reachability graph of any Petri net system is totally reachable and deterministic;
 it is also weakly periodic (see the state equation), and finite when bounded.
Hence, if a transition system is not totally reachable or not deterministic, it may not be PN-solvable.
The bottom of Figure~\ref{aug21-synet-aut.fig} illustrates a Petri net system,
which is a possible PN-solution of the transition system given on top of the same figure.
It may happen that a transition system has no PN-solution,
but if it has one, it has many ones, sometimes with very different structures.

\medskip
When linking transition systems and Petri nets, it is useful to introduce regions.

\begin{definition}{\sc Regions}\label{reg.def}\\
A region $(\rho,\bk,\fw)$ of a transition system $\TS=(S,\to,T,\is)$ is
a triple of functions $\rho$ (from states to $\nsymbol$), and $\bk,\fw$ (both from labels to $\nsymbol$),
satisfying the property that for any states $s,s'$ and label $a$:
\[(s,a,s')\text{ is an edge of $\TS$ }\;\;\impl\;\;\rho(s)\geq\bk(a)\land\rho(s')-\rho(s)=\fw(a)-\bk(a)
\]
since this is the typical behaviour of a place $p$ with token count $\rho$ during the firing of $a$
with backward and forward connections $\bk,\fw$ to $p$ (anywhere in $\TS$).
To solve an $\text{SSP}(s_1,s_2)$ (for State Separation Problem, where $s_1\neq s_2$),
we need to find an appropriate region $(\rho,\bk,\fw)$ satisfying $\rho(s_1)\neq\rho(s_2)$,
i.e., separating states $s_1$ and $s_2$.
For an $\text{ESSP}(s,a)$ (for Event-State Separation Problem, where $a$ is not enabled at $s$),
we need to find a region $(\rho,\bk,\fw)$ with $\rho(s)<\bk(a)$.
This can be done by solving suitable systems of linear inequalities which
arise from these two requirements, and from the requirement that regions are not  too restrictive.\QED
\end{definition}

A classical result~\cite{DR-synth-96} is that a transition system is PN-solvable if and only if
all its SSP and ESSP problems may be solved, and the corresponding places yield a possible solution.

\section{Products and sums} \label{prod.sct}

A product of two disjoint \lts{} is again an \lts{}. Its states are pairs of states of the two \lts{} and an edge exists
if one of the underlying states can do the transition. An example is shown in Figure~\ref{example-products.fig}.

\begin{definition}{\sc Product of two disjoint \lts}\label{prod.def}\\
Let $\TS_1=(S_1,\to_1,T_1,\is_1)$ and $\TS_2=(S_2,\to_2,T_2,\is_2)$
be two \lts{} with disjoint label sets  ($T_1\cap T_2=\emptyset$).
The (disjoint) product $\TS_1\otimes\TS_2$ is  the \lts{}
$\big(S_1\times S_2,\to, T_1\uplus T_2,(\is_1,\is_2)\big)$, where
$\mathord{\to}
   = \{\big((s_1,s_2),t_1,(s'_1,s_2)\big)  \mid (s_1,t_1,s'_1)\in\mathord{\to_1}\}
\cup \{\big((s_1,s_2),t_2,(s_1,s'_2)\big)  \mid (s_2,t_2,s'_2)\in\mathord{\to_2}\}$.\QED
\end{definition}

\begin{figure}[htb]
\vspace*{-1mm}
\centering
\begin{tikzpicture}[state/.style={circle,fill=black!100, inner sep=0.05cm}, edge/.style={-triangle 45, auto}, baseline=(s1)]
\refstepcounter{exampleTScounter}\label{prod_part1.ts}
\node[]at(-0.75, 0.5){$\TS_{\ref{prod_part1.ts}}$};
\node[state, label=above:$0$] (s1)at(0,0) {};
\node[state, label=above:$1$] (s2)at(1,0) {};
\draw[edge](s1)--node{$a$}(s2);
\end{tikzpicture}
\hfil
\begin{tikzpicture}[state/.style={circle,fill=black!100, inner sep=0.05cm}, edge/.style={-triangle 45, auto}, baseline=(s1)]
\refstepcounter{exampleTScounter}\label{prod_part2.ts}
\node[]at(-0.75, 0.5){$\TS_{\ref{prod_part2.ts}}$};
\node[state, label=left:$\is$] (s1)at(0,0) {};
\node[state, label=left:$s_1$]   (s2)at(0,-1) {};
\node[state, label=left:$s_2$]   (s3)at(0,-2) {};
\draw[edge](s1)--node{$b$}(s2);
\draw[edge](s2)--node{$b$}(s3);
\end{tikzpicture}
\hfil
\begin{tikzpicture}[state/.style={circle,fill=black!100, inner sep=0.05cm}, edge/.style={-triangle 45, auto}, baseline=(s1)]
\refstepcounter{exampleTScounter}\label{prod_result.ts}
\node[]at(-0.75, 0.5){$\TS_{\ref{prod_result.ts}}$};
\node[state, label=left:{$(0,\is)$}] (s1)at(0,0) {};
\node[state, label=left:{$(0,s_1)$}] (s2)at(0,-1) {};
\node[state, label=left:{$(0,s_2)$}] (s3)at(0,-2) {};
\node[state, label=right:{$(1,\is)$}] (s4)at(1,0) {};
\node[state, label=right:{$(1,s_1)$}] (s5)at(1,-1) {};
\node[state, label=right:{$(1,s_2)$}] (s6)at(1,-2) {};
\draw[edge](s1)--node{$b$}(s2);
\draw[edge](s2)--node{$b$}(s3);
\draw[edge](s4)--node{$b$}(s5);
\draw[edge](s5)--node{$b$}(s6);
\draw[edge](s1)--node{$a$}(s4);
\draw[edge](s2)--node{$a$}(s5);
\draw[edge](s3)--node{$a$}(s6);
\end{tikzpicture}\vspace*{-1mm}
\caption{Example for a disjoint product. We have $\TSref{prod_part1.ts}\otimes\TSref{prod_part2.ts}=\TSref{prod_result.ts}$.}
\label{example-products.fig}
\end{figure}
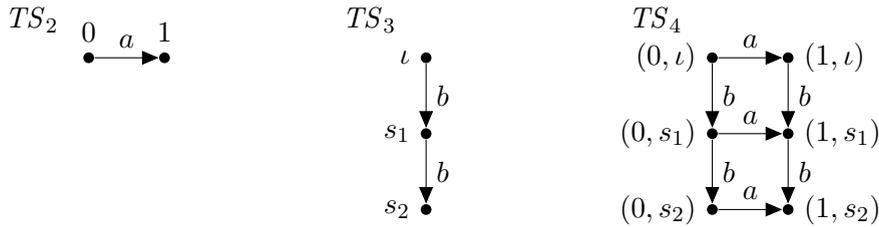

When a product is given and the individual label sets $T_1$ and $T_2$ are known, the factors can be computed by only
following edges with labels in $T_1$, resp.\ $T_2$, from the initial state:

\begin{proposition}{\bf (Factors of a product \cite{devacta17})}\label{dec.prop}\\
When $\TS=(S,\to,T_1\uplus T_2,\is)\equiv_T\TS_1\otimes\TS_2$,
$\TS$ is totally reachable (resp. deterministic) iff so are the factors $\TS_1$ and $\TS_2$.

\medskip
Moreover, $(s_1,s_2)[\sigma\rangle (s'_1,s'_2)$ in $\TS$ iff
$s_1[\sigma_1\rangle s_2$ in $\TS_1$ and $s_2[\sigma_2\rangle s'_2$ in $\TS_2$,
where $\sigma_1$ is the projection of $\sigma$ on $T_1^*$ and similarly for $\sigma_2$.
\eject
Then, $\TS_1\equiv_T (S^1,\to^1,T_1,\is)$ with
$S^1=\{s\in S \mid \exists \alpha_1\in  T_1^*:\is[\alpha\rangle s\}$
and $\mathord{\to^1}=\{(s_1,t_1,s_2)\in\mathord{\to}\mbox{ such that } s_1,s_2\in S^1,t_1\in T_1\}$,
and similarly for $\TS_2$.

Up to isomorphism\footnote{Those properties could be expressed in terms of categories,
but we shall refrain from doing this here.}, the disjoint product of \lts{} is commutative,
associative and has a neutral (the LTS with a single state and no label).
Each \lts{} has itself and the neutral system as (trivial) factors.

An \lts~is prime if it has exactly two factors (up to isomorphism).
If $\TS=(S,\to,T,\is)$ is connected and finite, there is a finite set $I$ of indices and a unique set of connected
prime \lts's $\{\TS_i|i\in I\}$ such that $\TS\equiv_T\bigotimes_{i\in I}\TS_i$.
\end{proposition} 

There is an interesting relation between \lts{} products and Petri nets:
if two nets are disjoint, putting them side by side yields a new net whose reachability graph is (up to isomorphism)
the disjoint product of the reachability graphs of the two original nets. This may be generalised up to Petri net isomorphism:

\begin{definition}{\sc (Disjoint) sum of Petri nets}\label{dsum.def}\\
Let $N_1=(P_1,T_1,F_1,M^1_0)$  and  $N_2=(P_2,T_2,F_2,M^2_0)$  be two Petri net systems
with disjoint  transition sets ($T_1\cap T_2=\emptyset$).
The disjoint sum $N_1\oplus N_2$ is defined (up to isomorphism) as the system $N=(P,T_1\cup T_2,F,M_0)$ where
$P=\zeta_1(P_1)\cup\zeta_2(P_2)$; $F(\zeta_1(p_1),t_1)=F_1(p_1,t_1)$,
$F(t_1,\zeta_1(p_1))=F_1(t_1,p_1)$, $F(\zeta_2(p_2),t_2)=F_2(p_2,t_2)$,
$F(t_2,\zeta_2(p_2))=F_2(t_2,p_2)$, $M_0(\zeta_1(p_1))=M_0^1(p_1)$ and $M_0(\zeta_2(p_2))=M_0^2(p_2)$,
for $t_1\in T_1, t_2\in T_2, p_1\in P_1, p_2\in P_2$.
In these formulas,  $\zeta_1$ is a bijection between $P_1$ and $\zeta_1(P_1)$ and
$\zeta_2$ is a bijection between $P_2$ and $\zeta_2(P_2)$ such that
$\zeta_1(P_1)\cap(\zeta_2(P_2)\cup T_1\cup T_2)=\emptyset$
and $\zeta_2(P_2)\cap(\zeta_1(P_1)\cup T_1\cup T_2)=\emptyset$.\QED
\end{definition}

It may be observed that the resulting system is not uniquely defined since it depends on the choice of the two bijections $\zeta_1$ and $\zeta_2$ used to separate the place sets from the rest, but this is irrelevant since we want to work up to isomorphism.
Again, up to isomorphism,
the disjoint sum of Petri net systems is commutative, associative and has a neutral (the empty net).

An additional remark that may be precious for applications  is that many subclasses of Petri nets found in the literature
(plain, free-choice, choice-free, join-free, fork-attribution, homogeneous, \ldots, not defined here)
are compatible with the presented (de)composition,
in the sense that a disjoint sum of nets belongs to such a subclass if and only if
each component belongs to the same subclass.
In particular, we have

\begin{corollary}{\bf  (Bound preservation)}\label{safep.cor}\\
$N_1\oplus N_2$  is safe iff so are $N_1$ and $N_2$.
If $N_1\oplus N_2$ is $k$-safe, so are $N_1$ and $N_2$.

\medskip
Finally, if $N_1$ is $k_1$-safe and $N_2$ is $k_2$-safe, then $N_1\oplus N_2$ is $\max(k_1,k_2)$-safe.
\end{corollary} 

\begin{proposition}{\bf (Reachability graph of a sum of nets)} \label{sum0.prop}\\
The  reachability graph of a disjoint sum of net systems is isomorphic to
the disjoint product of the reachability graphs of the composing nets.
\end{proposition}

There is a kind of reverse of this property \cite{dev-ACSD16,devacta17}.

\begin{proposition}{\bf (Petri net solution of a disjoint product of \lts)}\label{sum.prop}\\
A disjoint product of \lts{} has a  Petri net solution iff each composing \lts{}
has a  Petri net solution, and a possible solution is the disjoint sum of the latter.
\end{proposition}

Let us now examine when and how an \lts{} may be decomposed into (non-trivial) disjoint factors.
From Proposition~\ref{dec.prop}, it is enough to discover an adequate decomposition of the label set $T=T_1\uplus T_2$.
A general characterisation of such adequate decompositions is presented in \cite{dev-ACSD16},
but it may be simplified in the context of Petri net synthesis.

\begin{definition}{\sc General diamond property} \label{gdiam.def}\\
An \lts{} $\TS=(S,\to,T,\is)$ presents the {\it general diamond property} for two distinct labels $a,b\in T$ if,
whenever there are two adjacent edges in a diamond like in Fig.~\ref{gdiam.fig}, the other two are also present
(in other words, $\forall s,s_1,s_2\in S$, $\forall u\in\{a,-{a}\},\forall v\in\{b,-{b}\}:
s[u\rangle s_1 \land s[v\rangle s_2 \impl (\exists s'\in S: s_1[v\rangle s' \land s_2[u\rangle s')$).
If $T_1,T_2\subseteq T$ with $T_1\cap T_2=\emptyset$, $\TS$ will be said {\it $\{T_1,T_2\}$-gdiam}
if it presents the general diamond property for each pair of labels $a\in T_1, b\in T_2$
(note that any \lts{} $\TS=(S,\to,T,\is)$ is $\{\emptyset,T\}$-gdiam).\QED
\end{definition}

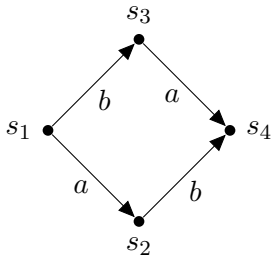
\begin{figure*}[hbt]
\begin{center}
\begin{tikzpicture}[scale=1.2]
\node[circle,fill=black!100,inner sep=0.05cm](s1)at(0,0)[label=left:$s_1$]{};
\node[circle,fill=black!100,inner sep=0.05cm](s2)at(1,-1)[label=below:$s_2$]{};
\node[circle,fill=black!100,inner sep=0.05cm](s3)at(1,1)[label=above:$s_3$]{};
\node[circle,fill=black!100,inner sep=0.05cm](s4)at(2,0)[label=right:$s_4$]{};
\draw[-triangle 45](s1)--node[auto,swap,inner sep=1.5pt,pos=0.5]{$a$}(s2);
\draw[-triangle 45](s1)--node[auto,swap,inner sep=1.5pt,pos=0.5]{$b$}(s3);
\draw[-triangle 45](s2)--node[auto,swap,inner sep=1.5pt,pos=0.5]{$b$}(s4);
\draw[-triangle 45](s3)--node[auto,swap,inner sep=1.5pt,pos=0.5]{$a$}(s4);
\node[]at(6.5,1.2){$s_1[a\rangle s_2\;\land\;s_1[b\rangle s_3\Rightarrow \exists s_4:s_3[a\rangle s_4\;\land\;s_2[b\rangle s_4$};
\node[]at(6.5,0.4){$s_3[a\rangle s_4\;\land\;s_2[b\rangle s_4\Rightarrow \exists s_1:s_1[a\rangle s_2\;\land\;s_1[b\rangle s_3$};
\node[]at(6.5,-0.4){$s_1[a\rangle s_2\;\land\;s_2[b\rangle s_4\Rightarrow \exists s_3:s_1[b\rangle s_3\;\land\;s_3[a\rangle s_4$};
\node[]at(6.5,-1.2){$s_1[b\rangle s_3\;\land\;s_3[a\rangle s_4\Rightarrow \exists s_2:s_1[a\rangle s_2\;\land\;s_2[b\rangle s_4$};
\node[]at(8,0.8){(forward persistence)};
\node[]at(8,0){(backward persistence)};
\node[]at(8,-0.8){(permutation)};
\node[]at(8,-1.6){(permutation)};
\end{tikzpicture}
\end{center}\vspace*{-8mm}
\caption{General diamond property.}
\label{gdiam.fig}
\end{figure*}

Note that the four configurations in Figure~\ref{gdiam.fig} are condensed in a single formula in Definition~\ref{gdiam.def},
using both direct and reverse arcs.
Among the many properties implied by general diamonds, we may cite:

\begin{proposition}{\bf  (Projections and permutation \cite{fact18})}\label{perm.prop}\\
Let $\TS=(S,\to,T,\is)$ be a $\{T_1,T_2\}$-gdiam \lts{} with $T_1\cap T_2=\emptyset$.
If $s[\alpha\rangle s'$ for some $s,s'\in S$ and general path $\alpha\in (\pm T_1\cup \pm T_2)^*$,
let $\alpha_1$ be the projection of $\alpha$ on $\pm T_1$ (i.e., $\alpha$ where all the elements
in $\pm T_2$ are dropped)
and $\alpha_2$ be the projection of $\alpha$ on $\pm T_2$ (thus dropping the elements in $\pm T_1$).
Then there are $s_1,s_2\in S$ such that $s[\alpha_1\rangle s_1[\alpha_2\rangle s'$
and $s[\alpha_2\rangle s_2[\alpha_1\rangle s'$.
\end{proposition}

This also implies a variant of the well-known Keller's theorem \cite{keller},
meaning that the general diamonds property is local, but implies a global variant.

\begin{proposition}{\bf (General diamonds imply big general diamonds \cite{fact18})} \label{genkel.prop}\\
Let $\TS=(S,\to,T,\is)$ be a $\{T_1,T_2\}$-gdiam \lts{} with $T_1\cap T_2=\emptyset$.
If $s[\alpha_1\rangle s_1$ and $s[\alpha_2\rangle s_2$ for some $s,s_1,s_2\in S$,
$\alpha_1\in (\pm T_1)^*$
and $\alpha_2\in (\pm T_2)^*$,
then for some $s'\in S$, $s_1[\alpha_2\rangle s'$ and $s_2[\alpha_1\rangle s'$.
\end{proposition} 

This may be interpreted as the fact that big general diamonds are filled with small ones.

\begin{proposition}{\bf  (Sequences generated by big general diamonds)}\label{seqdiam.prop}\\
Let $\TS=(S,\to,T_1\uplus T_2,\is)$ be a $\{T_1,T_2\}$-gdiam \lts{}.
If $s_0[\alpha_1\rangle s_1$ and $s_0[\alpha_2\rangle s_1$ for some $s_0,s_1\in S$,
$\alpha_1\in (\pm T_1)^*$ and $\alpha_2\in (\pm T_2)^*$,
\begin{enumerate}
\itemsep=0.9pt
\item
then for any $i\in\zsymbol$, there is some $s_i\in S$ such that $s_i[\alpha_1\rangle s_{i+1}$ and $s_i[\alpha_2\rangle s_{i+1}$;
\item
if $S$ is finite, for some $h\neq k\in\zsymbol$, $s_h=s_k$;
\item
if the various $s_i$'s belong to $[\is\rangle^{T_1}$
which is Petri net solvable, then $s_h=s_k$ for any $h,k\in\zsymbol$.
\end{enumerate}
\end{proposition}

\begin{proof}
The first point results from successive applications of Proposition~\ref{genkel.prop} to $s_i[\alpha_1\rangle s_{i+1}$
and $s_i[\alpha_2\rangle s_{i+1}$, as well as to $s_{i+1}[-\alpha_1\rangle s_{i}$ and $s_{i+1}[-\alpha_2\rangle s_{i}$.

\medskip
The case where $S$ is finite is then immediate.

The last point results from the fact that, in any Petri net solution of $[\is\rangle^{T_1}$, if $M_s$ is
the marking corresponding to state $s$ (if reachable), then from Proposition~\ref{steq.prop},
$\forall i\in\zsymbol: M_{s_{i+1}}-M_{s_i}=C\cdot\Parikh(\alpha_1)$, hence is constant.
As a consequence, $\forall i,j\in\zsymbol:M_{s_j}-M_{s_i}=(j-i)\cdot C\cdot\Parikh(\alpha_1)$.
When $S$ is finite, from the previous point, $C\cdot\Parikh(\alpha_1)=0$, the markings corresponding to all $s_i$'s are the same,
hence the $s_i$'s are also the same.
But this remains true if $\TS$ is infinite, because all the markings are nonnegative
(if for some $p\in P$ we have $C\cdot\Parikh(\alpha_1)\neq 0$, then either $M_{s_i}(p)$ becomes negative for $i$ large enough,
or for $i$ small enough in $\zsymbol$).
\end{proof}

Now, the interest of the notion of general diamonds in our context arises from the following observation:

\begin{proposition}\label{gdiam.prop}{\bf (Product implies general diamonds \cite{fact18})}\\
Let $\TS_1=(S_1,\to_1,T_1,\is_1)$ and $\TS_2=(S_2,\to_2,T_2,\is_2)$
 be two  \lts{} with disjoint label sets, then $\TS_1\otimes\TS_2$
presents the general diamond property for each $a\in T_1$ and $b\in T_2$.
\end{proposition} 

\begin{figure}[!b]
\vspace*{-4mm}
\begin{center}
\begin{tikzpicture}[scale=2.0]
\refstepcounter{exampleTScounter}\label{diam1.ts}
\node[]at(0,0.75){$\TS_{\ref{diam1.ts}}$};
\node[circle,fill=black!100,inner sep=0.05cm](s1)at(0,0)[label=left:$\is$]{};
\node[circle,fill=black!100,inner sep=0.05cm](s2)at(1,0)[label=right:$s_1$]{};
\node[circle,fill=black!100,inner sep=0.05cm](s3)at(1,1)[label=above:$s_2$]{};
\draw[-triangle 45](s1)--node[auto,swap,inner sep=1.5pt,pos=0.5]{$a$}(s2);
\draw[-triangle 45](s3)--node[auto,swap,inner sep=1.5pt,pos=0.5]{$a$}(s2);
\draw[-triangle 45](s2)edge[-triangle 45, out=-135, in=-45, distance=1cm]node[auto,swap]{$b$}(s2);
\draw[-triangle 45](s3)edge[bend right=20]node[auto,swap,inner sep=1pt,pos=0.5]{$b$}(s1);
\draw[-triangle 45](s1)edge[bend right=20]node[auto,swap,inner sep=1pt,pos=0.5]{$b$}(s3);
\end{tikzpicture}\hspace{2em}
\begin{tikzpicture}[scale=2.0]
\refstepcounter{exampleTScounter}\label{diam2.ts}
\node[]at(0,0.75){$\TS_{\ref{diam2.ts}}$};
\node[circle,fill=black!100,inner sep=0.05cm](s1)at(0,0)[label=left:$\is$]{};
\node[circle,fill=black!100,inner sep=0.05cm](s2)at(1,-1)[label=below:$s_2$]{};
\node[circle,fill=black!100,inner sep=0.05cm](s3)at(1,0)[label=above:$s_3$]{};
\node[circle,fill=black!100,inner sep=0.05cm](s4)at(2,0)[label=right:$s_4$]{};
\node[circle,fill=black!100,inner sep=0.05cm](s5)at(1,1)[label=above:$s_5$]{};
\draw[-triangle 45](s1)--node[auto,swap,inner sep=1.5pt,pos=0.5]{$a$}(s2);
\draw[-triangle 45](s1)--node[auto,swap,inner sep=1.5pt,pos=0.5]{$b$}(s3);
\draw[-triangle 45](s2)--node[auto,swap,inner sep=1.5pt,pos=0.5]{$b$}(s4);
\draw[-triangle 45](s3)--node[auto,swap,inner sep=1.5pt,pos=0.5]{$a$}(s4);
\draw[-triangle 45](s1)--node[auto,swap,inner sep=1.5pt,pos=0.5]{$a$}(s5);
\draw[-triangle 45](s5)--node[auto,swap,inner sep=1.5pt,pos=0.5]{$b$}(s4);
\end{tikzpicture}
\end{center}\vspace*{-6mm}
\caption{General diamond property does not imply product.}
\label{gddiag.fig}
\end{figure}
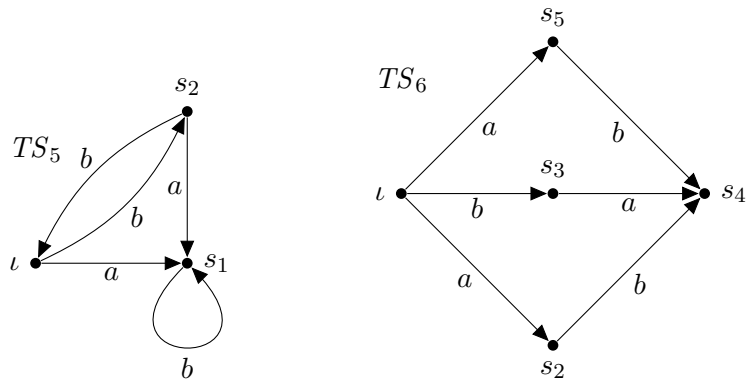

Unfortunately, the reverse property does not hold in all generality, as shown by the counterexamples in Fig.~\ref{gddiag.fig}.
Indeed, in both cases, $a$ and $b$ form general diamonds. However, from Proposition~\ref{dec.prop},
with $T_1=\{a\}$ and $T_2=\{b\}$,
 $\TS_{\ref{diam1.ts}}$ should be isomorphic to the product of $\is\xrightarrow[]{a}s_1$ and $\is\xleftrightarrow[]{b}s_2$,
 since $\{\is,s_1\}$ are the only states directly reachable with $a$ and  $\{\is,s_2\}$ are the only states directly reachable with $b$;
 however, this is not the case since we would then have four states in the product.
 Similarly, $\TS_{\ref{diam2.ts}}$ should be isomorphic to the product of $\{\is\xrightarrow[]{a}s_2,\is\xrightarrow[]{a}s_5\}$
 and of a single arrow $\is\xrightarrow[]{b}s_3$, which is not the case since then we would have six states in the product.

But note that both counterexample are non-deterministic, hence cannot have Petri net solutions.
And indeed, an essential result is that the local general diamond property suffices in the context of synthesis:

\begin{theorem} \label{PNsynt.thm}{\bf [General diamonds and Petri net synthesis imply product]}\\
If a (finite) totally reachable \lts{} $\TS=(S,\to,T_1\uplus T_2,\is)$ is deterministic and satisfies the general diamond property
for each pair of labels $a\in T_1$ and $b\in T_2$,
then it is Petri net synthesisable iff so are $[ \is \rangle^{T_1}$ and $[\is \rangle^{T_2}$;
moreover, we then have
$\TS\equiv_T[\is \rangle^{T_1}\otimes[ \is \rangle^{T_2}$
and therefore a possible solution of the synthesis problem for $\TS$
is the disjoint sum of a solution of $[\is \rangle^{T_1}$
and a solution of $[\is \rangle^{T_2}$.
\end{theorem} 

\begin{proof}
First, we may observe that if $\TS$ is Petri net synthesisable, (it is deterministic, totally reachable and)
by dropping the transitions in $T_2$ in such a solution we shall get a solution to $[ \is \rangle^{T_1}$,
and by dropping the transitions in $T_1$ we shall get a solution to $[ \is \rangle^{T_2}$.

\medskip
 Let us now assume that $\TS$ is totally reachable, deterministic and satisfies the general diamond property
for each pair of labels $a\in T_1$ and $b\in T_2$;
let us also assume that $[\is \rangle^{T_1}$ and $[\is \rangle^{T_2}$ are  Petri net synthesisable.

 For each state $s\in S$, from the total reachability we have $\is[\sigma\rangle s$ and, from Proposition~\ref{perm.prop},
 $\is[\alpha\rangle s_1[\beta\rangle s$ and $\is[\beta\rangle s_2[\alpha\rangle s$
for some $s_1,s_2\in S$, $\alpha\in T_1^*$ and $\beta\in T_2^*$ ($\alpha_1$ being the projection of $\sigma$ on $T_1^*$,
and $\beta$ being the projection of $\sigma$ on $T_2^*$), so that $s_1\in[\is\rangle^{T_1}$ and $s_2\in[\is\rangle^{T_2}$.

\medskip
Conversely, if $s_1\in[\is\rangle^{T_1}$ and $s_2\in[\is\rangle^{T_2}$, i.e., there is
$\is[\alpha\rangle s_1$ in $[\is\rangle^{T_1}$ and $\is[\beta\rangle s_2$  in $[\is\rangle^{T_2}$,
then from Proposition~\ref{genkel.prop} we also have $\is[\alpha\rangle s_1[\beta\rangle s$ and
$\is[\beta\rangle s_2[\alpha\rangle s$ for some  $s\in S$.

It remains to show that the correspondence between $s$ and the pair $(s_1,s_2)$ is unique to derive that
$\TS\equiv_T[\is \rangle^{T_1}\otimes[ \is \rangle^{T_2}$
(the compatibility of the transition rules again results from the projections and antiprojections in Propositions~\ref{perm.prop}
and~\ref{genkel.prop}),
so that $TS$ is solved by the sum of a solution of $[\is \rangle^{T_1}$ and a solution of $[\is \rangle^{T_2}$.

For the direction $(s_1,s_2)\to s$, if $\is[\alpha'_1\rangle s_1$ for some $\alpha'_1\in T_1^*$,
from Proposition~\ref{genkel.prop} we have
$\is[\alpha'_1\rangle s_1[\alpha_2\rangle s'$ and $\is[\alpha_2\rangle s_2[\alpha'_1\rangle s'$ for some $s'\in S$.
But from the determinism, $s'=s$. And similarly  if $\is[\alpha'_2\rangle s_2$ with $\alpha'_2\in T_2^*$.

 For the other direction, let us assume that $\is[\sigma\rangle s$ for $\sigma\in(T_1\cup T_2)^*$ and $s\in S$,
 and that $\is[\alpha_1\rangle s_1[\alpha_2\rangle s$ as well as $\is[\alpha'_1\rangle s'_1[\alpha'_2\rangle s$
 with $\alpha_1,\alpha'_1\in T_1^*$, $\alpha_2,\alpha'_2\in T_2^*$ and $s\in S$.
 We need to show that $s_1=s'_1$ (the case based on $T_2$ will be obtained symmetrically).

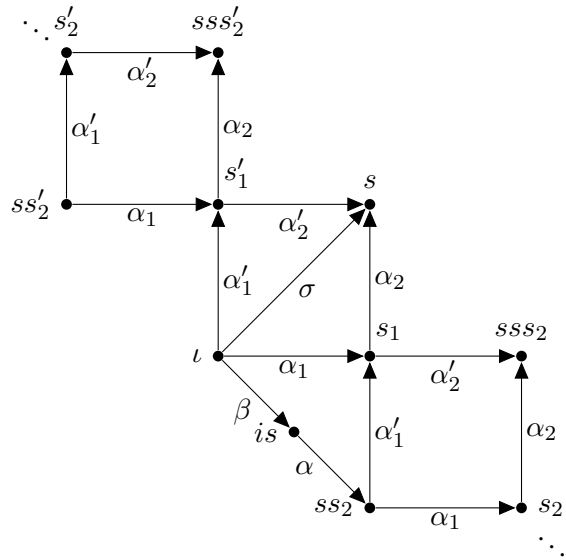
\begin{figure}[hbt]
\vspace*{-3mm}
\begin{center}
\begin{tikzpicture}[scale=2.0]
\refstepcounter{exampleTScounter}\label{diam1b.ts}
\node[circle,fill=black!100,inner sep=0.05cm](is)at(0,0)[label=left:$\is$]{};
\node[circle,fill=black!100,inner sep=0.05cm](s1)at(1,0)[label=above:\hspace*{0.5cm}$s_1$]{};
\node[circle,fill=black!100,inner sep=0.05cm](sp1)at(0,1)[label=above:\hspace*{0.5cm}$s'_1$]{};
\node[circle,fill=black!100,inner sep=0.05cm](s)at(1,1)[label=above:$s$]{};
\draw[-triangle 45](is)--node[auto,swap,inner sep=1.5pt,pos=0.5]{$\alpha_1$}(s1);
\draw[-triangle 45](s1)--node[auto,swap,inner sep=1.5pt,pos=0.5]{$\alpha_2$}(s);
\draw[-triangle 45](is)--node[auto,swap,inner sep=1.5pt,pos=0.5]{$\alpha'_1$}(sp1);
\draw[-triangle 45](sp1)--node[auto,swap,inner sep=1.5pt,pos=0.5]{$\alpha'_2$}(s);
\draw[-triangle 45](is)--node[auto,swap,inner sep=1.5pt,pos=0.5]{$\sigma$}(s);
\node[circle,fill=black!100,inner sep=0.05cm](s2)at(2,-1)[label=right:$s_2$]{};
\node[circle,fill=black!100,inner sep=0.05cm](ss2)at(1,-1)[label=left:$ss_2$]{};
\node[circle,fill=black!100,inner sep=0.05cm](sss2)at(2,0)[label=above:$sss_2$]{};
\draw[-triangle 45](s2)--node[auto,swap,inner sep=1.5pt,pos=0.5]{$\alpha_2$}(sss2);
\draw[-triangle 45](s1)--node[auto,swap,inner sep=1.5pt,pos=0.5]{$\alpha'_2$}(sss2);
\draw[-triangle 45](ss2)--node[auto,swap,inner sep=1.5pt,pos=0.5]{$\alpha_1$}(s2);
\draw[-triangle 45](ss2)--node[auto,swap,inner sep=1.5pt,pos=0.5]{$\alpha'_1$}(s1);
\node[circle,fill=black!100,inner sep=0.05cm](sp2)at(-1,2)[label=above:$s'_2$]{};
\node[circle,fill=black!100,inner sep=0.05cm](ssp2)at(-1,1)[label=left:$ss'_2$]{};
\node[circle,fill=black!100,inner sep=0.05cm](sssp2)at(0,2)[label=above:$sss'_2$]{};
\draw[-triangle 45](sp2)--node[auto,swap,inner sep=1.5pt,pos=0.5]{$\alpha'_2$}(sssp2);
\draw[-triangle 45](sp1)--node[auto,swap,inner sep=1.5pt,pos=0.5]{$\alpha_2$}(sssp2);
\draw[-triangle 45](ssp2)--node[auto,swap,inner sep=1.5pt,pos=0.5]{$\alpha'_1$}(sp2);
\draw[-triangle 45](ssp2)--node[auto,swap,inner sep=1.5pt,pos=0.5]{$\alpha_1$}(sp1);
\node[]()at(-1.2,2.2)[]{$\ddots$};\node[]()at(2.2,-1.2)[]{$\ddots$};
\node[circle,fill=black!100,inner sep=0.05cm](iss)at(0.5,-0.5)[label=left:$is$]{};
\draw[-triangle 45](is)--node[auto,swap,inner sep=1.5pt,pos=0.5]{$\beta$}(iss);
\draw[-triangle 45](iss)--node[auto,swap,inner sep=1.5pt,pos=0.3]{$\alpha$}(ss2);
\end{tikzpicture}
\end{center}\vspace*{-9mm}
\caption{Proof of Theorem~\ref{PNsynt.thm}.}
\label{PNsynt.fig}
\end{figure}

Since we have both $s'_1[(-\alpha'_1)\alpha_1\rangle s_1$ and $s'_1[\alpha'_2(-\alpha_2)\rangle s_1$ ,
with $(-\alpha'_1)\alpha_1\in (\pm T_1)^*$ and $\alpha'_2(-\alpha_2)\in  (\pm T_2)^*$,
from Proposition~\ref{seqdiam.prop}(1) we may construct a series of configurations,
forward and backward, as also illustrated in Figure~\ref{PNsynt.fig},
with the same paths $(-\alpha'_1)\alpha_1$ and $\alpha'_2(-\alpha_2)$, leading to states $s_2,s_3,\ldots$ and $s'_2,s'_3,\ldots$

Since $\TS$ is totally reachable and satisfies the general diamond property for $T_1$ and $T_2$,
from Proposition~\ref{perm.prop} we also have a path $\is[\beta\rangle is[\alpha\rangle ss_2$
with $\beta\in T_2^*$ and $\alpha\in T_1^*$.
Since $\is[\beta\rangle is$ and $\is[\alpha_1(-\alpha'_1)(-\alpha)\rangle is$ with $\alpha_1(-\alpha'_1)(-\alpha)\in(\pm T)^*$,
we may again apply Proposition~\ref{seqdiam.prop}(1) (we shall only need it forwardly here),
constructing paths of increasing length $\beta^n$.
Since $\TS$ is finite, $[\is\rangle^{T_2}$ is Petri net solvable and those paths $\is[\beta^n\rangle$ belong to $[\is\rangle^{T_2}$,
by weak periodicity (see also the proof of  Proposition~\ref{seqdiam.prop}(3)) we have that $is=\is$.
Hence $ss_2$ belongs to $[\is\rangle^{T_1}$, and we may perform the same reasoning for $ss_3,ss_4,\ldots$
(with increasing paths in $(\pm T_1)^*$), as well as for $ss'_2,ss'_3, \ldots$

As a consequence, all the states $s_1,ss_2,s_3,ss_3,\ldots,ss'_2,s'_2,ss'_3,s'_3,\ldots$
belong to $[\is\rangle^{T_1}$.
Since $[\is\rangle^{T_1}$ is by hypothesis Petri net solvable, we may then apply  Proposition~\ref{seqdiam.prop}(3)
 and we get $s_1=s'_1$, as expected.
\end{proof}

It remains to find adequate subsets of labels $T_1$ and $T_2$,  partitioning $T$ and satisfying the general diamond property.
To do that, one may rely on the following, which is again a local property.

\begin{definition}{\sc Connected labels} \label{con.def}\\
Let $\TS=(S,\to,T,\is)$ be an \lts{} and $a,b\in T$ be two distinct labels.

\medskip
We shall denote by $a\leftrightarrow b$ the fact that they do not form general diamonds, i.e.,
there are states which do not satisfy one of the constraints in Figure~\ref{gdiam.fig}.

We shall also note $\centernot\Diamond=\mathord{\leftrightarrow^*}$, i.e.,
the reflexive and transitive closure of $\leftrightarrow$,
meaning in some sense that the labels are `non-diamondisable'.\QED 
\end{definition}

These relations mean that in any decomposition, if $a\leftrightarrow b$ they
must belong to the same component, i.e., $a\in T_1\Rightarrow [a]\subseteq T_1$,
where $[a]=\{b\in T \mid a\mathrel{\centernot\Diamond} b\}$.
But from Theorem \ref{PNsynt.thm}, we know this is enough: for each equivalence class,
either the synthesis works and we have a global solution by taking the disjoint sum of all the solutions,
or one (or more) of the subproblems fails, and we know there is no global solution for the whole system
(and we may spot the culprits).

Our proposed factorisation algorithm now works as follows:
First, iterate over all states of the given lts, and for each state check if the adjacent edges
form general diamonds. If not, their labels must be in the same equivalence class.
Then, for each equivalence class \([a]\) try Petri net synthesis on \([\is\rangle^{[a]}\).
If it works, the result is the disjoint sum of the computed Petri nets.
This constructs the equivalence relation by repeatedly joining classes, but it also allows to stop the iteration early
when only one equivalence class remains (in which case no non-trivial factorisation is possible).

\section{Articulations} \label{art.sct}

Let us consider two disjoint transition systems $\TS_1$ and $\TS_2$ and a state $s$ in the first of them ($s\in S_1$);
the general idea of their articulation around $s$ is to `plug' the second one on the chosen state,
as schematised in Figure~\ref{articul.fig}.

\begin{figure}[hbt]
\vspace*{-24mm}
\begin{tikzpicture}[scale=0.88]
\draw[very thick] (-1,0) .. controls (-5,4) and (3,4) .. (-1,0);
\node[](T1)at(-1,2){$T_1$};\node[](i1)at(-1,-0.2){$\iota_1$};
{\node[circle,minimum size = 2mm,fill=black!100,inner sep=0.05cm](s1)at(0.2,2)[label=right:${s}$]{};
\node[](TS1)at(-1,-1){{$\TS_1$}};}
{\node[circle,minimum size = 2mm,fill=black!100,inner sep=0.05cm](s12)at(0.2,2)[]{};}
{\draw[very thick] (3,0) .. controls (-1,4) and (7,4) .. (3,0);
\node[](T2)at(3,2){$T_2$};\node[](i2)at(3,-0.2){$\iota_2$};
\node[](TS1)at(3,-1){{$\TS_2$}};
\node[](TT)at(1,0){$T_1\cap T_2=\emptyset$};}
\draw[very thick] (7,0) .. controls (3,4) and (11,4) .. (7,0);
\node[circle,minimum size = 2mm,fill=black!100,inner sep=0.05cm](s1)at(8.2,2)[label=right:${s}$]{};
\draw[very thick] (8.25,2) .. controls (12,-2) and (12,6) .. (8.25,2);
\node[](i1)at(7,-0.2){$\iota$};
\node[]at(8.5,-1){{$\TS_1 \triangleleft s \triangleright \TS_2$}};
\node[]at(5,-1){{$\impl$}};
\end{tikzpicture}\vspace*{-8mm}
\caption{General idea of articulations.}
\label{articul.fig}\vspace*{-4mm}
\end{figure}
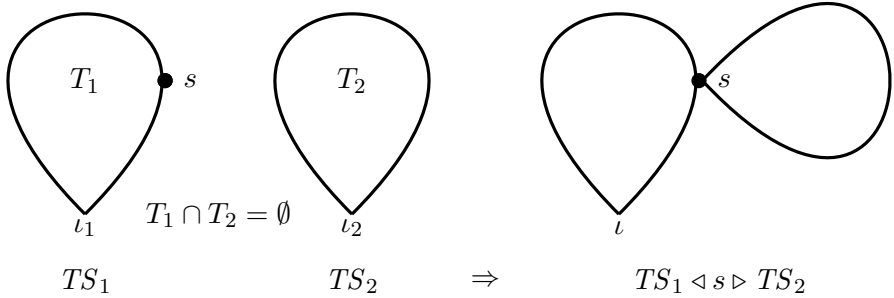

\begin{definition}{\sc Articulation of two disjoint \lts} \label{articul.def}\\
Let $\TS_1=(S_1,\to_1,T_1,\is_1)$ and $\TS_2=(S_2,\to_2,T_2,\is_2)$ be two (totally reachable and deterministic) LTSs
with $T_1\cap T_2=\emptyset$ and $s\in S_1$.
Thanks to isomorphisms we may assume that $S_1\cap S_2=\{s\}$ and $\is_2=s$. We shall then denote by $\TS_1 \triangleleft s \triangleright \TS_2=
(S_1\cup S_2,T_1\cup T_2,\to_1\cup\to_2,\is_1)$ the {\it articulation} of $\TS_1$ and $\TS_2$ around $s$.

\medskip
Conversely, let $\TS=(S,\to,T,\is)$ be a  (totally reachable and deterministic) LTS.
We shall say that a label $t$ is useful  if $\exists s,s'\in S:\; s[t\rangle s'$.
 Let $\emptyset\subseteq T_1\subseteq T$; we shall then denote by
 $\adj(T_1)=\{s\in S|\exists t\in T_1: s[t\rangle\mbox{ or }[t\rangle s\}
     \mbox{ if there are useful labels in }T_1,\{\is\}\mbox{ otherwise}$.
This is the {\it adjacency set} of $T_1$, i.e., the set of states connected to $T_1$
(with the convention that, if $T_1$ is empty or only contains useless labels, the result is the singleton initial state).
Let $T_2=T\setminus T_1$ and $s\in S$.
We shall say that $\TS$ is {\it articulated}\,\footnote{This notion
 has some similarity with the cut vertices (or articulation points) introduced for connected unlabelled undirected graphs,
 whose removal disconnects the graph. They have been used for instance
to decompose such graphs into biconnected components~\cite{HopTar73,WestTar92}.
Note however that here we have labelled graphs.} by $T_1$ and $T_2$ around $s$
 if  $\adj(T_1)\cap\adj(T_2)=\{s\}$, $\forall s_1\in\adj(T_1)\exists\alpha_1\in T_1^*:\is[\alpha_1\rangle s_1$ and
$\forall s_2\in\adj(T_2)\exists\alpha_2\in T_2^*:s[\alpha_2\rangle s_2$. \QED
\end{definition} 

This operator is only defined up to isomorphism since we may need to rename the state sets
(usually the right one, but we may also rename the left one, or both).
The only constraint is that, after the relabellings, $s$ is the unique common state of $\TS_1$ and $\TS_2$,
and is the state where the two systems are to be articulated.
Figure~\ref{art.fig} illustrates this operator. It also shows that the articulation highly relates
on the state around which the articulation takes part.
It may also be observed that, if $\TS_0=(\{\is\},\emptyset,\emptyset,\is)$ is the trivial empty LTS,
we have that, for any state $s$ of $\TS$,  $\TS\triangleleft s \triangleright \TS_0\equiv \TS$, i.e.,
we have a kind of right neutral trivial articulation.
Similarly, $\TS_0\triangleleft \is \triangleright \TS\equiv \TS$, i.e.,
we have a kind of left neutral trivial articulation. However, these neutrals  will play no role in the following of this paper,
so that we shall exclude them from our considerations (and assume the edge label sets to be non-empty,
and only composed of useful labels).

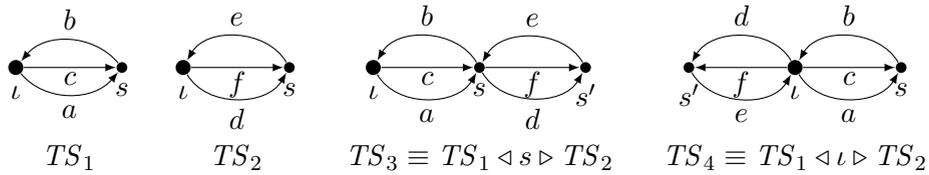
\begin{figure}[hbt]
\refstepcounter{exampleno}\label{ex1.ts}
\begin{center}
\begin{tikzpicture}[scale=0.7]\vspace*{-2mm}
\node[]at(1,-1.7){$\TS_{\ref{ex1.ts}}$};
\node[circle,fill=black!100,inner sep=0.07cm](s0)at(0,0)[label=below:$\is$]{};
\node[circle,fill=black!100,inner sep=0.05cm](s1)at(2,0)[label=below:$s$]{};
\draw[-latex](s0)--node[auto,swap,inner sep=1.5pt,pos=0.5]{$c$}(s1);
\draw[-latex](s0)to[out=-50,in=-130]node[auto,swap]{$a$}(s1);
\draw[-latex](s1)to[out=130,in=50]node[auto,swap]{$b$}(s0);
\end{tikzpicture}\hspace{0.2cm}
\refstepcounter{exampleno}\label{ex2.ts}
\begin{tikzpicture}[scale=0.7]\vspace*{-2mm}
\node[]at(1,-1.7){$\TS_{\ref{ex2.ts}}$};
\node[circle,fill=black!100,inner sep=0.07cm](s0)at(0,0)[label=below:$\is$]{};
\node[circle,fill=black!100,inner sep=0.05cm](s1)at(2,0)[label=below:$s$]{};
\draw[-latex](s0)--node[auto,swap,inner sep=1.5pt,pos=0.5]{$f$}(s1);
\draw[-latex](s0)to[out=-60,in=-120]node[auto,swap]{$d$}(s1);
\draw[-latex](s1)to[out=120,in=60]node[auto,swap]{$e$}(s0);
\end{tikzpicture}\hspace{0.2cm}
\refstepcounter{exampleno}\label{ex3.ts}
\begin{tikzpicture}[scale=0.7]\vspace*{-2mm}
\node[]at(2,-1.7){$\TS_{\ref{ex3.ts}}\equiv\TS_{\ref{ex1.ts}} \triangleleft s \triangleright \TS_{\ref{ex2.ts}}$};
\node[circle,fill=black!100,inner sep=0.07cm](s0)at(0,0)[label=below:$\is$]{};
\node[circle,fill=black!100,inner sep=0.05cm](s1)at(2,0)[label=below:$s$]{};
\node[circle,fill=black!100,inner sep=0.05cm](s2)at(4,0)[label=below:$s'$]{};
\draw[-latex](s0)--node[auto,swap,inner sep=1.5pt,pos=0.5]{$c$}(s1);
\draw[-latex](s0)to[out=-60,in=-120]node[auto,swap]{$a$}(s1);
\draw[-latex](s1)to[out=120,in=60]node[auto,swap]{$b$}(s0);
\draw[-latex](s1)--node[auto,swap,inner sep=1.5pt,pos=0.5]{$f$}(s2);
\draw[-latex](s1)to[out=-60,in=-120]node[auto,swap]{$d$}(s2);
\draw[-latex](s2)to[out=120,in=60]node[auto,swap]{$e$}(s1);
\end{tikzpicture}\hspace{0.2cm}
\refstepcounter{exampleno}\label{ex4.ts}
\begin{tikzpicture}[scale=0.7]\vspace*{-2mm}
\node[]at(0,-1.7){$\TS_{\ref{ex4.ts}}\equiv\TS_{\ref{ex1.ts}} \triangleleft \is \triangleright \TS_{\ref{ex2.ts}}$};
\node[circle,fill=black!100,inner sep=0.07cm](s0)at(0,0)[label=below:$\is$]{};
\node[circle,fill=black!100,inner sep=0.05cm](s1)at(2,0)[label=below:$s$]{};
\node[circle,fill=black!100,inner sep=0.05cm](s2)at(-2,0)[label=below:$s'$]{};
\draw[-latex](s0)--node[auto,swap,inner sep=1.5pt,pos=0.5]{$c$}(s1);
\draw[-latex](s0)to[out=-60,in=-120]node[auto,swap]{$a$}(s1);
\draw[-latex](s1)to[out=120,in=60]node[auto,swap]{$b$}(s0);
\draw[-latex](s0)--node[auto,swap,inner sep=1.5pt,pos=0.5,below]{$f$}(s2);
\draw[-latex](s0)to[out=120,in=60]node[auto,swap]{$d$}(s2);
\draw[-latex](s2)to[out=-60,in=-120]node[auto,swap]{$e$}(s0);
\end{tikzpicture}
\end{center}\vspace*{-6mm}
\caption{Some articulations.}
\label{art.fig}
\end{figure}

Several easy but interesting properties may be derived for this articulation operator \cite{RD-articul-PN}.

\begin{proposition}\label{articul.prop}{\bf (Both forms of articulation are equivalent)}\\
If $\TS=(S,\to,T,\is)$ is articulated by $T_1$ and $T_2$ around $s$, then with $\to_1=\to\cap \adj(T_1)\times T_1\times\adj(T_1)$
(i.e., the  restriction of $\to$ to $T_1$) and similarly for $\to_2$,
the structures $\TS_1=(\adj(T_1),\to_1,T_1,\is)$ and $\TS_2=(\adj(T_2),\to_2,T_2,s)$ are
totally reachable LTSs and $\TS\equiv_{T_1\uplus T_2}\TS_1 \triangleleft s \triangleright \TS_2$
(in that case we do not need to apply a relabelling to $\TS_1$ and $\TS_2$).

\medskip
Conversely, $\TS_1 \triangleleft s \triangleright \TS_2$ is articulated by the label sets of $\TS_1$ and $\TS_2$ around~$s$.
\end{proposition}

\begin{proposition}\label{art.prop}{\bf (Evolutions of an articulation)}\\
If $\TS\equiv\TS_1 \triangleleft s \triangleright \TS_2$,
$\is[\alpha\rangle s'$ is an evolution of $\TS$ iff  it is an alternation of evolutions of $\TS_1$ and $\TS_2$
separated by occurrences of $s$, i.e., either $\alpha\in T_1^*$ or
$\alpha=\alpha_1\alpha_2\ldots\alpha_n$ such that $\alpha_i\in T_1^*$ if $i$ is odd,  $\alpha_i\in T_2^*$ if $i$ is even,
$\is[\alpha_1\rangle s$ and $\forall i\in\{1,2,\ldots,n-1\}:[\alpha_i\rangle s[\alpha_{i+1}\rangle$.
\end{proposition} 

For instance, for $\TS_{\ref{ex3.ts}}$ in Figure~\ref{art.fig}, a possible evolution is $\is[abc\rangle s[f\hspace{-0.1em}ede\rangle s[b\rangle \is$,
but also equivalently $\is[a\rangle s[\leer\rangle s[bc\rangle s[f\hspace{-0.1em}e\rangle s[\leer\rangle s[de\rangle s[b\rangle \is$
(where $\leer$ is the empty sequence).

\begin{proposition}\label{assoc.prop}{\bf (Associativity of articulations)}\\
Let us assume that $\TS_1$, $\TS_2$ and $\TS_3$ are three LTSs with label sets $T_1$, $T_2$ and $T_3$ respectively,
pairwise disjoint. Let $s_1$ be a state of $\TS_1$ and $s_2$ be a state of $\TS_2$.
Then, $\TS_1\triangleleft s_1\triangleright(\TS_2\triangleleft s_2\triangleright \TS_3)\equiv_{T_1\cup T_2\cup T_3}
(\TS_1\triangleleft s_1\triangleright\TS_2)\triangleleft s'_2\triangleright \TS_3$,
where $s'_2$ corresponds in $\TS_1\triangleleft s_1\triangleright\TS_2$ to $s_2$ in $\TS_2$
(let us recall that the articulation operator may rename the states of the second operand).
\end{proposition}

This is illustrated by Figure~\ref{artass.fig}.

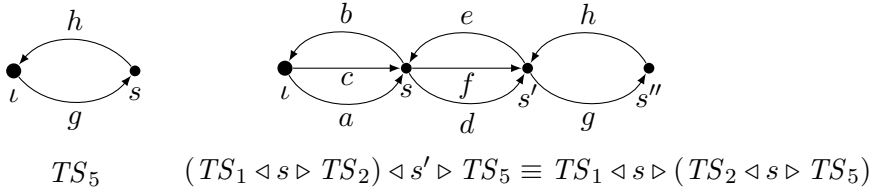
\begin{figure}[hbt]
\vspace*{-4mm}
\refstepcounter{exampleno}\label{ex5.ts}
\begin{center}
\begin{tikzpicture}[scale=0.8]
\node[]at(1,-1.7){$\TS_{\ref{ex5.ts}}$};
\node[circle,fill=black!100,inner sep=0.07cm](s0)at(0,0)[label=below:$\is$]{};
\node[circle,fill=black!100,inner sep=0.05cm](s1)at(2,0)[label=below:$s$]{};
\draw[-latex](s0)to[out=-50,in=-130]node[auto,swap]{$g$}(s1);
\draw[-latex](s1)to[out=130,in=50]node[auto,swap]{$h$}(s0);
\end{tikzpicture}\hspace{0.2cm}
\begin{tikzpicture}[scale=0.8]
\node[]at(4,-1.7){$(\TS_{\ref{ex1.ts}} \triangleleft s \triangleright \TS_{\ref{ex2.ts}})\triangleleft s' \triangleright \TS_{\ref{ex5.ts}}
\equiv \TS_{\ref{ex1.ts}} \triangleleft s \triangleright (\TS_{\ref{ex2.ts}}\triangleleft s \triangleright \TS_{\ref{ex5.ts}})$};
\node[circle,fill=black!100,inner sep=0.07cm](s0)at(0,0)[label=below:$\is$]{};
\node[circle,fill=black!100,inner sep=0.05cm](s1)at(2,0)[label=below:$s$]{};
\node[circle,fill=black!100,inner sep=0.05cm](s2)at(4,0)[label=below:$s'$]{};
\node[circle,fill=black!100,inner sep=0.05cm](s3)at(6,0)[label=below:$s''$]{};
\draw[-latex](s0)--node[auto,swap,inner sep=1.5pt,pos=0.5]{$c$}(s1);
\draw[-latex](s0)to[out=-60,in=-120]node[auto,swap]{$a$}(s1);
\draw[-latex](s1)to[out=120,in=60]node[auto,swap]{$b$}(s0);
\draw[-latex](s1)--node[auto,swap,inner sep=1.5pt,pos=0.5]{$f$}(s2);
\draw[-latex](s1)to[out=-60,in=-120]node[auto,swap]{$d$}(s2);
\draw[-latex](s2)to[out=120,in=60]node[auto,swap]{$e$}(s1);
\draw[-latex](s2)to[out=-60,in=-120]node[auto,swap]{$g$}(s3);
\draw[-latex](s3)to[out=120,in=60]node[auto,swap]{$h$}(s2);
\end{tikzpicture}
\end{center}\vspace*{-7mm}
\caption{Associativity of  articulations.}
\label{artass.fig}\vspace*{-2mm}
\end{figure}

\begin{proposition}\label{com.prop}{\bf  (Commutative articulations)}\\
If $\TS_1=(S_1,\to_1,T,\is_{1})$  and $\TS_2=(S_2,\to_2,T,\is_{2})$  with disjoint label sets (i.e., $T_1\cap T_2=\emptyset$),
then $\TS_1 \triangleleft \is_1 \triangleright \TS_2\equiv_{T_1\cup T_2}\TS_2 \triangleleft \is_2 \triangleright \TS_1$.
\end{proposition}

Note that, in the left member of the equivalence, we must rename $\is_2$ into $\is_1$,
and in the right member we must rename $\is_1$ into $\is_2$, in order to apply Definition~\ref{articul.def} (defined up to isomorphisms).
For instance, in Figure \ref{art.fig}, $\TS_{\ref{ex4.ts}}\equiv\TS_{\ref{ex1.ts}} \triangleleft \is \triangleright \TS_{\ref{ex2.ts}}\equiv \TS_{\ref{ex2.ts}} \triangleleft \is \triangleright \TS_{\ref{ex1.ts}}$.

\begin{proposition}\label{comassoc.prop}{\bf (Commutative associativity of articulations)}\\
Let us assume that $\TS_1$, $\TS_2$ and $\TS_3$ are three LTSs with label sets $T_1$, $T_2$ and $T_3$ respectively,
pairwise disjoint. Let $s_2$ and $s_3$ be two  states of $\TS_1$ ($s_2=s_3$ is allowed).
Then, $(\TS_1\triangleleft s_2\triangleright\TS_2)\triangleleft s_3\triangleright \TS_3\equiv_{T_1\cup T_2\cup T_3}
(\TS_1\triangleleft s_3\triangleright\TS_3)\triangleleft s_2\triangleright \TS_2$.
\end{proposition}

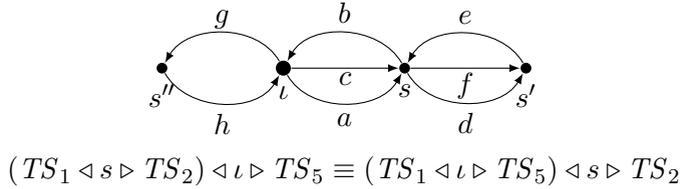
\begin{figure}[hbt]
\vspace*{-4mm}
\begin{center}
\begin{tikzpicture}[scale=0.8]
\node[]at(1,-1.7){$(\TS_{\ref{ex1.ts}} \triangleleft s \triangleright \TS_{\ref{ex2.ts}}) \triangleleft \is \triangleright
 \TS_{\ref{ex5.ts}}\equiv(\TS_{\ref{ex1.ts}} \triangleleft \is \triangleright \TS_{\ref{ex5.ts}}) \triangleleft s \triangleright
  \TS_{\ref{ex2.ts}}$};\vspace*{-2mm}
\node[circle,fill=black!100,inner sep=0.07cm](s0)at(0,0)[label=below:$\is$]{};
\node[circle,fill=black!100,inner sep=0.05cm](s1)at(2,0)[label=below:$s$]{};
\node[circle,fill=black!100,inner sep=0.05cm](s2)at(4,0)[label=below:$s'$]{};
\draw[-latex](s0)--node[auto,swap,inner sep=1.5pt,pos=0.5]{$c$}(s1);
\draw[-latex](s0)to[out=-60,in=-120]node[auto,swap]{$a$}(s1);
\draw[-latex](s1)to[out=120,in=60]node[auto,swap]{$b$}(s0);
\draw[-latex](s1)--node[auto,swap,inner sep=1.5pt,pos=0.5]{$f$}(s2);
\draw[-latex](s1)to[out=-60,in=-120]node[auto,swap]{$d$}(s2);
\draw[-latex](s2)to[out=120,in=60]node[auto,swap]{$e$}(s1);
\node[circle,fill=black!100,inner sep=0.05cm](s3)at(-2,0)[label=below:$s''$]{};
\draw[-latex](s0)to[out=120,in=60]node[auto,swap,inner sep=1.5pt,pos=0.5]{$g$}(s3);
\draw[-latex](s3)to[out=-60,in=-120]node[auto,swap]{$h$}(s0);
\end{tikzpicture}
\end{center}\vspace*{-7mm}
\caption{Commutative associativity of articulations.}
\label{comassart.fig}\vspace*{-2mm}
\end{figure}

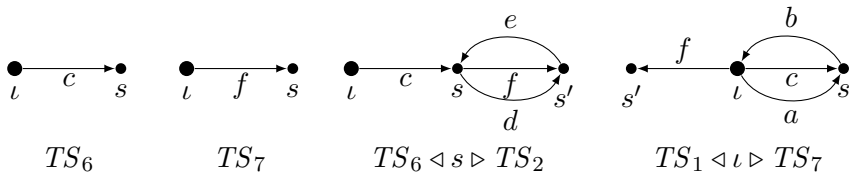
\begin{figure}[!b]
\vspace*{-4mm}
\refstepcounter{exampleno}\label{ex6.ts}
\begin{center}
\begin{tikzpicture}[scale=0.7]
\node[]at(1,-1.7){$\TS_{\ref{ex6.ts}}$};
\node[circle,fill=black!100,inner sep=0.07cm](s0)at(0,0)[label=below:$\is$]{};
\node[circle,fill=black!100,inner sep=0.05cm](s1)at(2,0)[label=below:$s$]{};
\draw[-latex](s0)--node[auto,swap,inner sep=1.5pt,pos=0.5]{$c$}(s1);
\end{tikzpicture}\hspace{0.26cm}
\refstepcounter{exampleno}\label{ex7.ts}
\begin{tikzpicture}[scale=0.7]
\node[]at(1,-1.7){$\TS_{\ref{ex7.ts}}$};
\node[circle,fill=black!100,inner sep=0.07cm](s0)at(0,0)[label=below:$\is$]{};
\node[circle,fill=black!100,inner sep=0.05cm](s1)at(2,0)[label=below:$s$]{};
\draw[-latex](s0)--node[auto,swap,inner sep=1.5pt,pos=0.5]{$f$}(s1);
\end{tikzpicture}\hspace{0.26cm}
\begin{tikzpicture}[scale=0.7]
\node[]at(2,-1.7){$\TS_{\ref{ex6.ts}} \triangleleft s \triangleright \TS_{\ref{ex2.ts}}$};
\node[circle,fill=black!100,inner sep=0.07cm](s0)at(0,0)[label=below:$\is$]{};
\node[circle,fill=black!100,inner sep=0.05cm](s1)at(2,0)[label=below:$s$]{};
\node[circle,fill=black!100,inner sep=0.05cm](s2)at(4,0)[label=below:$s'$]{};
\draw[-latex](s0)--node[auto,swap,inner sep=1.5pt,pos=0.5]{$c$}(s1);
\draw[-latex](s1)--node[auto,swap,inner sep=1.5pt,pos=0.5]{$f$}(s2);
\draw[-latex](s1)to[out=-60,in=-120]node[auto,swap]{$d$}(s2);
\draw[-latex](s2)to[out=120,in=60]node[auto,swap]{$e$}(s1);
\end{tikzpicture}\hspace{0.26cm}
\begin{tikzpicture}[scale=0.7]
\node[]at(0,-1.7){$\TS_{\ref{ex1.ts}} \triangleleft \is \triangleright \TS_{\ref{ex7.ts}}$};
\node[circle,fill=black!100,inner sep=0.07cm](s0)at(0,0)[label=below:$\is$]{};
\node[circle,fill=black!100,inner sep=0.05cm](s1)at(2,0)[label=below:$s$]{};
\node[circle,fill=black!100,inner sep=0.05cm](s2)at(-2,0)[label=below:$s'$]{};
\draw[-latex](s0)--node[auto,swap,inner sep=1.5pt,pos=0.5]{$c$}(s1);
\draw[-latex](s0)to[out=-60,in=-120]node[auto,swap]{$a$}(s1);
\draw[-latex](s1)to[out=120,in=60]node[auto,swap]{$b$}(s0);
\draw[-latex](s0)--node[auto,swap,inner sep=1.5pt,pos=0.5]{$f$}(s2);
\end{tikzpicture}
\end{center}\vspace*{-7mm}
\caption{Sequential articulations.}
\label{seqart.fig}
\end{figure}

\begin{proposition}\label{seq.prop}{\bf (Sequence articulations)}\\
If $\TS_1=(S_1,\to_1,T,\is_{1})$  and $\TS_2=(S_2,\to_2,T,\is_{2})$  with disjoint label sets (i.e., $T_1\cap T_2=\emptyset$),
if $\forall s_1\in S_1\exists\alpha_1\in T_1^*:s_1[\alpha_1\rangle s$ ($s$ is a {\it home state in $\TS_1$})
and $\nexists t_1\in T_1:s[t_1\rangle$ ($s$ is a {\it dead end in $\TS_1$}),\\
then $\TS_1 \!\triangleleft s \triangleright \!\TS_2$ behaves like a sequence, i.e.,
once $\TS_2$ has started it is no longer possible to execute~$T_1$.

The same occurs when $\is_2$ does not occur in a non-trivial cycle, i.e.,
$\is_2[\alpha_2\rangle\is_2\land\alpha_2\in T_2^*\impl\alpha_2=\leer$:
once $\TS_2$ has started it is no longer possible to execute $T_1$.
\end{proposition}

This is illustrated in Figure~\ref{seqart.fig}.
It may be observed that sequences in \cite{BDK02} (and in Figure~\ref{seq.fig}) correspond to the intersection of both cases.

\medskip
Let us now examine the connection between articulations and Petri net synthesis.

\begin{proposition}\label{art1.prop}{\bf (Synthesis of components of an articulation)}\\
If $\TS=(S,\to,T_1\uplus T_2,\is)$ is articulated by $T_1$ and $T_2$ around $s$,
so that $\TS\equiv \TS_1 \triangleleft s \triangleright \TS_2$ with $\TS_1=(\adj(T_1),\to_1,T_1,\is)$ and $\TS_2=(\adj(T_2),\to_2,s)$ (see Proposition~\ref{articul.prop}),
and is PN-solvable, component $\TS_1$ and $\TS_2$ are also PN-solvable.
Moreover, in the corresponding solution for $\TS_1$, if the decomposition is not trivial,
the marking corresponding to $s$ is not dominated by any other reachable marking.
\end{proposition}

\begin{proof}
Let $ N=(P,T,F,M_0)$ be a solution for $\TS$.
It is immediate that  $ N_1=(P,T_1,F_1,M_0)$, where $F_1$ is the restriction of $F$ to $T_1$, is a solution for $\TS_1$
(but there may be many other ones).

\medskip
Similarly, if $M$ is the marking of $ N$ (and $ N_1$) corresponding to $s$, it may be seen that  $ N_2=(P,T_2,F_2,M)$, where $F_2$ is the restriction of $F$ to $T_2$, is a solution for $\TS_2$
(but there may be many other ones).

Moreover, if the decomposition is not trivial, $T_2\neq\emptyset$.
 Let us thus assume that $s[t_2\rangle$ for some label $t_2\in T_2$ and
  $M'$ is a marking of $ N_1$ corresponding to some state $s'$ in $\TS_1$
with $M'\gneqq M$, then $s\neq s'$, $s'[t_2\rangle$ and $s$ is not the unique articulation between $T_1$ and $T_2$.
\end{proof} 

Note that there may also be solutions to $\TS_1$ (other than $N_1$) such that the marking $M$ corresponding to $s$ is dominated.
This is illustrated in Figure~\ref{art1.fig}.

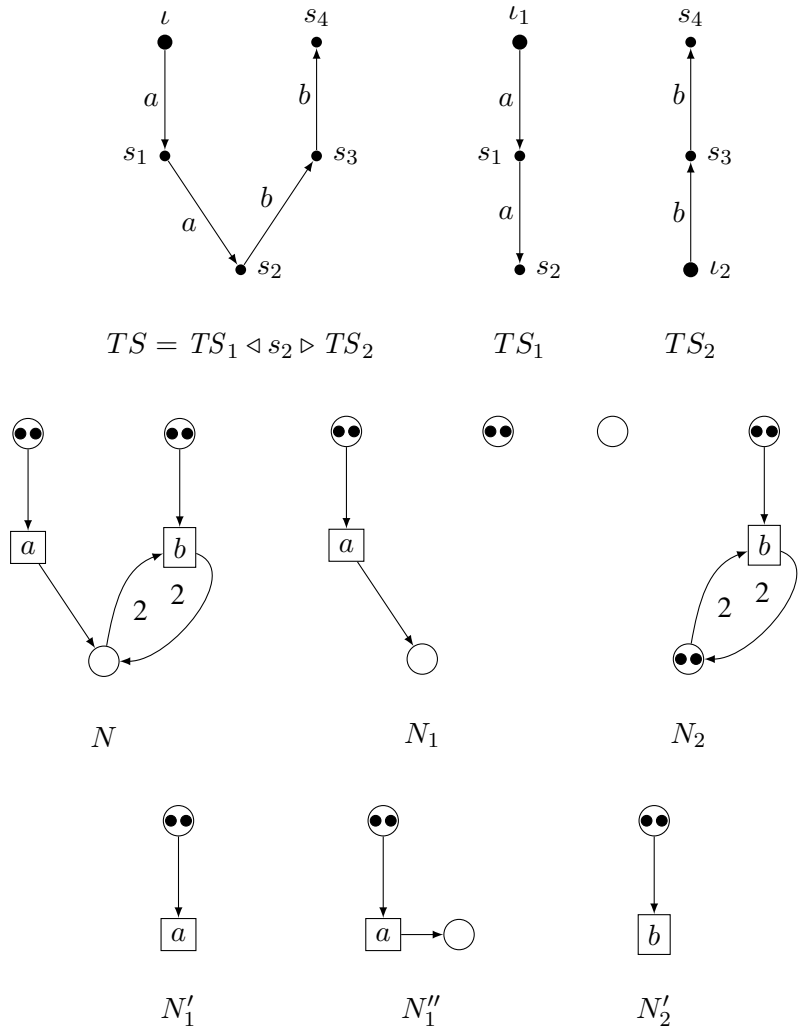
\begin{figure}[!ht]
\begin{center}
\begin{tikzpicture}[scale=1.0]
\node[circle,fill=black!100,inner sep=0.07cm](s0)at(-1,3)[label=above:$\is$]{};
\node[circle,fill=black!100,inner sep=0.05cm](s1)at(-1,1.5)[label=left:$s_1$]{};
\node[circle,fill=black!100,inner sep=0.05cm](s2)at(0,0)[label=right:$s_2$]{};
\node[circle,fill=black!100,inner sep=0.05cm](s3)at(1,1.5)[label=right:$s_3$]{};
\node[circle,fill=black!100,inner sep=0.05cm](s4)at(1,3)[label=above:$s_4$]{};
\draw[-latex](s0)--node[auto,swap,inner sep=2pt,pos=0.5]{$a$}(s1);
\draw[-latex](s1)--node[auto,swap,inner sep=2pt,pos=0.5]{$a$}(s2);
\draw[-latex](s2)--node[auto,inner sep=2pt,pos=0.5]{$b$}(s3);
\draw[-latex](s3)--node[auto,inner sep=2pt,pos=0.55]{$b$}(s4);
\node[]at(0,-1){$TS=\TS_1 \triangleleft s_2 \triangleright \TS_2$};
\end{tikzpicture}\hspace*{1.0cm}
\begin{tikzpicture}[scale=1.0]
\node[circle,fill=black!100,inner sep=0.07cm](s0)at(0,3)[label=above:$\is_1$]{};
\node[circle,fill=black!100,inner sep=0.05cm](s1)at(0,1.5)[label=left:$s_1$]{};
\node[circle,fill=black!100,inner sep=0.05cm](s2)at(0,0)[label=right:$s_2$]{};
\draw[-latex](s0)--node[auto,swap,inner sep=2pt,pos=0.5]{$a$}(s1);
\draw[-latex](s1)--node[auto,swap,inner sep=2pt,pos=0.5]{$a$}(s2);
\node[]at(0,-1){$TS_1$};
\end{tikzpicture}\hspace*{1.0cm}
\begin{tikzpicture}[scale=1.0]
\node[circle,fill=black!100,inner sep=0.07cm](s2)at(0,0)[label=right:$\is_2$]{};
\node[circle,fill=black!100,inner sep=0.05cm](s3)at(0,1.5)[label=right:$s_3$]{};
\node[circle,fill=black!100,inner sep=0.05cm](s4)at(0,3)[label=above:$s_4$]{};
\draw[-latex](s2)--node[auto,inner sep=2pt,pos=0.5]{$b$}(s3);
\draw[-latex](s3)--node[auto,inner sep=2pt,pos=0.55]{$b$}(s4);
\node[]at(0,-1){$TS_2$};
\end{tikzpicture}\\[0.6cm]
\begin{tikzpicture}[scale=1.0]
\node[draw,minimum size=0.4cm](a)at(0,1.5){$a$};
\node[draw,minimum size=0.4cm](b)at(2,1.5){$b$};
\node[circle,draw,minimum size=0.4cm](p0)at(0,3)[]{};\filldraw[black](-0.1,3)circle(2pt);\filldraw[black](0.1,3)circle(2pt);
\node[circle,draw,minimum size=0.4cm](p1)at(1,0)[]{};
\node[circle,draw,minimum size=0.4cm](p2)at(2,3)[]{};\filldraw[black](1.9,3)circle(2pt);\filldraw[black](2.1,3)circle(2pt);
\draw[-latex](p0)--(a);\draw[-latex](a)--(p1);
\draw[-latex](p2)--(b);
\draw[-latex](b)to[out=-20,in=0]node[auto,swap]{2}(p1);
\draw[-latex](p1)to[out=80,in=200]node[auto,swap]{2}(b);
\node[]at(1,-1){$N$};
\end{tikzpicture}\hspace*{1.0cm}
\begin{tikzpicture}[scale=1.0]
\node[draw,minimum size=0.4cm](a)at(0,1.5){$a$};
\node[circle,draw,minimum size=0.4cm](p0)at(0,3)[]{};\filldraw[black](-0.1,3)circle(2pt);\filldraw[black](0.1,3)circle(2pt);
\node[circle,draw,minimum size=0.4cm](p1)at(1,0)[]{};
\node[circle,draw,minimum size=0.4cm](p2)at(2,3)[]{};\filldraw[black](1.9,3)circle(2pt);\filldraw[black](2.1,3)circle(2pt);
\draw[-latex](p0)--(a);\draw[-latex](a)--(p1);
\node[]at(1,-1){$N_1$};
\end{tikzpicture}\hspace*{1.0cm}
\begin{tikzpicture}[scale=1.0]
\node[draw,minimum size=0.4cm](b)at(2,1.5){$b$};
\node[circle,draw,minimum size=0.4cm](p0)at(0,3)[]{};
\node[circle,draw,minimum size=0.4cm](p1)at(1,0)[]{};\filldraw[black](0.9,0)circle(2pt);\filldraw[black](1.1,0)circle(2pt);
\node[circle,draw,minimum size=0.4cm](p2)at(2,3)[]{};\filldraw[black](1.9,3)circle(2pt);\filldraw[black](2.1,3)circle(2pt);
\draw[-latex](p2)--(b);
\draw[-latex](b)to[out=-20,in=0]node[auto,swap]{2}(p1);
\draw[-latex](p1)to[out=80,in=200]node[auto,swap]{2}(b);
\node[]at(1,-1){$N_2$};
\end{tikzpicture}\\[0.6cm]
\begin{tikzpicture}[scale=1.0]
\node[draw,minimum size=0.4cm](a)at(0,1.5){$a$};
\node[circle,draw,minimum size=0.4cm](p0)at(0,3)[]{};\filldraw[black](-0.1,3)circle(2pt);\filldraw[black](0.1,3)circle(2pt);
\draw[-latex](p0)--(a);
\node[]at(0,0.5){$N'_1$};
\end{tikzpicture}\hspace*{2.0cm}
\begin{tikzpicture}[scale=1.0]
\node[draw,minimum size=0.4cm](a)at(0,1.5){$a$};
\node[circle,draw,minimum size=0.4cm](p0)at(0,3)[]{};\filldraw[black](-0.1,3)circle(2pt);\filldraw[black](0.1,3)circle(2pt);
\node[circle,draw,minimum size=0.4cm](p1)at(1,1.5)[]{};
\draw[-latex](p0)--(a);\draw[-latex](a)--(p1);
\node[]at(0.5,0.5){$N''_1$};
\end{tikzpicture}\hspace*{2.0cm}\begin{tikzpicture}[scale=1.0]
\node[draw,minimum size=0.4cm](b)at(2,1.5){$b$};
\node[circle,draw,minimum size=0.4cm](p2)at(2,3)[]{};\filldraw[black](1.9,3)circle(2pt);\filldraw[black](2.1,3)circle(2pt);
\draw[-latex](p2)--(b);
\node[]at(2,0.5){$N'_2$};
\end{tikzpicture}
\vspace{-1.5em}
\end{center}
\caption{The lts $\TS$ is articulated around $s_2$, with $T_1=\{a\}$ and $T_2=\{b\}$, hence leading to $\TS_1$ and $\TS_2$.
It is solved by $N$, and the corresponding solutions for $\TS_1$ and $\TS_2$ are $N_1$ and $N_2$, respectively.
$\TS_1$ also has the solution $N'_1$ but the marking corresponding to $s_2$ is then empty, hence it is dominated by the initial marking (as well as by the intermediate one). This is not the case for the other solution $N''_1$
(obtained from $N_1$ by erasing the useless  isolated place: we never claimed that $N_1$ is a minimal solution).
$\TS_2$ also has the solution $N'_2$.}
\label{art1.fig}
\end{figure}

\medskip
The other way round, let us now assume that $\TS=\TS_1 \triangleleft s \triangleright \TS_2$ is an articulated LTS
and that it is possible to solve $\TS_1$ and $\TS_2$.
Is it possible from that to build a solution of $\TS$?

To do that, we shall add the constraint already observed in Proposition~\ref{art1.prop} that,
in the solution of $\TS_1$, the marking corresponding to $s$ is not dominated by another reachable marking.
If this is satisfied we shall say that the solution is {\em adequate} with respect to $s$.
Hence, in the treatment of the system in Figure~\ref{art1.fig},
we want to avoid considering the solution $N'_1$ of $\TS_1$; on the contrary, $N_1$ or $N''_1$ will be acceptable.

\medskip
If $\TS_1$ is reversible and PN-solvable, any solution is adequate. Indeed, for any pair of distinct states $s,s'\in S_1$
we then have a path $s[\alpha\rangle s'$ with $\alpha\in T_1^*$. In the solution $\PN_1$ of $\TS_1$,
if $M$ is the marking corresponding to $s$ and $M'$ is the one corresponding to $s'$,
we have $M\neq M'$ and $M[\alpha\rangle M'$ and if $M\lneqq M'$ we also have an infinite path $M'[\alpha^\infty\rangle$.
Since  $\PN_1$ is a solution of $\TS_1$, we also have $s[\alpha^n\rangle s_i$ for an infinite series of different states
$s_i$ for $n\in\nsymbol$, and $\TS_2$ as well as $\TS$ may not be finite as we assumed in this paper.
Note that, from a similar argument, since $\TS_2$ is finite, no marking reachable in $\PN_2$ dominates the initial one,
corresponding to $\is_2=s$.

However, if $\TS_1$ is solvable, it is always possible to get a solution adequate at $s$, as for any state in fact.

\begin{proposition}\label{adeq.prop}{\bf (Adequate solutions)}\\
Let $\TS_1$ be a (finite) solvable \lts{} and $s$ any of its states.
Then there is a solution of $\TS_1$ adequate at~$s$.
\end{proposition}

\begin{proof}
Let $\PN_1$ be any solution. For any place $p\in P_1$, we may construct a complement or mirror place $\widetilde{p}$
such that $\forall t\in T_1: F(t,\widetilde{p})=F(p,t)\land F(\widetilde{p},t)=F(t,p)$ so that, for any reachable marking $\widetilde{M}$
of the new net, $\widetilde{M}(p)+\widetilde{M}(\widetilde{p})$ is constant. It remains to chose the initial marking of this place in order
not to exclude some evolutions available in $\PN_1$. Since $\TS_1$ is finite, the marking of $p$ is bounded (as for any other place).
Let $k$ be that bound and let us chose $\widetilde{M}_0(p)=k-M_0(p)+\max_{t\in T_1}F(t,p)$.
That way, for any reachable marking, $\widetilde{M}(\widetilde{p})\geq\max_{t\in T_1}F(t,p)=\max_{t\in T_1}F(\widetilde{p},t)$,
so that place $\widetilde{p}$ does not block any transition that would be enabled by place $p$.
We may thus conclude that the introduction of $\widetilde{p}$ does not change the reachability graph
and the new net is still a PN-solution of $\TS_1$. As a consequence, for each reachable marking $M$ of $\PN_1$,
if $M(p)>M_s(p)$ (where $M_s$ is the marking corresponding to $s$), $\widetilde{M}(\widetilde{p})<\widetilde{M}_s(\widetilde{p})$.
Hence, if we introduce a complement place for each place of $\PN_1$, in the new net
no reachable marking may dominate another one and the constraint mentioned above is always satisfied.
\end{proof} 

However there is a simpler way to get an adequate solution (since we know there is one), in the following way:

\begin{proposition} \label{additional.prop}{\bf (Forcing an adequate solution for $\TS_1$)}\\
Let us add to $\TS_1$ an arc $s[u\rangle s$ where $u$ is a new fresh label. Let $\TS'_1$ be the LTS so obtained.
If $\TS'_1$ is not solvable, there is no (adequate) solution. Otherwise, solve $\TS'_1$ and erase $u$ from the solution.
Let $N_1$ be the net obtained with the procedure just described: it is a solution of $\TS_1$ with the adequate property that the marking corresponding to $s$ is not dominated by another one.
\end{proposition}

\begin{proof}
If there is an adequate solution $N_1$ of $\TS_1$ (and from the reasoning above there is one iff there is a solution),
with a marking $M$ corresponding to $s$,
let us add a new transition $u$ to it with, for each place $p$ of $N_1$, $W(p,u)=M(p)= W(u,p)$:
the reachability graph of this new net is (isomorphic to) $\TS'_1$ since $u$ is enabled by
marking $M$ (or any larger one, but there is none) and does not modify the marking.
Hence, if there is no (adequate) solution of $\TS_1$, there is no solution of $\TS'_1$.

\medskip
Let us now assume there is a solution $N'_1$ of $\TS'_1$. The marking $M$ corresponding to $s$ is not dominated  otherwise there would be a loop $M'[s\rangle M'$ elsewhere in the reachability graph of $N'_1$, hence also in $\TS'_1$. Hence, dropping $u$ in $N'_1$ will lead to an adequate solution of $\TS_1$.
\end{proof} 

For instance, when applied to $\TS_1$ in Figure~\ref{art1.fig}, this will lead to $N''_1$, and not $N'_1$ ($N_1$ could also be produced, but it is likely that a `normal' synthesis tool will not construct the additional isolated place).

Now, to understand how one may generate a solution for $\TS$ from the ones obtained for $\TS_1$ and $\TS_2$,
we may first carefully examine the example illustrated in Figure~\ref{art2.fig}:
some side conditions (i.e., pairs of place-transition with arcs going both ways, with identical weights)
occur in the global solution.
This leads to the following construction.

\begin{figure}[htb]
\vspace*{-2mm}
\begin{center}
\begin{tikzpicture}[scale=1.0]
\node[circle,fill=black!100,inner sep=0.07cm](s0)at(-1,3)[label=above:$\is$]{};
\node[circle,fill=black!100,inner sep=0.05cm](s1)at(-1,1.5)[label=above:$s$]{};
\node[circle,fill=black!100,inner sep=0.05cm](s2)at(-1,0)[label=below:$s_2$]{};
\draw[-latex](s0)to[out=-20,in=20]node[right,swap,inner sep=2pt,pos=0.5]{$a$}(s1);
\draw[-latex](s1)to[out=160,in=200]node[left,swap,inner sep=2pt,pos=0.5]{$b$}(s0);
\draw[-latex](s1)to[out=-20,in=20]node[right,inner sep=2pt,pos=0.5]{$c$}(s2);
\draw[-latex](s2)to[out=160,in=200]node[left,inner sep=2pt,pos=0.55]{$d$}(s1);
\node[]at(0,-1){$TS=\TS_1 \triangleleft s \triangleright \TS_2$};
\end{tikzpicture}\hspace*{1.0cm}
\begin{tikzpicture}[scale=1.0]
\node[circle,fill=black!100,inner sep=0.07cm](s0)at(0,3)[label=above:$\is_1$]{};
\node[circle,fill=black!100,inner sep=0.05cm](s1)at(0,1.5)[label=below:$s$]{};
\draw[-latex](s0)to[out=-20,in=20]node[right,swap,inner sep=2pt,pos=0.5]{$a$}(s1);
\draw[-latex](s1)to[out=160,in=200]node[left,swap,inner sep=2pt,pos=0.5]{$b$}(s0);
\node[]at(0,-1){$TS_1$};
\end{tikzpicture}\hspace*{1.0cm}
\begin{tikzpicture}[scale=1.0]
\node[circle,fill=black!100,inner sep=0.07cm](s0)at(0,1.5)[label=above:$\is_2$]{};
\node[circle,fill=black!100,inner sep=0.05cm](s1)at(0,0)[label=below:$s_2$]{};
\draw[-latex](s0)to[out=-20,in=20]node[right,swap,inner sep=2pt,pos=0.5]{$c$}(s1);
\draw[-latex](s1)to[out=160,in=200]node[left,swap,inner sep=2pt,pos=0.5]{$d$}(s0);
\node[]at(0,-1){$TS_2$};
\end{tikzpicture}\\[1cm]
\begin{tikzpicture}[scale=1.0]
\node[draw,minimum size=0.4cm](a)at(-1,1.5){$a$};
\node[draw,minimum size=0.4cm](b)at(1,1.5){$b$};
\node[circle,draw,minimum size=0.4cm](p0)at(0,3)[]{};\filldraw[black](0,3)circle(2pt);
\node[circle,draw,minimum size=0.4cm](p1)at(0,0)[]{};
\draw[-latex](p0)--(a);\draw[-latex](a)--(p1);\draw[-latex](p1)--(b);\draw[-latex](b)--(p0);
\node[]at(0,-1){$N_1$};
\end{tikzpicture}\hspace*{1.0cm}
\begin{tikzpicture}[scale=1.0]
\node[draw,minimum size=0.4cm](a)at(-1,1.5){$c$};
\node[draw,minimum size=0.4cm](b)at(1,1.5){$d$};
\node[circle,draw,minimum size=0.4cm](p0)at(0,3)[]{};\filldraw[black](0,3)circle(2pt);
\node[circle,draw,minimum size=0.4cm](p1)at(0,0)[]{};
\draw[-latex](p0)--(a);\draw[-latex](a)--(p1);\draw[-latex](p1)--(b);\draw[-latex](b)--(p0);
\node[]at(0,-1){$N_2$};
\end{tikzpicture}\hspace*{1.0cm}
\begin{tikzpicture}[scale=1.0]
\node[draw,minimum size=0.4cm](a)at(0,0){$a$};
\node[draw,minimum size=0.4cm](b)at(1.5,0.75){$b$};
\node[draw,minimum size=0.4cm](c)at(0,3){$c$};
\node[draw,minimum size=0.4cm](d)at(3,1.5){$d$};
\node[circle,draw,minimum size=0.4cm](p0)at(0,1.5)[]{};
\node[circle,draw,minimum size=0.4cm](p2)at(3,3)[]{};
\node[circle,draw,minimum size=0.4cm](p1)at(1.5,2.25)[]{};\filldraw[black](1.5,2.25)circle(2pt);
\node[circle,draw,minimum size=0.4cm](p3)at(3,0)[]{};\filldraw[black](3,0)circle(2pt);
\draw[-latex](a)--(p0);\draw[-latex](p3)--(a);
\draw[-latex](p0)--(b);\draw[-latex](b)--(p3);
\draw[-latex](c)--(p2);\draw[-latex](p2)--(d);
\draw[-latex](p1)--(c);\draw[-latex](d)--(p1);
\draw[-latex](c)to[out=-60,in=60](p0);\draw[-latex](p0)to[out=120,in=-120](c);
\draw[-latex](p1)to[out=-60,in=60](b);\draw[-latex](b)to[out=120,in=-120](p1);
\node[]at(1.5,-1){$N$};
\end{tikzpicture}
\vspace{-1.5em}
\end{center}
\caption{The lts $\TS$ is articulated around $s$, with $T_1=\{a,b\}$ and $T_2=\{c,d\}$,
hence leading to $\TS_1$ and $\TS_2$.
It is solved by $N$, and the corresponding solutions for $\TS_1$ and $\TS_2$ are $N_1$ and $N_2$, respectively.
In $N$, we may recognise $N_1$ and $N_2$, connected by two kinds of side conditions:
the first one connects the label $b$ out of $s$ in $\TS_1$ to the initial marking of $N_2$,
the other one connects the label $c$ out of $\is_2$ in $\TS_2$ to the marking of $N_1$ corresponding to $s$.}
\label{art2.fig}\vspace*{-1mm}
\end{figure}
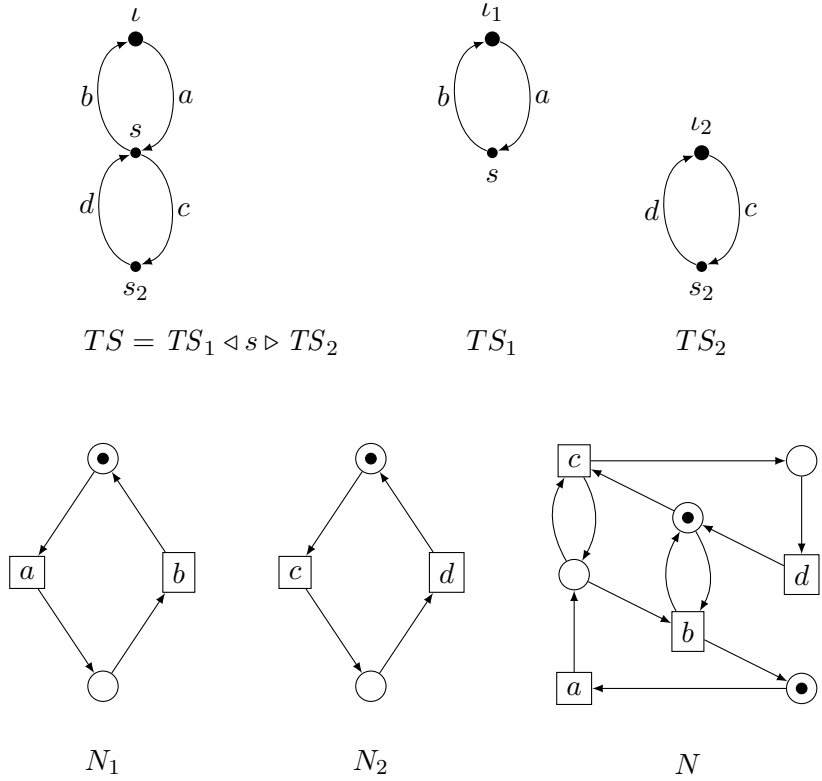

\begin{definition}{\sc Articulation of Petri Nets}\label{articulPN.def}\\
Let $\PN_1=(P_1,T_1,F_1,M_0^1)$ and $\PN_2=(P_2,T_2,F_2,M_0^2)$ be two disjoint bounded Petri net systems
and $M$ a reachable marking of $\PN_1$ not dominated by another one.
$\PN_1\triangleleft M \triangleright\PN_2$ is the Petri net built first by putting side by side $\PN_1$ and $\PN_2$.

\medskip
Then, for each transition $t_1$ enabled at $M$ in $\PN_1$, and each place $p_2\in P_2$ such that $M_0^2(p_2)>0$,
create a side condition $F(t_1,p_2)=F(p_2,t_1)=M(p_2)$.
For each transition $t_2$ initially enabled in $\PN_2$, and each place $p_1\in P_1$ such that $M(p_1)>0$,
create a side condition $F(t_2,p_1)=F(p_1,t_2)=M(p_1)$.

\medskip
No other modification is afforded.\QED
\end{definition} 

\begin{proposition}\label{constr.prop}{\bf (Synthesis of articulation)}\\
 Let $\TS=\TS_1\triangleleft s \triangleright\TS_2$.
 If $\TS_1$ or $\TS_2$ are not solvable, so is $\TS$.\\
Otherwise, let $\PN_1$ be a solution of $\TS_1$ adequate at $s$,
i.e., such that the marking $M_s$ corresponding to $s$ is not dominated by another reachable marking,
and let $\PN_2$ be a disjoint solution of $\TS_2$.
Then the net $\PN_1\triangleleft M_s \triangleright\PN_2$ is a solution of $\TS$.
\end{proposition}

\begin{proof}
The property arises from Proposition~\ref{art1.prop} and
the observation that $\PN_1$ with the additional side conditions behaves like the original $\PN_1$
if $\PN_2$ is in $M_0^2$ and $\PN_2$ with the additional side conditions behaves like the original $\PN_2$
if $\PN_1$ is in $M_s$.
When $\PN_1$ is not in $M_s$, (modified) $\PN_2$ may not leave $M_0^2$ since $M_s$ is not dominated in $\PN_1$.
When $\PN_1$ reaches $M_s$, (modified) $\PN_2$ may start its job
but then (modified) $\PN_1$ may not leave $M_s$ until $\PN_2$ returns to $M_0^2$ since the latter is not dominated in $\PN_2$.
When $\PN_2$ is in $M_0^2$, (modified) $\PN_1$ may leave $M_s$
but then (modified) $\PN_2$ may not leave $M_0^2$ since $M_s$ is not dominated in $\PN_1$.
The evolutions of the constructed net thus correspond exactly to what is described by $\TS$.
\end{proof}  

\noindent Note that we do not claim this is the only solution, but the goal is to find a solution when there is one.


It remains to show when and how an LTS  may be decomposed by a non-trivial  articulation (or several ones).
Let us thus consider some  LTS $\TS=(S,\to,T,\is)$.
We may assume it is finite, totally reachable, deterministic and weakly live (there is no useless label).

\medskip
First, we may observe that, for any two distinct  labels $t,t'\in T$, if $|\adj(\{t\})\cap\adj(\{t'\})|>1$,
$t$ and $t'$ must belong to the same subset for defining an articulation.
Let us extend the function $\adj$ to non-empty subsets of labels by stating $\adj(T')=\cup_{t\in T'}\adj(t)$
when $\emptyset\subset T'\subset T$ .
We then have that, if $\emptyset\subset T_1,T_2\subset T$ and we know that all the labels in $T_1$
must belong to the same subset for defining an articulation, and similarly for $T_2$,
$|\adj(T_1)\cap\adj(T_2)|>1$ implies that
$T_1\cup T_2$ must belong to the same subset of labels defining an articulation.
If we get the full set $T$, that means that there is no possible articulation (but the trivial one).

Hence, starting from any partition $\mathcal T$ of $T$
(initially, if $T=\{t_1,t_2,\ldots,t_n\}$, we shall start from the finest partition ${\mathcal T}=\{\{t_1\},\{t_2\},\ldots,\{t_n\}\}$),
we shall construct the finest partition compatible with the previous rule:\medskip

\noindent \textbf{while there is $T_1,T_2\in{\mathcal T}$ such that $T_1\neq T_2$ and $|\adj(T_1)\cap\adj(T_2)|>1$, replace $T_1$ and $T_2$ in $\mathcal T$ by $T_1\cup T_2$}.\medskip

\noindent  At the end, if ${\mathcal T}=\{T\}$, we may stop with the result: {\em there is no non-trivial articulation}.

\medskip
Otherwise, we may define a finite bipartite undirected graph
whose nodes are the members of the partition $\mathcal T$ and some states of $S$,
such that  if $T_i,T_j\in\mathcal T, T_i\neq T_j$ and  $\adj(T_i)\cap\adj(T_j)=\{s\}$, there is a node $s$ in the graph,
connected to $T_i$ and $T_j$ (and this is the only reason to have a state as a node of the graph).
Since $\TS$ is weakly live and totally reachable, this graph is connected, and each state occurring in it has at least two neighbours
(on the contrary, a subset of labels may be connected to a single state).
Indeed, since $\TS$ is weakly live, $\cup_{T'\in\mathcal T}\adj(T')=S$.
Each state $s$ occurring as a node in the graph is connected to at least two members of the $\mathcal T$,
by the definition of the introduction of $s$ in the graph. Let $T_1$ be the member of $\mathcal T$ such that $\is\in\adj(T_1)$,
let $T_i$ be any other member of $\mathcal T$, and let us consider a path $\is[\alpha\rangle$ ending with some $t\in T_i$
(we may restrict our attention to a short such path, but this is not necessary):
each time there is a sequence $t't''$ in $\alpha$ such that $t'$ and $t''$ belong to two different  members $T'$ and $T''$ of
$\mathcal T$, we have $[t'\rangle s [t''\rangle$, where $s$ is the only state-node connected to $T'$ and $T''$,
hence in the graph we have $T'\rightarrow s\rightarrow T''$.
This will yield a path in the constructed graph going from $T_1$ to $T_i$, hence the connectivity.

If there is a cycle in this graph, that means that there is no way to group the members of $\mathcal T$ in this cycle in two subsets
such that the corresponding adjacency sets only have a single common state. Hence we need to fuse all these members,
for each such cycle, leading to a new partition, and we also need to go back to the refinement of the partition in order to be compatible with the intersection rule, and to the construction of the graph.

\begin{figure}[htbp]
\begin{center}
\begin{tikzpicture}[scale=1.0]
\node[circle,fill=black!100,inner sep=0.07cm](s0)at(0,0)[label=above:$\is$]{};
\node[circle,fill=black!100,inner sep=0.05cm](s1)at(1,0)[label=right:$s_1$]{};
\node[circle,fill=black!100,inner sep=0.05cm](s2)at(2,-1)[label=below:$s_2$]{};
\node[circle,fill=black!100,inner sep=0.05cm](s3)at(2,1)[label=above:$s_3$]{};
\draw[-latex](s0)to[out=30,in=150]node[above,swap,inner sep=2pt,pos=0.5]{$a$}(s1);
\draw[-latex](s1)to[out=-150,in=-30]node[below,swap,inner sep=2pt,pos=0.5]{$b$}(s0);
\draw[-latex](s1)to[]node[left,inner sep=2pt,pos=0.5]{$c$}(s2);
\draw[-latex](s2)to[]node[right,inner sep=2pt,pos=0.5]{$d$}(s3);
\draw[-latex](s3)to[]node[left,inner sep=2pt,pos=0.55]{$e$}(s1);
\node[circle,fill=black!100,inner sep=0.05cm](s4)at(3,1)[]{};
\node[circle,fill=black!100,inner sep=0.05cm](s5)at(4,1)[]{};
\draw[-latex](s3)to[]node[above,inner sep=2pt,pos=0.5]{$f$}(s4);
\draw[-latex](s4)to[]node[above,inner sep=2pt,pos=0.5]{$f$}(s5);
\node[circle,fill=black!100,inner sep=0.05cm](s6)at(0,-1)[]{};
\draw[-latex](s2)to[out=-150,in=-30]node[below,swap,inner sep=2pt,pos=0.5]{$g$}(s6);
\draw[-latex](s6)to[out=30,in=150]node[below,swap,inner sep=2pt,pos=0.5]{$h$}(s2);
\node[circle,fill=black!100,inner sep=0.05cm](s7)at(4,-1)[label=above:$s_7$]{};
\draw[-latex](s2)to[out=30,in=150]node[below,swap,inner sep=2pt,pos=0.5]{$i$}(s7);
\draw[-latex](s7)to[out=-150,in=-30]node[below,swap,inner sep=2pt,pos=0.5]{$j$}(s2);
\node[circle,fill=black!100,inner sep=0.05cm](s8)at(5,-1)[]{};
\node[circle,fill=black!100,inner sep=0.05cm](s9)at(5,0)[]{};
\draw[-latex](s7)to[]node[above,inner sep=2pt,pos=0.5]{$k$}(s8);
\draw[-latex](s8)to[]node[right,inner sep=2pt,pos=0.5]{$k$}(s9);
\node[]at(2,-2){$TS$};
\end{tikzpicture}\\[0.7cm]
\begin{tikzpicture}[scale=1.0]
\node[](s1)at(0,0){$\{a,b\}$};
\node[](ss1)at(1.5,0){$s_1$};
\node[](s2)at(3,0){$\{c,d,e\}$};
\node[](ss3)at(3,1){$s_3$};
\node[](s3)at(4.5,1){$\{f\}$};
\node[](ss2)at(4.5,0){$s_2$};
\node[](s4)at(6,0){$\{h,g\}$};
\node[](s5)at(4.5,-1){$\{i,j\}$};
\node[](ss7)at(6,-1){$s_7$};
\node[](s6)at(7.5,-1){$\{k\}$};
\draw[-](s1)to(ss1);\draw[-](ss1)to(s2);
\draw[-](s2)to(ss3);\draw[-](ss3)to(s3);
\draw[-](s2)to(ss2);\draw[-](ss2)to(s4);
\draw[-](ss2)to(s5);
\draw[-](s5)to(ss7);\draw[-](ss7)to(s6);
\node[]at(3,-2){$G$};
\end{tikzpicture}\\[0.7cm]
\begin{tikzpicture}[scale=1.0]
\node[circle,fill=black!100,inner sep=0.07cm](s0)at(0,3)[label=above:$\is$]{};
\node[circle,fill=black!100,inner sep=0.05cm](s1)at(0,1.5)[label=below:$s_1$]{};
\draw[-latex](s0)to[out=-20,in=20]node[right,swap,inner sep=2pt,pos=0.5]{$a$}(s1);
\draw[-latex](s1)to[out=160,in=200]node[left,swap,inner sep=2pt,pos=0.5]{$b$}(s0);
\node[]at(0,0.5){$TS_1$};
\end{tikzpicture}\hspace*{1.0cm}
\begin{tikzpicture}[scale=1.0]
\node[circle,fill=black!100,inner sep=0.07cm](s1)at(1,0)[label=left:$\is_2$]{};
\node[circle,fill=black!100,inner sep=0.05cm](s2)at(2,-1)[label=below:$s_2$]{};
\node[circle,fill=black!100,inner sep=0.05cm](s3)at(2,1)[label=above:$s_3$]{};
\draw[-latex](s1)to[]node[left,inner sep=2pt,pos=0.5]{$c$}(s2);
\draw[-latex](s2)to[]node[right,inner sep=2pt,pos=0.5]{$d$}(s3);
\draw[-latex](s3)to[]node[left,inner sep=2pt,pos=0.55]{$e$}(s1);
\node[]at(1.5,-2){$TS_2$};
\end{tikzpicture}\hspace*{1.0cm}
\begin{tikzpicture}[scale=1.0]
\node[circle,fill=black!100,inner sep=0.07cm](s0)at(0,2)[label=above:$\is_3$]{};
\node[circle,fill=black!100,inner sep=0.05cm](s1)at(0,1)[]{};
\node[circle,fill=black!100,inner sep=0.05cm](s2)at(0,0)[]{};
\draw[-latex](s0)to[]node[right,swap,inner sep=2pt,pos=0.5]{$f$}(s1);
\draw[-latex](s1)to[]node[left,swap,inner sep=2pt,pos=0.5]{$f$}(s2);
\node[]at(0,-1){$TS_3$};
\end{tikzpicture}\\[0.7cm]
\begin{tikzpicture}[scale=1.0]
\node[circle,fill=black!100,inner sep=0.07cm](s0)at(0,3)[label=above:$\is_4$]{};
\node[circle,fill=black!100,inner sep=0.05cm](s1)at(0,1.5)[label=below:$$]{};
\draw[-latex](s0)to[out=-20,in=20]node[right,swap,inner sep=2pt,pos=0.5]{$h$}(s1);
\draw[-latex](s1)to[out=160,in=200]node[left,swap,inner sep=2pt,pos=0.5]{$g$}(s0);
\node[]at(0,0.5){$TS_4$};
\end{tikzpicture}\hspace*{1.0cm}
\begin{tikzpicture}[scale=1.0]
\node[circle,fill=black!100,inner sep=0.07cm](s0)at(0,3)[label=above:$\is_5$]{};
\node[circle,fill=black!100,inner sep=0.05cm](s1)at(0,1.5)[label=below:$s_7$]{};
\draw[-latex](s0)to[out=-20,in=20]node[right,swap,inner sep=2pt,pos=0.5]{$i$}(s1);
\draw[-latex](s1)to[out=160,in=200]node[left,swap,inner sep=2pt,pos=0.5]{$j$}(s0);
\node[]at(0,0.5){$TS_5$};
\end{tikzpicture}\hspace*{1.0cm}
\begin{tikzpicture}[scale=1.0]
\node[circle,fill=black!100,inner sep=0.07cm](s0)at(0,2)[label=above:$\is_6$]{};
\node[circle,fill=black!100,inner sep=0.05cm](s1)at(0,1)[]{};
\node[circle,fill=black!100,inner sep=0.05cm](s2)at(0,0)[]{};
\draw[-latex](s0)to[]node[right,swap,inner sep=2pt,pos=0.5]{$k$}(s1);
\draw[-latex](s1)to[]node[left,swap,inner sep=2pt,pos=0.5]{$k$}(s2);
\node[]at(0,-1){$TS_6$};
\end{tikzpicture}\\[0.2cm]
$TS\equiv\TS_1\triangleleft s_1\triangleright(((\TS_2\triangleleft s_3\triangleright\TS_3)\triangleleft s_2\triangleright\TS_4)\triangleleft s_2\triangleright(\TS_5\triangleleft s_7\triangleright\TS_6))$
\end{center}\vspace*{-2mm}
\caption{The lts $\TS$ leads to the graph $G$. The corresponding components are $\TS_1$ to $\TS_6$,
which may easily be synthesised;
note that, from the total reachability of $\TS$, they are all totally reachable themselves.
This leads to the articulated expression  below.}
\label{decomp.fig}
\end{figure}

Finally, we shall get an acyclic graph $G$, with at least three nodes
(otherwise we stopped the articulation algorithm with the information that there is no non-trivial decomposition).

We shall now define a procedure $\art(SG)$ that builds an LTS expression based on articulations from a subgraph $SG$ of $G$
with a chosen state-node root. We shall then apply it recursively to $G$,
leading finally to an articulation-based (possibly complex) expression equivalent to the original LTS~$\TS$.

The basic case will be that, if  $SG$ is a graph composed of a state $s$ connected to a subset node $T_i$,
$\art(SG)$ will be the LTS $\TS_i=(\adj(T_i),T_i,\to_i,s)$ (as usual $\to_i$ is the projection of $\to$ on $T_i$;
by construction, it will always be the case that $s\in\adj(T_i)$).

First, if $\is$ is a state-node of the graph, $G$ then has the form of a star with root $\is$ and a set of satellite subgraphs
$G_1$, $G_2$, ..., $G_n$ ($n$ is at least 2). Let us denote by $SG_i$ the subgraph with root $\is$ connected to $G_i$:
the result will then be the (commutative, see Proposition~\ref{com.prop}) articulation around $\is$ of all the LTSs $\art(SG_i)$.

Otherwise, let $T_1$ be the (unique) label subset in the graph such that $\is\in\adj(T_1)$.
$G$ may then be considered as a star with $T_1$ at the center, surrounded by subgraphs $SG_1$, $SG_2$, ..., $SG_n$
(here $n$ may be $1$), each one with a root $s_i$ connected to $T_1$
(we have here that $s_i\in\adj(T_1)$, and we allow $s_i=s_j$): the result is then
$((\ldots((\adj(T_1),T_1,\to_1,\is)\triangleleft s_1\triangleright\art(SG_1))\triangleleft s_2\triangleright\art(SG_2))\ldots)\triangleleft s_n\triangleright\art(SG_n))$.
Note that, if $n>1$, the order in which we consider the subgraphs is irrelevant from Proposition~\ref{comassoc.prop}.

\medskip
Finally, if a subgraph starts from a state $s'$, followed by a subset $T'$, itself followed by subgraphs $SG_1$, $SG_2$, ..., $SG_n$
($n\geq 1$; if it is $0$ we have the base case), each one with a root $s_i$ connected to $T'$
(we have here that $s'\in\adj(T')$, and we allow $s_i=s_j$): the result is then
$((\ldots((\adj(T'),T',\to',s')\triangleleft s_1\triangleright\art(SG_1))\triangleleft s_2\triangleright\art(SG_2))\ldots)\triangleleft s_n\triangleright\art(SG_n))$.
Again,  if $n>1$, the order in which we consider the subgraphs is irrelevant from Proposition~\ref{comassoc.prop}.

\medskip
This procedure is illustrated in Figure~\ref{decomp.fig}.

\medskip
Contrary to what happened for the product of nets, articulations do not preserve many Petri net subclasses.
For instance, if $\PN_1$ and $\PN_2$ are plain (no arc weight greater than 1),
$\PN_1 \triangleleft M \triangleright\PN_2$ is not plain (unless $M$ and $M_0^2$ do not have more than one token in any place);
however, since we may have many solutions to a synthesis problem, it may happen that $\PN_1 \triangleleft M \triangleright\PN_2$  is not safe but that another solution of the same problem is safe.
On the contrary, an immediate consequence of Definition~\ref{articulPN.def} is that

\begin{corollary}\label{safea.cor}{\bf (Bound preservation)}\\
$\PN_1 \triangleleft M \triangleright\PN_2$  is safe iff so are $\PN_1$ and $\PN_2$.
If $\PN_1 \triangleleft M \triangleright\PN_2$ is $k$-safe, so are $\PN_1$ and $\PN_2$;
finally, if $\PN_1$ is $k_1$-safe and $\PN_2$ is $k_2$-safe,
then $\PN_1 \triangleleft M \triangleright\PN_2$ is $\max(k_1,k_2)$-safe.
\end{corollary}

\section{Mixed decomposition}\label{mixed.sct}

In the previous sections we have introduced two pairs of (families of) operators acting on transition systems and Petri net systems:
$\TS_1\otimes\TS_2\mbox{ - }\PN_1\oplus\PN_2$ and
$\TS_1\triangleleft s\triangleright\TS_2\mbox{ - }\PN_1\triangleleft M\triangleright\PN_2$.

They may be intermixed, as exemplified in Figure~\ref{seq.fig}, where
$\TS=\TS(start);(\TS(a)\otimes\TS(b))$; $\TS(end)$ may be rewritten
$\TS=\TS(start)\triangleleft s_1\triangleright((\TS(a)\otimes\TS(b))\triangleleft s_4\triangleright\TS(end))$.

We may wonder however if there are cases where a transition system may be decomposed both as a
(non-trivial)  product and as an (non-trivial) articulation.
In the following, we shall assume there is no useless label (i.e., each label occurs at least once in a transition),
and that the transition systems are totally reachable (otherwise there is no Petri net solution).
If $T$ is a set of labels, we shall  also denote by $\is_T$ the transition system with a single state $\is$ and, for each $t\in T$, a loop $\is[t\rangle \is$ (up to isomorphism, it is the only transition system with only one state and label set $T$, without useless label).

\medskip
First, we may have $\TS\equiv\TS_1\otimes\TS_2\equiv\TS_1\triangleleft s_1\triangleright\TS_2$,
 but only if $\TS,\TS_1,\TS_2$ have a single state, as illustrated by  Figure~\ref{single.fig}.

\begin{figure}[hbt]
\vspace*{-2mm}
\begin{tikzpicture}[scale=0.9]
\draw[-latex] (-1,0) .. controls (-3,2) and (1,2) .. (-1,0);
\node[](T1)at(-1,1){$a$};
\node[circle,fill=black!100,inner sep=0.07cm](i1)at(-1,0)[label=below:$\is_1$]{};
\node[](TS1)at(-1,-1){{$\TS_1$}};
\draw[-latex] (1,0) .. controls (-1,2) and (3,2) .. (1,0);
\node[](T2)at(1,1){$b$};
\node[circle,fill=black!100,inner sep=0.07cm](i2)at(1,0)[label=below:$\is_2$]{};
\node[](TS2)at(1,-1){{$\TS_2$}};
\draw[-latex] (6,0) .. controls (2,2) and (6,2) .. (6,0);
\node[](T1)at(5,1){$a$};
\draw[-latex] (6,0) .. controls (6,2) and (10,2) .. (6,0);
\node[](T1)at(7,1){$b$};
\node[circle,fill=black!100,inner sep=0.07cm](i)at(6,0)[label=below:$\is$]{};
\node[](TS)at(6,-1){{$\TS_1\otimes\TS_2\equiv_{\{a,b\}}\TS_1\triangleleft \is_1\triangleright\TS2$}};
\end{tikzpicture}\vspace*{-2mm}
\caption{Singleton case.} \label{single.fig}\vspace*{-3mm}
\end{figure}
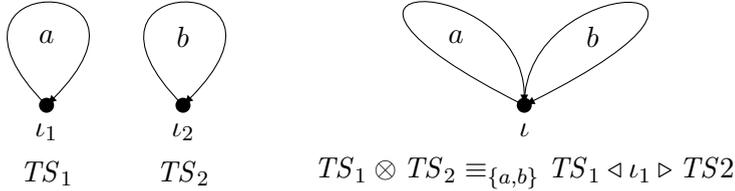

 \begin{proposition}\label{single.prop}{\bf (Equivalent decompositions (first case))}\\
 $\TS\equiv\TS_1\otimes\TS_2\equiv\TS_1\triangleleft s_1\triangleright\TS_2$, with $|T_1|>0<|T_2|$,
 iff $|S_1|=1=|S_2|$ (hence also $|S|=1$ and $s_1=\is_1$).
 \end{proposition}

 \begin{proof}
 If $\{T_1,T_2\}$ is a partition of $T$, we have that
 $\is_T\equiv\is_{T_1}\otimes\is_{T_2}\equiv\is_{T_1}\triangleleft \is\triangleright\is_{T_2}$.

 \medskip
 Let us now assume that  $\TS\equiv\TS_1\otimes\TS_2\equiv\TS_1\triangleleft s_1\triangleright\TS_2$.
 We must have $|S|=|S_1|\cdot|S_2|=|S_1|+|S_2|-1$, so that $(|S_1|-1)\cdot(|S_2|-1)=0$ and $|S_1|=1$ or $|S_2|=1$.
 Let us assume that $|S_2|>1$ (the case $|S_1|>1$ is symmetrical):
 there must be $s\neq s'\in S_2$ and $t\in T_2$ such that $(s,t,s')\in\to_2$.
 Since $|S_1|=1$ and it is assumed there is no useless label, $\TS_1\equiv\is_{T_1}$ for some  partition $\{T_1,T_2\}$ of $T$.
 In $\TS_1\otimes\TS_2$, we have $(\is,s)[t\rangle(\is,s')$, $(\is,s)[a\rangle(\is,s)$ and $(\is,s')[a\rangle(\is,s')$ for any $a\in T_1$.
 Hence $(\is,s)\neq(\is,s')\in\adj(t)\cap\adj(a)$ and $t,a\in T_2$ in $\TS_1\triangleleft \is\triangleright\TS_2$,
 contradicting the fact that $\{T_1,T_2\}$ is a partition of $T$.
 Hence, we must have $|S_1|=|S_2|=1$.
 \end{proof} 

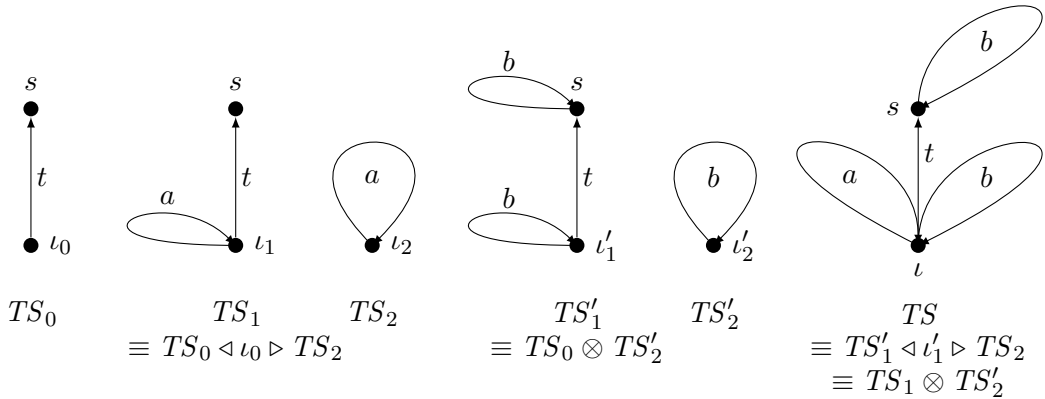
\begin{figure}[!b]
\vspace*{-7mm}
\begin{tikzpicture}[scale=0.9]
\node[circle,fill=black!100,inner sep=0.07cm](i0)at(-4,0)[label=right:$\is_0$]{};
\node[circle,fill=black!100,inner sep=0.07cm](s0)at(-4,2)[label=above:$s$]{};
\draw[-latex](i0)to[]node[right,swap,inner sep=2pt,pos=0.5]{$t$}(s0);
\node[](TS1)at(-4,-1){{$\TS_0$}};
\draw[-latex] (-1,0) .. controls (-4,0) and (-2,1) .. (-1,0);
\node[](T1)at(-2,0.7){$a$};
\node[circle,fill=black!100,inner sep=0.07cm](i1)at(-1,0)[label=right:$\is_1$]{};
\node[circle,fill=black!100,inner sep=0.07cm](s)at(-1,2)[label=above:$s$]{};
\draw[-latex](i1)to[]node[right,swap,inner sep=2pt,pos=0.5]{$t$}(s);
\node[](TS1)at(-1,-1){{$\TS_1$}};
\node[](TS1)at(-1,-1.5){{$\equiv\TS_0\triangleleft \is_0\triangleright\TS_2$}};
\draw[-latex] (1,0) .. controls (-1,2) and (3,2) .. (1,0);
\node[](T2)at(1,1){$a$};
\node[circle,fill=black!100,inner sep=0.07cm](i2)at(1,0)[label=right:$\is_2$]{};
\node[](TS2)at(1,-1){{$\TS_2$}};
\draw[-latex] (4,0) .. controls (1,0) and (3,1) .. (4,0);
\node[](T1)at(3,0.7){$b$};
\draw[-latex] (4,2) .. controls (1,2) and (3,3) .. (4,2);
\node[](T1)at(3,2.7){$b$};
\node[circle,fill=black!100,inner sep=0.07cm](i1)at(4,0)[label=right:$\is'_1$]{};
\node[circle,fill=black!100,inner sep=0.07cm](s)at(4,2)[label=above:$s$]{};
\draw[-latex](i1)to[]node[right,swap,inner sep=2pt,pos=0.5]{$t$}(s);
\node[](TSp1)at(4,-1){{$\TS'_1$}};
\node[](TS1)at(4,-1.5){{$\equiv\TS_0\otimes\TS'_2$}};
\draw[-latex] (6,0) .. controls (4,2) and (8,2) .. (6,0);
\node[](T2)at(6,1){$b$};
\node[circle,fill=black!100,inner sep=0.07cm](i2)at(6,0)[label=right:$\is'_2$]{};
\node[](TSp2)at(6,-1){{$\TS'_2$}};
\draw[-latex] (9,0) .. controls (5,2) and (9,2) .. (9,0);
\node[](T1)at(8,1){$a$};
\draw[-latex] (9,0) .. controls (9,2) and (13,2) .. (9,0);
\node[](T1)at(10,1){$b$};
\draw[-latex] (9,2) .. controls (9,4) and (13,4) .. (9,2);
\node[](T1)at(10,3){$b$};
\node[circle,fill=black!100,inner sep=0.07cm](i')at(9,0)[label=below:$\is$]{};
\node[circle,fill=black!100,inner sep=0.07cm](s')at(9,2)[label=left:$s$]{};
\draw[-latex](i')to[]node[right,swap,inner sep=2pt,pos=0.7]{$t$}(s');
\node[](TS)at(9,-1){{$\TS$}};
\node[](TS1)at(9,-1.5){{$\equiv\TS'_1\triangleleft \is'_1\triangleright\TS_2$}};
\node[](TS1)at(9,-2){{$\equiv\TS_1\otimes\TS'_2$}};
\end{tikzpicture}\vspace*{-4mm}
\caption{Ambiguous case.} \label{ambig.fig}
\end{figure}

 In terms of synthesis, the gain of decomposing $\TS$ is not very high in this case,
 since there is an easy solution (composed of isolated transitions).

There are special cases, however, where we have a choice between a product and an articulation decomposition
with non-singleton state spaces, but with different partitions of the label set, as illustrated by Figure~\ref{ambig.fig}.
This only occurs however when one of the components  is a singleton system.

\begin{proposition}\label{ambig.prop}{\bf (Equivalent decompositions (second case))}\\
If $\TS\equiv\TS_1\otimes\TS_2\equiv\TS'_1\triangleleft s'\triangleright\TS'_2$,
then $|S_1|=1$ or $|S_2|=1$ as well as $|S'_1|=1$ or $|S'_2|=1$.
\end{proposition}

\begin{proof}
Let us first assume that $|S_1|>1<|S_2|$.
 In $\TS_1$, there are $s_1\neq s'_1\in S_1$ and $a_1\in T_1$ with $(s_1,a_1,s'_1)\in\to_1$.
For any $a_2\in T_2$, since there is no useless label, we have $(s_2,a_2,s'_2)\in\to_2$ (here we allow $s_2=s'_2$).
In $\TS_1\otimes\TS_2$, we then have  that $(s_1,s_2)[a_1\rangle(s'_1,s_2)$, $(s_1,s'_2)[a_1\rangle(s'_1,s'_2)$
 $(s_1,s_2)[a_2\rangle(s_1,s'_2)$ and $(s'_1,s_2)[a_2\rangle(s'_1,s'_2)$, so that $|\adj(a_1)\cap\adj(a_2)|>1$ and,
 in any articulation decomposition of $\TS$, $a_1$ must belong to the same component as any $a_2\in T_2$.
 Symmetrically, some $a_2$ in $T_2$ must belong to the same component as any $a_1\in T_1$ in any articulation of $\TS$.
 As a consequence, any articulation of $\TS$ must be trivial (all the labels belong to the same component).

\medskip
 Let us now assume that $|S'_1|>1<|S'_2|$. From the previous point, we may not have $|S_1|>1<|S_2|$;
 without loss of generality, we shall assume $|S_1|=1<|S_2|=|S|$.
 Let $a_1\in T_1$ (there must be one, since we assumed $T_1\neq\emptyset$);
 then $a_1$ must occur as a loop around each state of $\TS\equiv\TS_1\otimes\TS_2$.
 Hence $a_1$ is adjacent to any state and we may only have $\TS\equiv\TS'_1\triangleleft s'\triangleright\TS'_2$
 if there is a single component in the articulation, contradicting $|S'_1|>1<|S'_2|$.
\drop{ In $\TS'_2$, we must have $\is'_2\neq s'_2$ with $\is'_2[a'_2\rangle s'_2$ for some $a'_2\in T'_2$.
 In $\TS'_1$, we must have $s'_1\neq s'$ with $s'_1[a'_1\rangle s'$ for some $a'_1\in T'_1$
 (or $s'[a'_1\rangle s'_1$, the situation will be similar, or both).
Moreover,  $s'_1\not\in\adj(a'_2)$ and $s'_2\not\adj(a'_1)$.
 Since, $|S_1|=1$, let us assume that,  in $\TS_1\otimes\TS_2$, $s'$ corresponds to $(\is_1,s_2)$.
 Then,  if $a'_1\in T_1$, $a'_1$ is adjacent to any state, which contradicts $s'_1\not\in\adj(a'_2)$.
If $a'_2\in T_2$}
\end{proof} 

\begin{corollary}\label{ambig.cor}{\bf (General ambiguous form)}\\
The only cases where we have a choice between a product and an articulation of transition systems have the form
(up to a permutation of the roles of the three components)
$$(\TS_1\otimes\is_{T_2})\triangleleft (s_1,\is)\triangleright\is_{T_3}\equiv
(\TS_1\triangleleft s_1\triangleright\is_{T_3})\otimes\is_{T_2}$$
where $T_2$, $T_3$ and the label set of $\TS_1$,  are pairwise disjoint and $s_1$ is a state of $\TS_1$. \\
It is Petri net solvable iff so is $\TS_1$ and a possible solution is a solution of $\TS_1$ adequate for $s_1$,
plus one transition for each label of $T_3$ connected by a side condition to the places marked by the
marking corresponding to $s_1$ (with weights given by this marking), plus one isolated transition for each label of~$T_2$.
\end{corollary}

When we have a choice between a factorisation and an articulation,
the former should probably be preferred since a synthesis may then be put in the form of a sum of nets.
But anyway, since at least one of the components of the product and at least one of the components of the articulation
have only one state, the gain for the synthesis is low.
Let us recall however that another interest of these decompositions is to exhibit an internal structure
for the examined transition system.

For components with at least two states (hence not only composed of loops), there is no ambiguity in the decomposition,
but we may alternate products and articulations, using the procedures  detailed at the end of Sections~\ref{prod.sct} and \ref{art.sct}. At the end, we shall obtain non decomposable components, but also possibly a choice between a product and an articulation if we get singleton components as in the examples in Figures~\ref{single.fig} and \ref{ambig.fig}.

\section{Experiments}\label{experiment.sct}

Let us now examine more closely the gain in efficiency the decompositions explored in this paper afford for the Petri net synthesis.
Several remarks may be done before we start true experiments.

\medskip
First, it is better to distinguish positive cases (where there is a solution) from negative ones.
Indeed, in a negative case, we may stop when a necessary condition is not satisfied or
when a SSP or ESSP problem may not be solved. In each case, the performance highly relies on the order in which
sub- problems are considered.
And even if we pursue in order to find all the obstacles, it may happen that finding that a problem has no solution
has not the same complexity than finding a solution when there is one.
Hence, in the following, we shall only consider positive examples.
In the same spirit we shall not try to find an optimal solution
(for instance with a minimal number of places, or with small weights, etc.).

Next, there is $|S|\cdot(|S|-1)/2$ SSP problems and
$|S|\cdot|T|-|\to|$ ($=\sum_{s\in S}(|T|-\mbox{outdegree}(s)$) ESSP problems
(where the outdegree of a state is the number of arcs originated from it), possibly after some pre-synthesis checks.

The pre-synthesis checks allow to quickly reject inadequate transition systems,
without needing expensive solutions to SSP and ESSP problems. This includes the checks for total reachability and determinism,
but also other ones if we search for solutions in a specific subclass of net systems.
For instance, if we search for choice-free solutions (i.e., nets where each place has at most one output transition),
many necessary structural conditions have been developed for a pre-synthesis phase~\cite{BDS-acta17,BestDE20}.
For positive cases, the pre-synthesis phase is usually negligible
with respect to the time and memory requested by the solution of ESSP and SSP problems,
but may provide interesting informations and data-structures for solving those separation problems.

In large reachability graphs, the number of states is usually much larger than the number of transitions.
For instance, if we increase the number of initial tokens, the size of the graph may increase in a dramatic manner,
while the number of transitions does not change. This leads to consider complexities in terms of the number of states $|S|$
for the positive synthesis of large transition systems, and to neglect the impact of the number of labels $|T|$.

To solve a separation problem, we have to find an adequate region $(\rho,\bk,\fw)$. The number of unknowns is $|S|+2\cdot|T|$
and for each arc $s[t\rangle s'\in\to$ we have two linear constraints: $\rho(s)\geq\bk(t)$ and $\rho(s')=\rho(s)+\fw(a)-\bk(a)$.
For a $\text{SSP}(s,s')$, we have to add the constraint $\rho(s)\neq\rho(s')$,
and for an $\text{ESSP}(s,t)$ we have to add the constraint $\rho(s)<\bk(a)$.
Hence, we have to solve systems of homogeneous linear constraints of the same size, and
the complexity $CS$ of all the separation problems is about the same.
This leads to a positive synthesis complexity of the kind $A\cdot|S|^2\cdot CS+E\cdot|S|\cdot CS$
for some coefficients $A$ and $B$.

However, this is not correct in general.
First, for the synthesis of some subclasses of Petri nets, it is known that SSPs are irrelevant:
the regions (or places) built to solve the ESSPs also solve all the SSPs.
This is the case for the choice-free nets~\cite{BDS-acta17}, hence also for all their subclasses.
But even for general Petri net syntheses, this leads to the idea to first solve the ESSP problems and, for each SSP problem,
check first if there is no region built previously which already solves it,
which is very quick since there are generally few needed regions and the check is easy.
This has been used for instance in the tool APT~\cite{EB-US-15}, and the experience shows that most of the time,
no new system of linear constraints has to be solved for SSP problems, and otherwise only a very small number of them is needed.
The same is true for ESSP problems: instead of systematically building and solving the corresponding system of linear constraints,
we may first check if one of the regions built previously does not already solve the new separation problem
(the efficiency of the procedure then relies on the order in which ESSP problems are considered).
This may be interpreted in the following way: the global complexity is of the kind $(A\cdot|S|^2+E\cdot|S|)\cdot CS$;
in many cases $A=0$; if this is not true, it may happen that $A\cdot|S|^2$ does not grow faster than $E\cdot|S|$;
and if it does, most of the time $A$ is so small that the dominance of $A\cdot|S|^2$ only occurs for synthesis problems that are so huge that, anyway, the problem is out of practical reach, due to memory and/or execution time overflow problems.
Note that $E\cdot|S|$ also may grow much less than linearly.
This is difficult to evaluate theoretically beforehand and experiments may be useful for that.

For the synthesis of weighted Petri nets, each system of linear constraints to be solved is homogeneous\footnote{
i.e., there is no independent terms; in that case, finding a solution in the integer domain is equivalent to finding one
in the rational domain, since it is always possible to multiply a rational solution by an adequate factor to get an integer one.}
and we  may use the Karmarkar algorithm~\cite{karmarkar} for that,
which has a known worse case complexity  polynomial in the size of the system (here linear in $|S|$).
The exact polynomial exponent is difficult to determine, but it is usually expected to be around 5, hence the degree 7
 mentioned in the introduction section, to solve a quadratic number of SSP problems,
 but we know now that we should expect a much smaller number of separation problems
 to be solved in the form of a system of linear constraints of linear size.
 Moreover, it is known that, in practice, the Karmarkar algorithm is outperformed by the simplex algorithm~\cite{dantzig51},
 which however is exponential in principle~\cite{klee72}.
 The same arises with classical SMT solvers (APT uses SMTInterpol~\cite{smtinterpol,kroen10}).
 This seems to be due to the fact that the families of problems leading to an exponential growth are too artificial
 and are not met in true applications.
 So again, experiments would be welcome to clarify the complexity of tools like APT.

\medskip
 If we now add decomposition algorithms, like the two ones described in this paper, in the pre-synthesis phase,
 we may expect some more improvements.
The decomposition algorithms are very quick with respect to the proper synthesis (see for instance~\cite{fact18})
so that,  even if no non-trivial components are detected, the loss due to the addition of new work in the preliminary phase
is negligible. If non-trivial components are found, a post-processing phase must be added to recompose the various synthesised
Petri net systems, but this is again very quick and negligible with respect to the various intermediate syntheses.
The improvement will be more noticeable if the components are approximatively of the same size.

\begin{proposition}\label{impr.prop}{\bf (Improvement due to decomposition techniques)}\\
Let us assume that, for a family of transition systems with increasing size $|S|$,
the complexity of a positive synthesis is approximately proportional to $|S|^h$ for some $h\in\mathbb{R}_{>0}$.

\medskip
If the transition systems may be factorised in $k$ factors of approximately the same size, the gain is about
$$\frac{|S|^{h(1-1/k)}}{k}$$
If the transition systems may be articulated in $k$ components of approximately the same size, the gain is about
$$k^{h-1}$$
\end{proposition}

\begin{proof}
In the first case, each factor has an approximate size of $\sqrt[k]|S|$, and there are $k$ of them, so that the gain is
$$\frac{|S|^h}{k\cdot|S|^{h/k}}$$
which leads to the first formula.

\medskip
In the second case, each factor has an approximate size of $|S|/k$, and there are $k$ of them, so that the gain is
$$\frac{|S|^h}{k\cdot(|S|/k)^{h}}$$
which leads to the second formula.
\end{proof}  

Hence, we shall consider families of examples with an increasing number of copies of the same transition system as components.
In a product of $n$ copies of systems with $|S|$ states, the state space grows (exponentially) as $|S|^n$
(the number of labels grows linearly: $|T|\cdot n$). In an articulation of $n$ copies of systems with $|S|$ states,
the state space grows (linearly) as $|S|\cdot n-n+1$ (the number of labels still grows linearly).
Contrary to the product cases where there is no variant possible in the way the components are combined,
this is not true for the articulations, where we can chose various ways of plugging a new component on the previous member
of the family, and it could happen that the performance of the chosen tool (here APT) behaves differently on the different families.
We shall here consider three ways of performing the plugging: the star-shape, the daisy-flower shape and the caterpillar shape,
schematised in Figure~\ref{families.fig} with 4 components each.
The first one articulates all the components around the initial state.
The second one articulates the components around different states of the first component.
The last one plugs each component but the first one on a non-initial state of the previous component.
In each case the time used by the synthesis of the $n$ components is linear and takes $n$ times the time to synthesise one of them.

\begin{figure}[hbt]
\vspace*{-6mm}
\hspace*{-1cm}\begin{tikzpicture}[scale=0.9]
\draw[] (6,0) .. controls (2,2) and (6,2) .. (6,0);
\node[](T1)at(5,1){$T_2$};
\draw[] (6,0) .. controls (6,2) and (10,2) .. (6,0);
\node[](T1)at(7,1){$T_1$};
\draw[] (6,0) .. controls (2,-2) and (6,-2) .. (6,0);
\node[](T1)at(5,-1){$T_3$};
\draw[] (6,0) .. controls (6,-2) and (10,-2) .. (6,0);
\node[](T1)at(7,-1){$T_4$};
\node[circle,fill=black!100,inner sep=0.07cm](i)at(6,0)[label=right:$\is$]{};
\node[](TS)at(6,-2){{star-shape}};
\end{tikzpicture}\hspace{-1cm}
\begin{tikzpicture}[scale=0.9]
\node[circle,fill=black!100,inner sep=0.07cm](i)at(6,0)[label=right:$\is$]{};
\draw[] (6,0) .. controls (4,1.5) and (8,1.5) .. (6,0);
\node[](T1)at(6,0.5){$T_1$};
\draw[] (6,0) .. controls (4,-2) and (8,-2) .. (6,0);
\node[](T1)at(6,-1){$T_2$};
\node[circle,fill=black!100,inner sep=0.07cm](i)at(5.5,0.5)[]{};
\draw[] (5.5,0.5) .. controls (3.5,2) and (3.5,-1) .. (5.5,0.5);
\node[](T1)at(4.5,0.5){$T_3$};
\draw[] (6.5,0.5) .. controls (8.5,2) and (8.5,-1) .. (6.5,0.5);
\node[](T1)at(7.5,0.5){$T_4$};
\node[circle,fill=black!100,inner sep=0.07cm](i)at(6.5,0.5)[]{};
\node[](TS)at(6,-2){{daisy-shape}};
\end{tikzpicture}\hspace*{-1cm}
\begin{tikzpicture}[scale=0.9]
\node[circle,fill=black!100,inner sep=0.07cm](i)at(6,0)[label=right:$\is$]{};
\draw[] (6,0) .. controls (4,1.5) and (8,1.5) .. (6,0);
\node[](T1)at(6,0.5){$T_1$};
\draw[] (6.5,0.5) .. controls (8.5,2.3) and (8.5,-1.3) .. (6.5,0.5);
\node[](T1)at(7.5,0.5){$T_2$};
\node[circle,fill=black!100,inner sep=0.07cm](i)at(6.5,0.5)[]{};
\draw[] (8,0.5) .. controls (10,2.3) and (10,-1.3) .. (8,0.5);
\node[](T1)at(9,0.5){$T_3$};
\node[circle,fill=black!100,inner sep=0.07cm](i)at(8,0.5)[]{};
\draw[] (9.5,0.5) .. controls (11.5,2.3) and (11.5,-1.3) .. (9.5,0.5);
\node[](T1)at(10.5,0.5){$T_4$};
\node[circle,fill=black!100,inner sep=0.07cm](i)at(9.5,0.5)[]{};
\node[](TS)at(8,-2){{caterpilar-shape}};
\end{tikzpicture}
\caption{Three families of articulations.} \label{families.fig}
\end{figure}
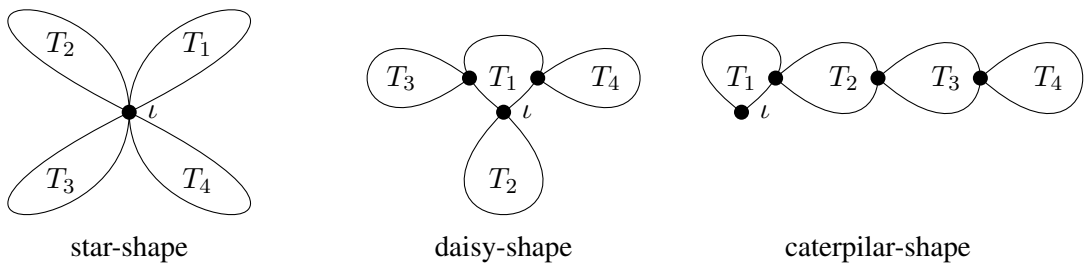

Note however that it may happen that the performances are better than so expected.
Indeed, when decomposing a complex transition system, it may happen that some (or all) components belong to a special case,
to which the original system does not belong and for which a special implementation may speed up the synthesis.
For instance, when decomposing a system into factors, if we get systems with a  PN-solution which is a connected marked graph
(i.e., each place has one input and one output transition, with weight 1),
we may use the very effective synthesis described in~\cite{besdev-lata}; the original system then also has a  marked graph solution
(an addition of marked graphs is a marked graph), but it is not connected and the speed up may not be applied.
When decomposing a system with articulations, it may happen that some (or all) components have a choice-free solution,
for which we know how to reduce the number and size \cite{BDS-acta17} of ESSP problems to be solved
(the SSP problems are then irrelevant). We shall not create such situations in our experiments however.

\smallskip
For the articulations, we shall use the system $\TS~\ref{ts21}$ in Figure~\ref{aug21-synet-aut.fig},
which has 23 states and no choice-free solution.
We compared the obtained CPU time\footnote{The machine we used is based on processors Intel Xeon Gold 2.6 GHz} for the three families mentioned above, up to 100 components.
The results are schematised in $\,$Figure~\ref{cpuarticul.fig}, where$\,$ round dots are used for the$\,$ star-shape

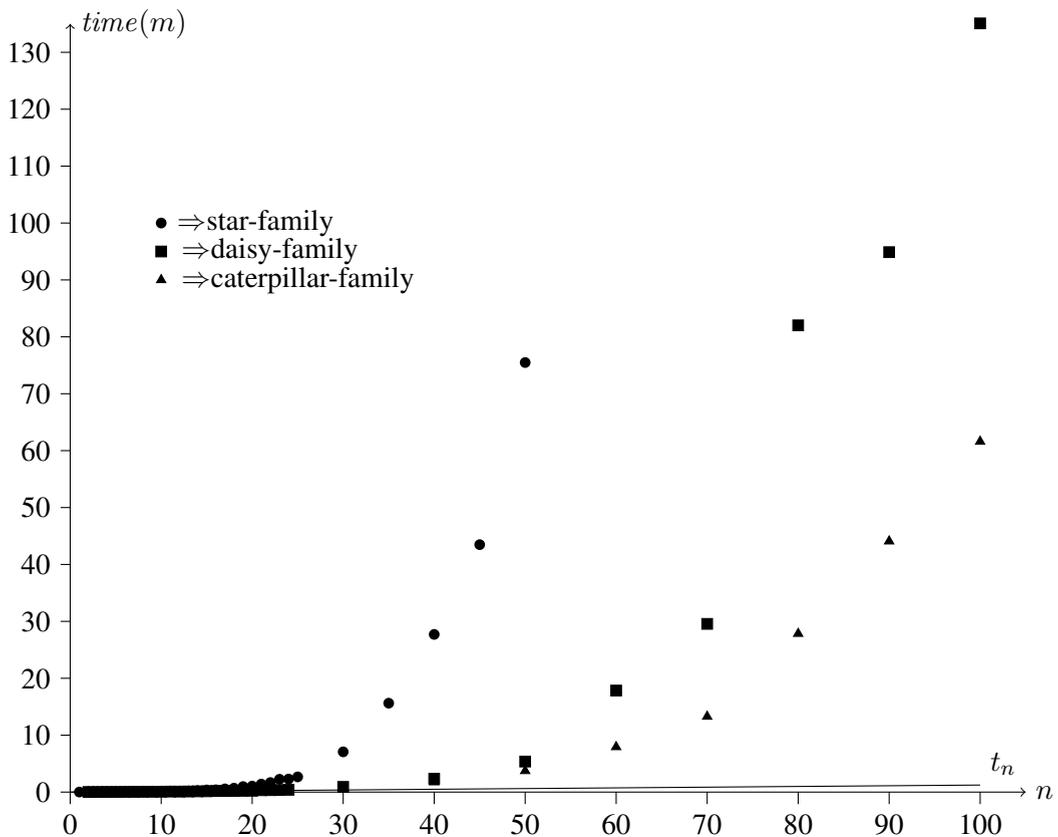
\begin{figure}[!h]
\centering
\begin{tikzpicture}[yscale=0.075, xscale=0.12]
        \draw[->] (0,0) -- (105,0) node[right] {$n$};
        \draw[->] (0,0) -- (0,135) node[right] {$\mathit{time}(m)$};
        \draw[-] (0,0) -- (100,74.3/60) node[above right] {$t_n$};
   	\foreach \x in {0,10,...,100}
     		\draw (\x,1pt) -- (\x,-60pt)
			node[anchor=north] {\x};
    	\foreach \y in {0,10,...,135}
     		\draw (1pt,\y) -- (-30pt,\y)
     			node[anchor=east] {\y};
\node[circle,fill=black!100,inner sep=0.05cm](sc) at (10,100)[label=right:$\impl$star-family]{};
\node(dc) at (10,95){\pgfuseplotmark{square*}};\node at (10,95)[label=right:$\impl$daisy-family]{};
\node(cc) at (10,90){\pgfuseplotmark{triangle*}};\node at (10,90)[label=right:$\impl$caterpillar-family]{};
\node[circle,fill=black!100,inner sep=0.05cm](sc1) at (1,0.743/60)[]{};
\node[circle,fill=black!100,inner sep=0.05cm](sc2) at (2,1.131/60)[]{};
\node[circle,fill=black!100,inner sep=0.05cm](sc3) at (3,1.378/60)[]{};
\node[circle,fill=black!100,inner sep=0.05cm](sc4) at (4,1.752/60)[]{};
\node[circle,fill=black!100,inner sep=0.05cm](sc5) at (5,2.141/60)[]{};
\node[circle,fill=black!100,inner sep=0.05cm](sc6) at (6,2.442/60)[]{};
\node[circle,fill=black!100,inner sep=0.05cm](sc7) at (7,2.973/60)[]{};
\node[circle,fill=black!100,inner sep=0.05cm](sc8) at (8,3.817/60)[]{};
\node[circle,fill=black!100,inner sep=0.05cm](sc9) at (9,5.167/60)[]{};
\node[circle,fill=black!100,inner sep=0.05cm](sc10) at (10,6.335/60)[]{};
\node[circle,fill=black!100,inner sep=0.05cm](sc11) at (11,8.171/60)[]{};
\node[circle,fill=black!100,inner sep=0.05cm](sc12) at (12,9.512/60)[]{};
\node[circle,fill=black!100,inner sep=0.05cm](sc13) at (13,11.235/60)[]{};
\node[circle,fill=black!100,inner sep=0.05cm](sc14) at (14,15.121/60)[]{};
\node[circle,fill=black!100,inner sep=0.05cm](sc15) at (15,22.525/60)[]{};
\node[circle,fill=black!100,inner sep=0.05cm](sc16) at (16,24.124/60)[]{};
\node[circle,fill=black!100,inner sep=0.05cm](sc17) at (17,33.639/60)[]{};
\node[circle,fill=black!100,inner sep=0.05cm](sc18) at (18,41.312/60)[]{};
\node[circle,fill=black!100,inner sep=0.05cm](sc19) at (19,58.068/60)[]{};
\node[circle,fill=black!100,inner sep=0.05cm](sc20) at (20,63.325/60)[]{};
\node[circle,fill=black!100,inner sep=0.05cm](sc21) at (21,84.225/60)[]{};
\node[circle,fill=black!100,inner sep=0.05cm](sc22) at (22,100.641/60)[]{};
\node[circle,fill=black!100,inner sep=0.05cm](sc23) at (23,135.198/60)[]{};
\node[circle,fill=black!100,inner sep=0.05cm](sc24) at (24,137.478/60)[]{};
\node[circle,fill=black!100,inner sep=0.05cm](sc25) at (25,160.194/60)[]{};
\node[circle,fill=black!100,inner sep=0.05cm](sc30) at (30,424.388/60)[]{};
\node[circle,fill=black!100,inner sep=0.05cm](sc35) at (35,937.379/60)[]{};
\node[circle,fill=black!100,inner sep=0.05cm](sc40) at (40,1662.57/60)[]{};
\node[circle,fill=black!100,inner sep=0.05cm](sc45) at (45,2608.97/60)[]{};
\node[circle,fill=black!100,inner sep=0.05cm](sc50) at (50,4529.12/60)[]{};
\node(dc2) at (2,1.185/60){\pgfuseplotmark{square*}};
\node(dc3) at (3,1.412/60){\pgfuseplotmark{square*}};
\node(dc4) at (4,1.699/60){\pgfuseplotmark{square*}};
\node(dc5) at (5,1.974/60){\pgfuseplotmark{square*}};
\node(dc6) at (6,2.327/60){\pgfuseplotmark{square*}};
\node(dc7) at (7,2.609/60){\pgfuseplotmark{square*}};
\node(dc8) at (8,3.004/60){\pgfuseplotmark{square*}};
\node(dc9) at (9,3.624/60){\pgfuseplotmark{square*}};
\node(dc10) at (10,3.94/60){\pgfuseplotmark{square*}};
\node(dc11) at (11,5.362/60){\pgfuseplotmark{square*}};
\node(dc12) at (12,5.193/60){\pgfuseplotmark{square*}};
\node(dc13) at (13,6.126/60){\pgfuseplotmark{square*}};
\node(dc14) at (14,7.82/60){\pgfuseplotmark{square*}};
\node(dc15) at (15,8.329/60){\pgfuseplotmark{square*}};
\node(dc16) at (16,11.611/60){\pgfuseplotmark{square*}};
\node(dc17) at (17,11.067/60){\pgfuseplotmark{square*}};
\node(dc18) at (18,12.242/60){\pgfuseplotmark{square*}};
\node(dc19) at (19,13.37/60){\pgfuseplotmark{square*}};
\node(dc20) at (20,14.161/60){\pgfuseplotmark{square*}};
\node(dc21) at (21,19.702/60){\pgfuseplotmark{square*}};
\node(dc22) at (22,20.651/60){\pgfuseplotmark{square*}};
\node(dc23) at (23,21.907/60){\pgfuseplotmark{square*}};
\node(dc24) at (24,24.139/60){\pgfuseplotmark{square*}};
\node(dc30) at (30,54.669/60){\pgfuseplotmark{square*}};
\node(dc40) at (40,140.542/60){\pgfuseplotmark{square*}};
\node(dc50) at (50,322.187/60){\pgfuseplotmark{square*}};
\node(dc60) at (60,1070.43/60){\pgfuseplotmark{square*}};
\node(dc70) at (70,1773.54/60){\pgfuseplotmark{square*}};
\node(dc80) at (80,4921.07/60){\pgfuseplotmark{square*}};
\node(dc90) at (90,5693.0/60){\pgfuseplotmark{square*}};
\node(dc100) at (100,8105.07/60){\pgfuseplotmark{square*}};
\node(cc2) at (2,1.23/60){\pgfuseplotmark{triangle*}};
\node(cc3) at (3,1.47/60){\pgfuseplotmark{triangle*}};
\node(cc4) at (4,1.738/60){\pgfuseplotmark{triangle*}};
\node(cc5) at (5,2.14/60){\pgfuseplotmark{triangle*}};
\node(cc6) at (6,2.346/60){\pgfuseplotmark{triangle*}};
\node(cc7) at (7,2.849/60){\pgfuseplotmark{triangle*}};
\node(cc8) at (8,3.141/60){\pgfuseplotmark{triangle*}};
\node(cc9) at (9,3.648/60){\pgfuseplotmark{triangle*}};
\node(cc10) at (10,4.07/60){\pgfuseplotmark{triangle*}};
\node(cc11) at (11,4.593/60){\pgfuseplotmark{triangle*}};
\node(cc12) at (12,5.587/60){\pgfuseplotmark{triangle*}};
\node(cc13) at (13,6.173/60){\pgfuseplotmark{triangle*}};
\node(cc14) at (14,7.112/60){\pgfuseplotmark{triangle*}};
\node(cc15) at (15,8.481/60){\pgfuseplotmark{triangle*}};
\node(cc16) at (16,8.425/60){\pgfuseplotmark{triangle*}};
\node(cc17) at (17,12.331/60){\pgfuseplotmark{triangle*}};
\node(cc18) at (18,13.172/60){\pgfuseplotmark{triangle*}};
\node(cc19) at (19,14.512/60){\pgfuseplotmark{triangle*}};
\node(cc20) at (20,14.915/60){\pgfuseplotmark{triangle*}};
\node(cc21) at (21,15.4/60){\pgfuseplotmark{triangle*}};
\node(cc22) at (22,17.847/60){\pgfuseplotmark{triangle*}};
\node(cc23) at (23,21.354/60){\pgfuseplotmark{triangle*}};
\node(cc24) at (24,23.712/60){\pgfuseplotmark{triangle*}};
\node(cc30) at (30,42.83/60){\pgfuseplotmark{triangle*}};
\node(cc40) at (40,104.94/60){\pgfuseplotmark{triangle*}};
\node(cc50) at (50,223.627/60){\pgfuseplotmark{triangle*}};
\node(cc60) at (60,475.838/60){\pgfuseplotmark{triangle*}};
\node(cc70) at (70,796.481/60){\pgfuseplotmark{triangle*}};
\node(cc80) at (80,1670.31/60){\pgfuseplotmark{triangle*}};
\node(cc90) at (90,2644.89/60){\pgfuseplotmark{triangle*}};
\node(cc100) at (100,3696.06/60){\pgfuseplotmark{triangle*}};
\end{tikzpicture}\vspace*{-2mm}
\caption{CPUtime for articulated families.}
\label{cpuarticul.fig}\vspace*{-2mm}
\end{figure}

\eject
\noindent family, square dots for the daisy-shape family, and triangle dots for the caterpillar-shape family.
It may be observed that the three families behave differently,
while the systems have the same size for the same number of components.
In particular, the star-shape family takes much more time, and beyond 50 components, APT crashed for memory exhaustion.
From 50 components, the daisy-shape family takes more time than the caterpillar-shape one.
Anyway, from $n=30$ components, the performance is worse than synthesising $n$ separate components
(and recombining them, indicated by the line $t_n$; note however that if we use a computer system with several
CPUs and/or cores, it is possible to launch several individual component syntheses in parallel, reducing seriously that indicator).
Since the plot is crushed for the small numbers of components, we present in Figure~\ref{zoom-cpu-articul.fig}
a zoom on the first results (here the time is given in seconds, instead of minutes).
Curiously, up to 10 components (221 states) it is more effective to synthesise the articulated system than to
desarticulate it and solve the components separately; and for the daisy and caterpillar families,
this is still true up to 20 components (441 states).

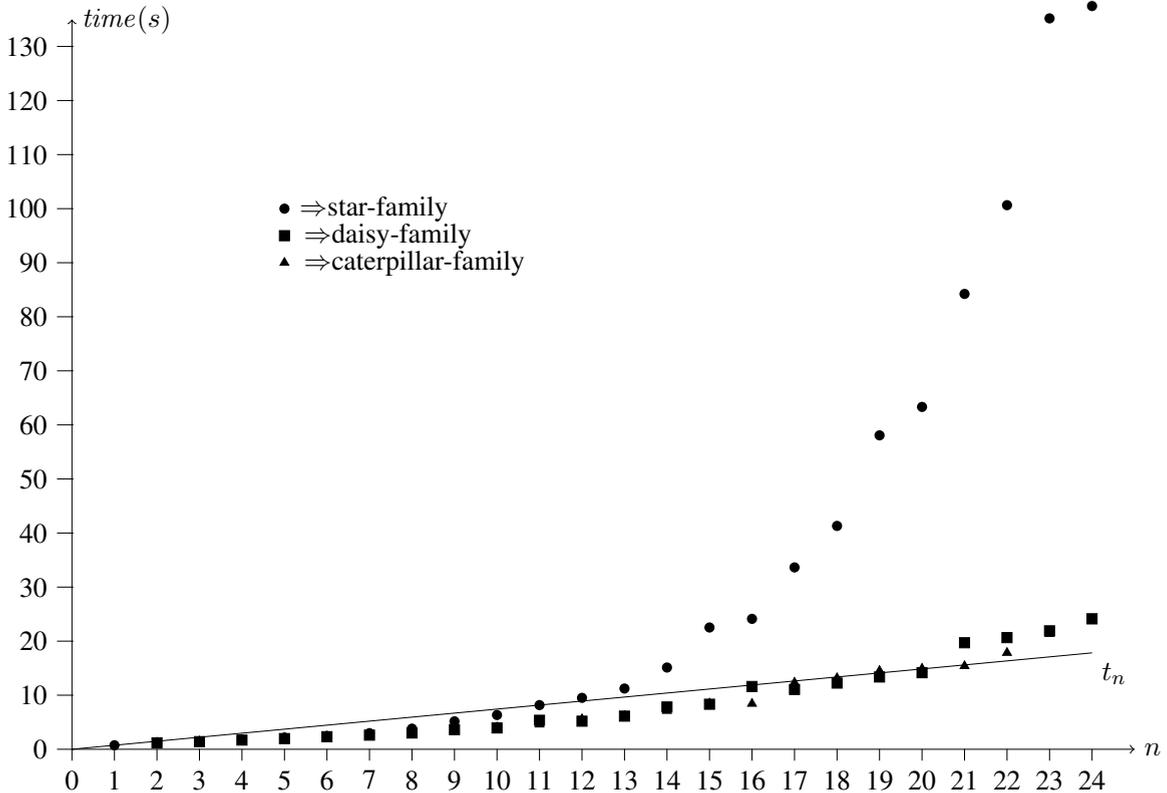
\begin{figure}[ht]
\centering
\scalebox{0.95}{
\begin{tikzpicture}[yscale=0.075, xscale=0.59]
        \draw[->] (0,0) -- (25,0) node[right] {$n$};
        \draw[->] (0,0) -- (0,135) node[right] {$\mathit{time}(s)$};
        \draw[-] (0,0) -- (24,24*0.743) node[below right] {$t_n$};
   	\foreach \x in {0,1,...,24}
     		\draw (\x,1pt) -- (\x,-60pt)
			node[anchor=north] {\x};
    	\foreach \y in {0,10,...,135}
     		\draw (1pt,\y) -- (-10pt,\y)
     			node[anchor=east] {\y};
\node[circle,fill=black!100,inner sep=0.05cm](sc) at (5,100)[label=right:$\impl$star-family]{};
\node(dc) at (5,95){\pgfuseplotmark{square*}};\node at (5,95)[label=right:$\impl$daisy-family]{};
\node(cc) at (5,90){\pgfuseplotmark{triangle*}};\node at (5,90)[label=right:$\impl$caterpillar-family]{};
\node[circle,fill=black!100,inner sep=0.05cm](sc1) at (1,0.743)[]{};
\node[circle,fill=black!100,inner sep=0.05cm](sc2) at (2,1.131)[]{};
\node[circle,fill=black!100,inner sep=0.05cm](sc3) at (3,1.378)[]{};
\node[circle,fill=black!100,inner sep=0.05cm](sc4) at (4,1.752)[]{};
\node[circle,fill=black!100,inner sep=0.05cm](sc5) at (5,2.141)[]{};
\node[circle,fill=black!100,inner sep=0.05cm](sc6) at (6,2.442)[]{};
\node[circle,fill=black!100,inner sep=0.05cm](sc7) at (7,2.973)[]{};
\node[circle,fill=black!100,inner sep=0.05cm](sc8) at (8,3.817)[]{};
\node[circle,fill=black!100,inner sep=0.05cm](sc9) at (9,5.167)[]{};
\node[circle,fill=black!100,inner sep=0.05cm](sc10) at (10,6.335)[]{};
\node[circle,fill=black!100,inner sep=0.05cm](sc11) at (11,8.171)[]{};
\node[circle,fill=black!100,inner sep=0.05cm](sc12) at (12,9.512)[]{};
\node[circle,fill=black!100,inner sep=0.05cm](sc13) at (13,11.235)[]{};
\node[circle,fill=black!100,inner sep=0.05cm](sc14) at (14,15.121)[]{};
\node[circle,fill=black!100,inner sep=0.05cm](sc15) at (15,22.525)[]{};
\node[circle,fill=black!100,inner sep=0.05cm](sc16) at (16,24.124)[]{};
\node[circle,fill=black!100,inner sep=0.05cm](sc17) at (17,33.639)[]{};
\node[circle,fill=black!100,inner sep=0.05cm](sc18) at (18,41.312)[]{};
\node[circle,fill=black!100,inner sep=0.05cm](sc19) at (19,58.068)[]{};
\node[circle,fill=black!100,inner sep=0.05cm](sc20) at (20,63.325)[]{};
\node[circle,fill=black!100,inner sep=0.05cm](sc21) at (21,84.225)[]{};
\node[circle,fill=black!100,inner sep=0.05cm](sc22) at (22,100.641)[]{};
\node[circle,fill=black!100,inner sep=0.05cm](sc23) at (23,135.198)[]{};
\node[circle,fill=black!100,inner sep=0.05cm](sc24) at (24,137.478)[]{};
\node(dc2) at (2,1.185){\pgfuseplotmark{square*}};
\node(dc3) at (3,1.412){\pgfuseplotmark{square*}};
\node(dc4) at (4,1.699){\pgfuseplotmark{square*}};
\node(dc5) at (5,1.974){\pgfuseplotmark{square*}};
\node(dc6) at (6,2.327){\pgfuseplotmark{square*}};
\node(dc7) at (7,2.609){\pgfuseplotmark{square*}};
\node(dc8) at (8,3.004){\pgfuseplotmark{square*}};
\node(dc9) at (9,3.624){\pgfuseplotmark{square*}};
\node(dc10) at (10,3.94){\pgfuseplotmark{square*}};
\node(dc11) at (11,5.362){\pgfuseplotmark{square*}};
\node(dc12) at (12,5.193){\pgfuseplotmark{square*}};
\node(dc13) at (13,6.126){\pgfuseplotmark{square*}};
\node(dc14) at (14,7.82){\pgfuseplotmark{square*}};
\node(dc15) at (15,8.329){\pgfuseplotmark{square*}};
\node(dc16) at (16,11.611){\pgfuseplotmark{square*}};
\node(dc17) at (17,11.067){\pgfuseplotmark{square*}};
\node(dc18) at (18,12.242){\pgfuseplotmark{square*}};
\node(dc19) at (19,13.37){\pgfuseplotmark{square*}};
\node(dc20) at (20,14.161){\pgfuseplotmark{square*}};
\node(dc21) at (21,19.702){\pgfuseplotmark{square*}};
\node(dc22) at (22,20.651){\pgfuseplotmark{square*}};
\node(dc23) at (23,21.907){\pgfuseplotmark{square*}};
\node(dc24) at (24,24.139){\pgfuseplotmark{square*}};
\node(cc2) at (2,1.23){\pgfuseplotmark{triangle*}};
\node(cc3) at (3,1.47){\pgfuseplotmark{triangle*}};
\node(cc4) at (4,1.738){\pgfuseplotmark{triangle*}};
\node(cc5) at (5,2.14){\pgfuseplotmark{triangle*}};
\node(cc6) at (6,2.346){\pgfuseplotmark{triangle*}};
\node(cc7) at (7,2.849){\pgfuseplotmark{triangle*}};
\node(cc8) at (8,3.141){\pgfuseplotmark{triangle*}};
\node(cc9) at (9,3.648){\pgfuseplotmark{triangle*}};
\node(cc10) at (10,4.07){\pgfuseplotmark{triangle*}};
\node(cc11) at (11,4.593){\pgfuseplotmark{triangle*}};
\node(cc12) at (12,5.587){\pgfuseplotmark{triangle*}};
\node(cc13) at (13,6.173){\pgfuseplotmark{triangle*}};
\node(cc14) at (14,7.112){\pgfuseplotmark{triangle*}};
\node(cc15) at (15,8.481){\pgfuseplotmark{triangle*}};
\node(cc16) at (16,8.425){\pgfuseplotmark{triangle*}};
\node(cc17) at (17,12.331){\pgfuseplotmark{triangle*}};
\node(cc18) at (18,13.172){\pgfuseplotmark{triangle*}};
\node(cc19) at (19,14.512){\pgfuseplotmark{triangle*}};
\node(cc20) at (20,14.915){\pgfuseplotmark{triangle*}};
\node(cc21) at (21,15.4){\pgfuseplotmark{triangle*}};
\node(cc22) at (22,17.847){\pgfuseplotmark{triangle*}};
\node(cc23) at (23,21.354){\pgfuseplotmark{triangle*}};
\node(cc24) at (24,23.712){\pgfuseplotmark{triangle*}};
\end{tikzpicture} }\vspace*{-1mm}
\caption{Zoom on CPUtime for articulated families.}
\label{zoom-cpu-articul.fig}\vspace*{-2mm}
\end{figure}

\medskip
We also performed a regression analysis to determine how the CPU time of the PN-synthesis grows
in term of the number of components in an articulation.
As already observed, the results are different for the different families of transition systems we considered,
and it seems difficult to find growth rules compatible with all the figures we obtained.
We then considered several subranges for the number of components.\\
First, up to 10 components, the three families behave about the same.
The best fit corresponds to an exponential growth of the kind $0.9\cdot 1.01^{|S|}$, thus with a base very close to $1$;
a power regression also gives a rather good fit with a formula  of the kind $0.07\cdot x^{0.77}$, lower than linear.\\
From $n=10$ up to $50$, the star-shape family becomes worse than the other two, which remain very similar.
For the first one,  the power regression gives rather good results, with  a degree of $4.25$,
while for the next two families we get a lower degree around $2.6$.\\
Beyond $n=50$, the star-family disappears (as explained before) and the next two behave differently
(with a strange, reproducible, hop for the daisy-flower family around 80 components,
maybe due to a change in the usage of the cache memory) .
For the daisy-family, we get a good  power regression with a degree of $4.6$;
and for the caterpillar-family the power is $4.1$.
We thus have a polynomial-like behaviour by chunks, with an increasing power of $|S|$.

For the factorisation decomposition, where the size of the products grows very rapidly
($|S|^n$ states for a product of $n$ components of the same size $|S|$), the performance is rapidly better
when applying synthesis to the various components, as detailed in~\cite{fact18}.

The previous analysis assumed that we consider synthesis problems where the target is the class of bounded weighted Petri nets.
Thus all the nets we consider are bounded, but we do not fix these bounds beforehand.
If we do, i.e., if we search for safe of $k$-safe solutions (with fixed $k$), then suddenly the problem becomes NP-complete (in the worst case, see for instance~\cite{BBD97,Tredup19}) instead of polynomial\footnote{
note that in this case we have to add constraints of the kind $\forall s\in S:M_s(p)\leq k$
to solve the various separation sub-problems, rendering the linear systems non-homogeneous;
 hence we may not rely to a rational solution to derive one in the integer domain.}.
 Corollaries~\ref{safep.cor} and \ref{safea.cor} show that factorisations and articulations may be used to solve
 $k$-safe synthesis problems.
 That means that the synthesis algorithms we know for generating safe or $k$-safe solutions are exponential in the worst cases.
 If it occurs that P=NP, that means that we will be able to derive polynomial algorithms, but probably with a high polynomial degree.
 Then the gain obtained by a divide and conquer strategy (when it works) may be much larger than the ones we mention above.

\section{Concluding remarks}
\label{concl.sct}

We have developed a theory, algorithms  and experiments around two (families of) pairs of operators acting on
labelled transition systems and Petri nets.
This allows both for structuring large transition systems and improving the efficiency of Petri net syntheses.

We may of course wonder if it will often be the case that transition systems needing a PN-synthesis present such structures.
In fact, it depends where these systems come from. For instance, if they come from the organiser of a tool competition,
needing a set of large positive examples: it may happen that they were obtained by combining smaller components
(especially if the organiser does not know the decomposition algorithms developed in the present paper).

Other possible issues are to examine how this may be specialised for some subclasses of Petri nets
and how these structures behave in the context of  approximate solutions, devised when an exact synthesis is not possible,
in the spirit of the notions and procedures developed in \cite{US-Lata18}.

Finally, other kinds of operator pairs could be searched for,  having interesting decomposition and recomposition procedures,
allowing again to speed up synthesis problems.

\subsection*{Acknowledgements}
The author thanks Eike Best, Uli Schlachter,
as well as  the anonymous referees of the conference papers from which the present work is issued,
for their useful remarks and suggestions.
The remarks of anonymous referees also allowed to improve the present paper.


\begin{thebibliography}{10}
\providecommand{\url}[1]{\texttt{#1}}
\providecommand{\urlprefix}{URL }
\expandafter\ifx\csname urlstyle\endcsname\relax
  \providecommand{\doi}[1]{doi:\discretionary{}{}{}#1}\else
  \providecommand{\doi}{doi:\discretionary{}{}{}\begingroup
  \urlstyle{rm}\Url}\fi
\providecommand{\eprint}[2][]{\url{#2}}

\bibitem{bbd}
Badouel E, Bernardinello L, Darondeau P.
\newblock {P}etri Net Synthesis.
\newblock Texts in Theoretical Computer Science. An {EATCS} Series.
  Springer-Verlag, 2015.
\newblock ISBN:978-3-662-47967-4.
\newblock  doi:10.1007/978-3-662-47967-4.

\bibitem{besdev-lata}
Best E, Devillers R.
\newblock Characterisation of the State Spaces of Live and Bounded Marked Graph
  {{P}etri} Nets.
\newblock In: 8th International Conference on Language and Automata Theory and
  Applications ({LATA} 2014). 2014 pp. 161--172.
\newblock doi:10.1007/978-3-319-04921-2\_13.

\bibitem{EB-US-15}
Best E, Schlachter U.
\newblock Analysis of {P}etri Nets and Transition Systems.
\newblock In: Proceedings 8th Interaction and Concurrency Experience, {ICE}
  2015, Grenoble, France, 4-5th June 2015. 2015 pp. 53--67.
\newblock \doi{10.4204/EPTCS.189.6}.

\bibitem{BDS-acta17}
Best E, Devillers R, Schlachter U.
\newblock Bounded choice-free {P}etri net synthesis: algorithmic issues.
\newblock \emph{Acta Inf.}, 2018.
\newblock \textbf{55}(7):575--611.
\newblock \doi{10.1007/s00236-017-0310-9}.

\bibitem{US-Lata18}
Schlachter U.
\newblock Over-Approximative {P}etri Net Synthesis for Restricted Subclasses of
  Nets.
\newblock In: Language and Automata Theory and Applications - 12th
  International Conference, {LATA} 2018, Ramat Gan, Israel, April 9-11, 2018,
  Proceedings. 2018 pp. 296--307.
\newblock \doi{10.1007/978-3-319-77313-1_23}.

\bibitem{BBD95}
Badouel E, Bernardinello L, Darondeau P.
\newblock Polynomial Algorithms for the Synthesis of Bounded Nets.
\newblock In: TAPSOFT'95: Theory and Practice of Software Development, 6th
  International Joint Conference CAAP/FASE, Aarhus, Denmark. 1995 pp. 364--378.
\newblock \doi{10.1007/3-540-59293-8_207}.

\bibitem{BBD97}
Badouel E, Bernardinello L, Darondeau P.
\newblock The Synthesis Problem for Elementary Net Systems is {NP}-Complete.
\newblock \emph{Theor. Comput. Sci.}, 1997.
\newblock \textbf{186}(1-2):107--134.
\newblock \doi{10.1016/S0304-3975(96)00219-8}.

\bibitem{devacta17}
Devillers R.
\newblock Factorisation of transition systems.
\newblock \emph{Acta Informatica}, 2018.
\newblock \textbf{55}(4):339--362.
\newblock \doi{10.1007/s00236-017-0300-y}.

\bibitem{fact18}
Devillers R, Schlachter U.
\newblock Factorisation of {{P}etri} Net Solvable Transition Systems.
\newblock In: Application and Theory of {P}etri Nets and Concurrency - 39th
  International Conference, {PETRI} {NETS} 2018, Bratislava, Slovakia. 2018 pp.
  82--98.
\newblock \doi{10.1007/978-3-319-91268-4_5}.

\bibitem{RD-articul-PN}
Devillers R.
\newblock Articulation of Transition Systems and its Application to {P}etri Net
  Synthesis.
\newblock In: Application and Theory of {P}etri Nets and Concurrency - 40th
  International Conference, {PETRI} {NETS} 2019, Aachen, Germany, June 23-28,
  2019, Proceedings. 2019 pp. 113--126.
\newblock \doi{10.1007/978-3-030-21571-2_8}.

\bibitem{arnold94}
Arnold A.
\newblock Finite Transition Systems - Semantics of Communicating Systems.
\newblock Prentice Hall international series in computer science. Prentice
  Hall, 1994.
\newblock ISBN:978-0-13-092990-7.

\bibitem{DR-synth-96}
Desel J, Reisig W.
\newblock The Synthesis Problem of {P}etri Nets.
\newblock \emph{Acta Inf.}, 1996.
\newblock \textbf{33}(4):297--315.
\newblock \doi{10.1007/s002360050046}.

\bibitem{dev-ACSD16}
Devillers R.
\newblock {Products of Transition Systems and Additions of {{P}etri} Nets}.
\newblock In: Proc. 16th International Conference on Application of Concurrency
  to System Design ({ACSD} 2016) J. Desel and A. Yakovlev (eds). 2016 pp.
  65--73.
\newblock doi:10.1109/ACSD.2016.10.

\bibitem{keller}
Keller RM.
\newblock {A Fundamental Theorem of Asynchronous Parallel Computation}.
\newblock In: Sagamore Computer Conference, August 20-23 1974, LNCS Vol. 24.
  1975 pp. 102--112.
\newblock doi:10.1007/3-540-07135-0\_113.

\bibitem{HopTar73}
Hopcroft JE, Tarjan RE.
\newblock Efficient Algorithms for Graph Manipulation {[H]} (Algorithm 447).
\newblock \emph{Commun. {ACM}}, 1973.
\newblock \textbf{16}(6):372--378.
\newblock \doi{10.1145/362248.362272}.

\bibitem{WestTar92}
Westbrook J, Tarjan RE.
\newblock Maintaining Bridge-Connected and Biconnected Components On-Line.
\newblock \emph{Algorithmica}, 1992.
\newblock \textbf{7}(5{\&}6):433--464.
\newblock \doi{10.1007/BF01758773}.

\bibitem{BDK02}
Best E, Devillers R, Koutny M.
\newblock The {B}ox {A}lgebra = {P}etri {N}ets + {P}rocess {E}xpressions.
\newblock \emph{Inf. Comput.}, 2002.
\newblock \textbf{178}(1):44--100.
\newblock \doi{10.1006/inco.2002.3117}.

\bibitem{BestDE20}
Best E, Devillers RR, Erofeev E.
\newblock A New Property of Choice-Free {P}etri Net Systems.
\newblock In: Application and Theory of {P}etri Nets and Concurrency - 41st
  International Conference, {PETRI} {NETS} 2020, Paris, France, June 24-25,
  2020, Proceedings. 2020 pp. 89--108.
\newblock \doi{10.1007/978-3-030-51831-8_5}.

\bibitem{karmarkar}
Karmarkar N.
\newblock A new polynomial-time algorithm for linear programming.
\newblock \emph{Combinatorica}, 1984.
\newblock \textbf{4}(4):373--396.
\newblock \doi{10.1007/BF02579150}.

\bibitem{dantzig51}
Dantzig G.
\newblock Maximization of a linear function of variables subject to linear
  inequalities.
\newblock In: Koopmans T (ed.), Activity Analysis of Production and Allocation,
  Proceedings. Wiley, New York, 1951 p. 339--347.

\bibitem{klee72}
Klee V, Minty G.
\newblock How Good Is the Simplex Algorithm?
\newblock In: Shisha O (ed.), Inequalities III, Proceedings. Academic Press,
  New York, 1951 p. 159--175.

\bibitem{smtinterpol}
SMTInterpol, an Interpolating SMT Solver.
\newblock https://ultimate.informatik.uni-freiburg.de/smtinterpol/.

\bibitem{kroen10}
Kroening D, Leroux J, R{\"{u}}mmer P.
\newblock Interpolating Quantifier-Free Presburger Arithmetic.
\newblock In: Logic for Programming, Artificial Intelligence, and Reasoning -
  17th International Conference, LPAR-17, Yogyakarta, Indonesia, October 10-15,
  2010. Proceedings. 2010 pp. 489--503.
\newblock \doi{10.1007/978-3-642-16242-8_35}.

\bibitem{Tredup19}
Tredup R.
\newblock Hardness Results for the Synthesis of b-bounded {P}etri Nets.
\newblock In: Application and Theory of {P}etri Nets and Concurrency - 40th
  International Conference, {PETRI} {NETS} 2019, Aachen, Germany, June 23-28,
  2019, Proceedings. 2019 pp. 127--147.
\newblock \doi{10.1007/978-3-030-21571-2_9}.
\end{thebibliography}
\end{document}